\documentclass[10pt,twoside,a4paper]{article}

\usepackage[utf8]{inputenc}
\usepackage[T1]{fontenc}
\usepackage{amsmath,amssymb,dsfont}
\numberwithin{equation}{section}
\usepackage{microtype}
\usepackage{graphicx,tikz,pgfplots}
\usepackage{tikz-3dplot}
\usetikzlibrary{arrows,shapes,fadings,intersections,decorations.pathreplacing,positioning}
\graphicspath{{images/}}
\pgfplotsset{compat=newest}
\usepackage[hyperref,amsmath,thmmarks]{ntheorem}
\usepackage{aliascnt}
\usepackage[a4paper,centering,bindingoffset=0cm,marginpar=2cm,margin=2.5cm]{geometry}
\usepackage[pagestyles]{titlesec}
\usepackage[font=footnotesize,format=plain,labelfont=sc,textfont=sl,width=0.75\textwidth,labelsep=period]{caption}
\usepackage{comment}
\usepackage{siunitx} 

\usetikzlibrary{quotes}
\usetikzlibrary{angles}
\usetikzlibrary{babel}

\usetikzlibrary{decorations.pathreplacing}
\usetikzlibrary{decorations.markings}

\usepackage{mathabx}
\tikzset{
    >=stealth',
    punkt/.style={
           rectangle,
           rounded corners,
           draw=black, very thick,
           text width=6.5em,
           minimum height=2em,
           text centered},
    pil/.style={
           ->,
           thick,
           shorten <=2pt,
           shorten >=2pt,}
}
\tikzstyle{block} = [rectangle, rounded corners, minimum width= 3cm, minimum height=1cm, text centered, draw=black, fill=blue!20,]

\tikzstyle{invisbleblock} = [minimum width= 3em, minimum height=1cm, text centered, ]

\tikzstyle{decision} = [diamond, minimum width=3cm, minimum height=1cm, text centered, draw=black, fill=orange!20]

\tikzstyle{arrow} = [thick,->,>=stealth]
\tikzstyle{line} = [thick,-]

\usepackage{dsfont,bbm}

\usepackage[boxed]{algorithm2e}

\usepackage{enumerate}

\title{Inverse Problems of Single Molecule Localization Microscopy}
\author{Montse Lopez-Martinez$^1$ {\footnotesize\href{mailto:lopez-martinez@iap.tuwien.ac.at}{lopez-martinez@iap.tuwien.ac.at}}\\
        Gwenael Mercier$^3$       {\footnotesize\href{mailto:gwenael.mercier@univie.ac.at}{gwenael.mercier@univie.ac.at}}\\
        Kamran Sadiq$^{2}$        {\footnotesize\href{mailto:kamran.sadiq@ricam.oeaw.ac.at}{kamran.sadiq@ricam.oeaw.ac.at}}\\
        Otmar Scherzer$^{2,3}$    {\footnotesize\href{mailto:otmar.scherzer@univie.ac.at}{otmar.scherzer@univie.ac.at}}\\
        Magdalena Schneider$^1$   {\footnotesize\href{mailto:schneider@iap.tuwien.ac.at}{schneider@iap.tuwien.ac.at}}\\
        John C Schotland$^4$        {\footnotesize\href{mailto:schotland@umich.edu}{schotland@umich.edu}}\\
        Gerhard J. Schütz$^1$        {\footnotesize\href{mailto:schuetz@iap.tuwien.ac.at}{schuetz@iap.tuwien.ac.at}}\\
        Roger Telschow$^3$        {\footnotesize\href{mailto:roger.telschow@univie.ac.at}{roger.telschow@univie.ac.at}}\\
}
\date{\today}

\usepackage[pdftex,colorlinks=true,linkcolor=blue,citecolor=green,urlcolor=blue,bookmarks=true,bookmarksnumbered=true]{hyperref}
\hypersetup
{
    pdfauthor={Authors},
    pdfsubject={Subject},
    pdftitle={Title},
    pdfkeywords={Keywords}
}

\newtheorem{lemma}{Lemma}[section]

\newaliascnt{proposition}{lemma}

\aliascntresetthe{proposition}

\newaliascnt{corollary}{lemma}

\aliascntresetthe{corollary}

\newaliascnt{theorem}{lemma}
\newtheorem{theorem}[theorem]{Theorem}
\aliascntresetthe{theorem}

\theorembodyfont{\normalfont}
\newaliascnt{definition}{lemma}
\newtheorem{definition}[definition]{Definition}
\aliascntresetthe{definition}

\newaliascnt{assumption}{lemma}
\newtheorem{assumption}[assumption]{Assumption}
\aliascntresetthe{assumption}

\newaliascnt{notation}{lemma}

\aliascntresetthe{notation}

\newaliascnt{example}{lemma}

\aliascntresetthe{example}

\newaliascnt{experiment}{lemma}
\newtheorem{experiment}[experiment]{Experiment}
\aliascntresetthe{experiment}

\theoremstyle{nonumberplain}
\theoremseparator{:}
\theoremheaderfont{\normalfont\itshape}

\newtheorem{remark}{Remark}

\theoremsymbol{\ensuremath{\square}}
\newtheorem{proof}{Proof}

\newcommand{\N}{\mathds{N}}
\newcommand{\R}{\mathds{R}}
\newcommand{\C}{\mathds{C}}

\newcommand{\abs}[1]{\left|#1\right|}
\newcommand{\norm}[1]{\left\|#1\right\|}
\newcommand{\set}[1]{\left\{#1\right\}}
\newcommand{\inner}[2]{\left<#1,#2\right>}

\newcommand{\e}{\mathrm e}

\renewcommand{\i}{\mathrm i}

\newcommand{\ve}{\varepsilon}

\newcommand{\focus}{{\texttt{f}}}
\newcommand{\refractive}{{\texttt{n}}}

\newcommand{\tensortwo}{\mathcal{T}}

\newcommand{\sphere}{\mathbb{S}^2}

\newcommand{\dd}{\, \mathrm{d} }

\newcommand{\br}{{\bf r}}
\newcommand{\bx}{{\bf x}}
\newcommand{\by}{{\bf y}}
\newcommand{\bk}{{\bf k}}
\newcommand{\bu}{{\bf u}}
\newcommand{\bv}{{\bf v}}
\newcommand{\be}{{\bf e}}
\newcommand{\bE}{{\bf E}}
\newcommand{\bpsi}{{\bf \Psi}}
\newcommand{\bphi}{{\bf \Phi}}
\newcommand{\brk}{\tilde{\bf k}}

\newcommand{\fobj}{\texttt{f}_{\texttt{obj}}}

\newcommand{\flens}{\texttt{f}_{\texttt{L}}}

\newcommand{\psf}{{\text{PSF}}}

\newcommand{\fouriercurrent}{\widehat{\bf J}}

\newcommand{\eindelta}{\tilde \delta}

\newcommand{\scal}[2]{\left \langle #1,#2 \right \rangle}
\newcommand{\ls}{\leqslant}
\newcommand{\gs}{\geqslant}

\newcommand{\rpsi}{r_3^{\! \bf \Psi}}

\newcommand{\vr}{\begin{pmatrix} r_1  \\ r_2 \\ r_3 \end{pmatrix}}
\newcommand{\vp}{\begin{pmatrix}  \Psi_1 \\ \Psi_2 \\ \Psi_3 \end{pmatrix}}

\newcommand{\obj}{{\texttt{obj}}} 

\newcommand{\ttd}{{\texttt{d}}}
\newcommand{\bfp}{{\texttt{bfp}}} 

\newcommand{\rbfp}{r_3^{\texttt{b}}}
\newcommand{\NA}{\mathtt{NA}}
\renewcommand{\L}{\mathtt{L}}

\newcommand{\tmax}{{\theta_{\mathrm{max}}}}
\DeclareMathOperator{\Grad}{\nabla_{\br}}
\DeclareMathOperator{\Div}{\nabla_{\br} \cdot}
\DeclareMathOperator{\Curl}{\nabla_{\br} \times}

\newcommand{\sspace}[1]{\mathcal{S}(#1)}

\newcommand{\tempdist}[1]{\mathcal{S'}(#1)}

\begin{document}

\maketitle
\thispagestyle{empty}
\begin{center}	
\hspace*{8em}
\parbox[t]{12em}{\footnotesize
\hspace*{-1ex}$^1$Institute of Applied Physics\\
TU-Wien\\
Getreidemarkt 9\\
A-1060 Vienna, Austria}
\hfil
\parbox[t]{17em}{\footnotesize
\hspace*{-1ex}$^2$Johann Radon Institute for Computational\\
\hspace*{1em}and Applied Mathematics (RICAM)\\
Altenbergerstraße 69\\
A-4040 Linz, Austria}
\\[4ex]
\hspace*{8em}
\parbox[t]{12em}{\footnotesize
\hspace*{-1ex}$^3$Faculty of Mathematics\\
University of Vienna\\
Oskar-Morgenstern-Platz 1\\
A-1090 Vienna, Austria}
\hfil   
\parbox[t]{17em}{\footnotesize
\hspace*{-1ex}$^4$Department of Mathematics \\
\hspace*{1em}and Department of Physics\\
University of Michigan\\
Ann Arbor, MI 48109, USA}
\end{center}

\begin{abstract}
Single molecule localization microscopy is a recently developed superresolution imaging technique to 
visualize structural properties of single cells. 
The basic principle consists in chemically attaching fluorescent dyes to the molecules, 
which after excitation with a strong laser may emit light. To achieve superresolution, signals of individual fluorophores are separated in time.
In this paper we follow the physical and chemical literature and derive mathematical models describing the 
propagation of light emitted from dyes 
in single molecule localization microscopy experiments via Maxwell's equations. This forms the basis of formulating inverse problems 
related to single molecule localization microscopy. We also show that the current status of reconstruction methods is a simplification 
of more general inverse problems for Maxwell's equations as discussed here.
\end{abstract}
\section{Introduction}
The structure and organization of proteins in cells relate directly to their biological function. Many proteins associate with each other and form functional supramolecular arrangements known as oligomers. Protein oligomers appear in a wide range of crucial biological processes, such as signal transduction, ion transport or immune reactions. The accurate characterization of the supramolecular organization of proteins, including oligomer stoichiometry and its spatial distribution, is fundamental to fully understand these biological processes. 

Several tools address the study of the structure of small biological units, most popular ones being x-ray crystallography and, most recently, cryo-electron microscopy, which have been used to characterize the structure of individual isolated proteins with a high level of detail \cite{Shi14,OrlMyaAndNatKha17}. However, currently these tools cannot be applied for studying quaternary protein assemblies in their native cellular environment, due to a lack in chemical contrast: it is impossible to single out the molecular structures of interest within the plethora of other molecular species. A solution is provided by fluorescence microscopy, where a single protein species is addressed by specific fluorescence labelling directly in the cell. While fluorescence microscopy allows for imaging these labelled structures at a high signal to noise ratio, its resolution is limited to around \SI{200}{\nano\meter} due to the diffraction of light. This prohibits a characterization of oligomeric arrangements with conventional light microscopy, since these structures are smaller than the resolution limit. In summary, the current life sciences are limited by a resolution gap, the upper limit of which is set by the diffraction limit of fluorescence microscopy, the lower limit by the difficulty to interpret crystallography experiments of oligomeric protein complexes.    

In principle, the arrival of superresolution microscopy techniques allows to overcome this gap. Virtually all superresolution techniques are based on fluorescence microscopy, and as such have to overcome or circumvent the problem of optical diffraction. A fluorescent label emits light that is imaged by the microscopy system as a blurry dot. This dot of diffracted light is known as the point spread function (PSF). Its size $d$ (the diameter of the 
essential support of the PSF) is determined by the light wavelength $\lambda$ and by the numerical aperture ($\NA$) of the objective. The angle $\tmax$ is one half of the angular aperture (A). 
Neglecting lens aberrations, it can be described analytically by an Airy function, where the distance between the maximum and its first minimum is given by (see \autoref{eq:Fourier_pupil_polar})

\begin{equation}\label{eq:diffraction_limit}
	d=\frac{\lambda}{2 \refractive \sin(\tmax)} = \frac{\lambda}{2\,\NA} \text{ with } \refractive \text{ the refractive index of the medium}.
\end{equation}

The size of the PSF determines the limit of resolution of conventional light microscopy. It was first described by Abbe \cite{Abb73}, and it is known as Abbe’s limit of diffraction: If two fluorescent labels are closer than the distance d, their PSFs overlap, and they cannot be distinguished from each other. For fluorescence microscopy, with wavelengths in the visible spectrum and objectives with numerical apertures generally lower than $1.3$, this resolution limit is in the order of \SI{200}{\nano\meter}.

For oligomeric protein structures, the distance between their subunits is typically in the range of a few nanometers, far smaller than the diffraction limit. The signals from the individual subunits overlap, and cannot be resolved by conventional light microscopy. A prominent example for such a structure is the nuclear pore complex (NPC), which is a large protein complex located in the nuclear membrane of eukaryotic cells. Its structure is well characterized through electron microscopy \cite{AppKosSpaOriGui15}. NPCs are composed of around $30$ proteins arranged in an $8$-fold symmetry forming a pore that regulates the transport across the nuclear membrane. The overall size ranges approximately between $80$ to \SI{120}{\nano\meter} depending on the species \cite{KabSchw15}. As we see in \autoref{fig:NPCdiffraction_intro}, even if only one protein in each symmetrical subunit is labelled, diffraction leads to one blurry dot as the image of the complex, where we can neither identify the number of subunits nor their spatial arrangement.

\begin{figure}[ht] 
	\centering
	\includegraphics[width=1\textwidth]{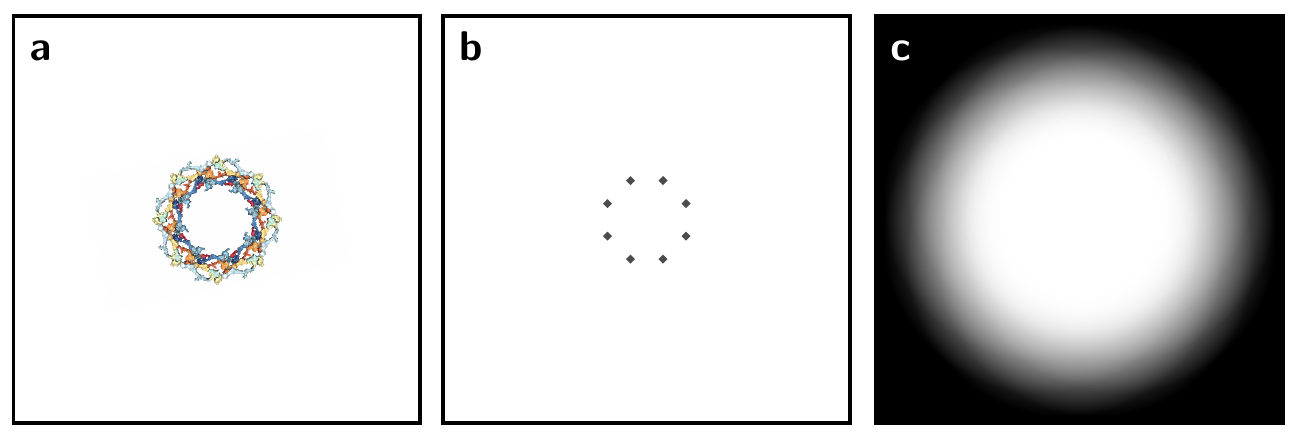}
	\caption{Illustration of the limit of light diffraction. (a) Crystallographic structure of the human NPC (pdb: 5A9Q) \cite{AppKosSpaOriGui15} viewed with NGL viewer \cite{RosBraValDuaPrl18}. (b) Simplification of the NPC structure showing its eight symmetric units. c) Representation of an ideal, diffraction-limited image of the structure in b.}
	\label{fig:NPCdiffraction_intro}
\end{figure}

Great efforts have been made to overcome this barrier, but it was not until the advent of superresolution microscopy that images with a resolution below the diffraction limit could be obtained. As a key asset, superresolution microscopy techniques circumvent Abbe’s limit of diffraction by utilizing photophysical properties of the fluorescent labels: they keep adjacent molecules at different fluorescence states, making it possible to differentiate them from each other. This is achieved using different techniques that can be combined into two general approaches:

\begin{enumerate}
	\item Techniques that use patterned illumination to control the fluorescence state of the labels, selecting which of them emit at a given moment. This approach includes, among others, Stimulated Emission Depletion (STED) \cite{KlaJakDybEgnHel00,KlaEngHel01,WesHel05}, Reversible Saturable Optical Fluorescence Transitions (RESOLFT)  \cite{HofEggJokHel05}, Minimal Photon Fluxes (MINFLUX) \cite{BalEilGwoWes17} or Saturated Structured Illumination Microscopy (SSIM) \cite{Gus05} methods.
	\item Techniques that use properties of the fluorescent labels to stochastically switch their fluorescent state, so that neighbouring labels do not emit at the same time. These techniques are commonly termed Single Molecule Localization Microscopy (SMLM) and include, among others, Stochastic Optical Reconstruction Microscopy (STORM) \cite{RusBat05}, Photoactivated Localization Microscopy (PALM) \cite{BetParSouLinOle06}, and DNA- Points Accumulation for Imaging in Nanoscale Topography (DNA-PAINT) \cite{JunSteScheKuzTin10}. In SMLM, the signals of the individual fluorophores are sequentially localized and used to reconstruct an image with subdiffraction resolution.
\end{enumerate}

In this work, we focus on SMLM techniques, where the working principle is described in \autoref{s:work}. 
The objective of this paper is to derive mathematical models of light propagation through the imaging device and to formulate associated inverse problems.
This sets the base for the formulation of the inverse problems of SMLM, which concerns the localization of the fluorescent labels with high localization precision and the reliable reconstruction of the imaged structures. We show that the currently used imaging workflow in SMLM can be viewed as solving 
an inverse problem for Maxwell's equation (see \autoref{s:ip}).
The inverse problem of SMLM has been previously investigated. In \cite{DavSwaUnlKarGol07}, a model for light propagation based on Maxwell’s equations is proposed and used to localize the positions and strengths of fluorescent dipoles. The model also accounts for the effects of the detection optics and employed a maximum likelihood reconstruction method. The inverse scattering problem with internal sources was investigated in \cite{GilLevScho18}, as a means of achieving sub wavelength resolution in SMLM. A local inversion formula was derived and the inverse problem was shown to be well-posed.

\section{Single Molecule Localization Microscopy (SMLM)} \label{s:work}

\subsubsection*{Principle of SMLM}

An SMLM experiment starts with the labelling of the proteins of interest with a fluorophore. There are different strategies for labelling, depending on the type of fluorescent probe, the molecule of interest, and its location in the cell. It should be taken into account that no labelling strategy is perfect, and labelling efficiency will likely be below $100\%$. In addition, the size of the probe or of the attachment of the linker molecule, in the cases were an intermediate is necessary, can affect the accuracy of the measurement. The influence of these aspects will be addressed in later sections in more detail.

During an SMLM measurement, the experimental conditions are tuned such that most of the fluorophores are in their dark state, and in each frame, a small subset of them is stochastically activated. The active fluorophores are sufficiently isolated from each other so that their PSFs do not overlap. After some time these fluorophores switch to a dark state, and a new subset of fluorophores is stochastically activated. This is repeated thousands of times, to ensure the collection of signals from enough fluorophores. We can see a scheme of an ideal experiment in \autoref{fig:principleSMLM}. The necessity for sparse labels per image and for enough localizations to reconstruct the structure results in movies with tens of thousands of frames. After data collection, all signals in all frames are fitted individually to obtain the coordinates of the fluorescent probes. All localizations are then collected and used to reconstruct a superresolution image.

\begin{figure}[ht] 
	\centering
	\includegraphics[width=1\textwidth]{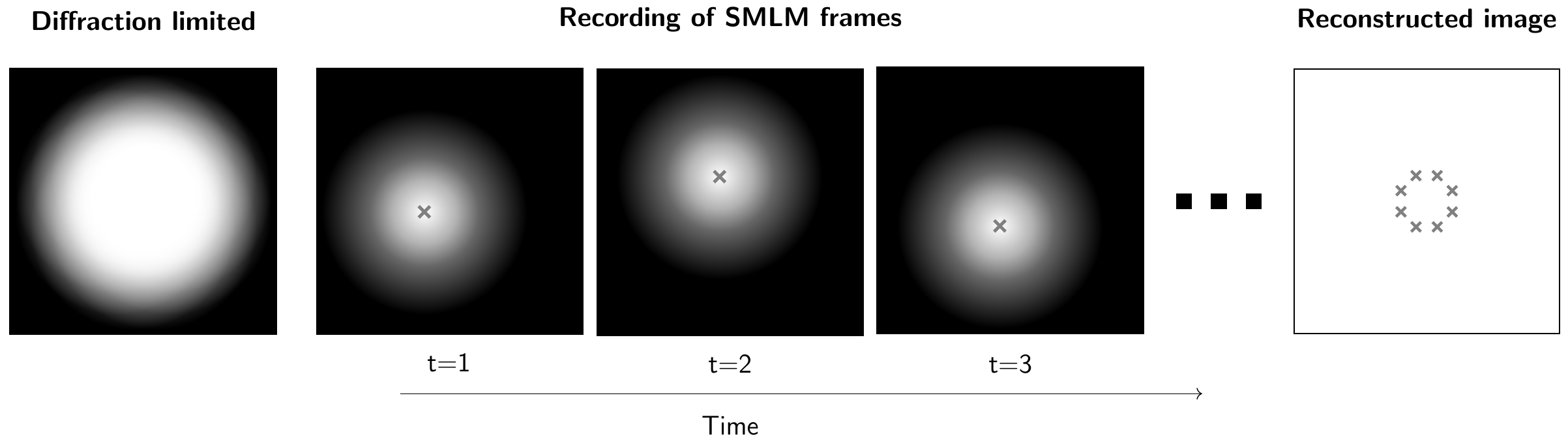}
	\caption{Scheme of an ideal SMLM experiment. In a classical diffraction-limited image, all the fluorophores are active, and the structure underneath – in this example an NPC – is unresolvable. In contrast, in an SMLM experiment only a sparse subset of fluorophores is active per image. In the first frame $(t=1)$ of this exemplary SMLM movie, only one fluorophore is active, while the others remain in their dark state. The PSF of this fluorophore, can be fitted mathematically, which yields the fluorophore localization. In $t\!=\!2$, the first fluorophore returns to its dark state, and another fluorophore is activated and can now be localized. This is repeated until all fluorophores have been localized. All localizations are collected in a final reconstructed image, which corresponds to the structure shown in \autoref{fig:NPCdiffraction_intro}b.}
	\label{fig:principleSMLM}
\end{figure}

The fluorophores used in SMLM are able to spontaneously change their fluorescence state. This property is commonly known as photoswitching or blinking. One dark-bright-dark cycle is usually called a blink. Photoswitching mechanisms are different for different kind of probes and can be the result of conformational changes in the dye molecule, chemical changes, or binding events. Typically, a combination of light illumination and the choice of special chemical conditions is used for deactivation, i.e. the transitions to a long-lived dark state. The activation, i.e. the transition back from this state, is usually light-induced, although other phenomena may apply (for example, binding events in the case of DNA-PAINT microscopy). An extensive review of available SMLM fluorophores and their properties can be found in \cite{LiVau18}.

In an ideal experiment, each fluorophore undergoes exactly one blink, in which it emits a high number of photons, and it remains in a dark state for the rest of the measurement. Commonly, however, fluorophores undergo multiple blinking events, or remain inactive during the whole imaging procedure. These non-ideal behaviours directly influence the quality of the final image, and they should be considered when analysing the data, as will be detailed in the next sections. The emission behaviour of a fluorescent label, and therefore the quality of the collected data, will depend largely on the kind of fluorophore used, the labelling strategy followed, and the environmental conditions of the fluorophore \cite{LevRan11}.

\subsubsection*{Fitting of localizations}

In an SMLM experiment, thousands of individual frames are recorded. Obtaining the final image requires post-processing of the recorded raw data. All blinking events are analyzed and the positions of the molecules are determined by fitting their signals. A variety of algorithms and software packages exist that can be applied to analyze the data \cite{SagPhaBabLukPen19}. Often, a Gaussian function is fitted to the detected intensity data using a maximum likelihood or least squares method. The coordinates of the center of the Gaussian peak are then taken as the position of the molecule. Finally, the localizations obtained from all recorded frames are combined to yield the reconstructed image.

\subsubsection*{Localization error and bias}
The achievable resolution in SMLM depends on how well the position of a molecule can be estimated by fitting its PSF. The fitting procedure is influenced by various factors of signal quality, including brightness, background noise and the pixel size of the detector. The error in the estimation of the molecule position follows a normal distribution. Its standard deviation is referred to as localization precision $\sigma_{\text{loc}}$. The mean of the error distribution is the localization accuracy $\mu_{\text{loc}}$. In the optimal case, it holds that $\mu_{loc}=0$, i.e. the estimation is unbiased. However, in practice a bias in the localization procedure may be present, e.g. due to distortions of the PSF. A bias may also arise from the labeling procedure. The size of some labels itself can be rather large, which displaces the position of the fluorophore from the actual molecule of interest by up to tens of nanometers.
Various formulas for the estimation of the localization precision $\sigma_{\text{loc}}$ have been proposed in the literature \cite{DesZanMloDiaBew14}. The theoretical limit for the best achievable localization precision is given by the Cramér-Rao lower bound (CRLB), which is critically dependent on the collected number of photons \cite{SmiJosRieLid10}.

\subsubsection*{Blinking and overcounting }
In SMLM, fluorophores switch between a fluorescent on-state and a non-fluorescent off-state. The transitions between the two states occur stochastically. Ideally, each fluorophore is detected exactly once during the whole imaging procedure, i.e. it is in the on-state in exactly one frame.

However, this is unlikely in a real experimental situation. Due to the stochastic nature of transitions between the states, fluorescent dyes can stay in the on-state for several consecutive frames and, moreover, repeatedly switch between the on- and the off-state. Thus, a single molecule may be detected multiple times. However, the position coordinates assigned to each detection slightly differ due to localization errors. Hence, it is not possible to distinguish whether localizations belong to one blinking molecule or to different molecules. Overcounting of single protein molecules may also occur as a consequence of non-stoichiometric labeling: Depending on the labeling procedure, a single molecule of interest does not necessarily carry one fluorescent dye only, but may be linked to multiple dyes.

The problem of overcounting is depicted in \autoref{fig:locmap}. Here, individual molecules of the NPC are assumed to be detected multiple times during the imaging procedure, leading to a misrepresentation of the actual structure.

Blinking statistics can be determined experimentally by labeling at sufficiently low concentrations of the dye, so that localizations from individual molecules of interest can be well separated. Analysis of the acquired localization data allows to determine statistics for the number of detections of individual molecules of interest, the duration of emission bursts ($t_{\text{on}}$) and the duration of dark times ($t_{\text{off}}$). In \autoref{fig:timeTraceBlink}, a schematic of a time trace of occupied states for an individual molecule is shown. An exemplary result for the blinking statistics of a fluorescent dye is depicted in \autoref{fig:blinkStat}.

\begin{figure}[ht] 
	\centering
	\includegraphics[width=0.6\textwidth]{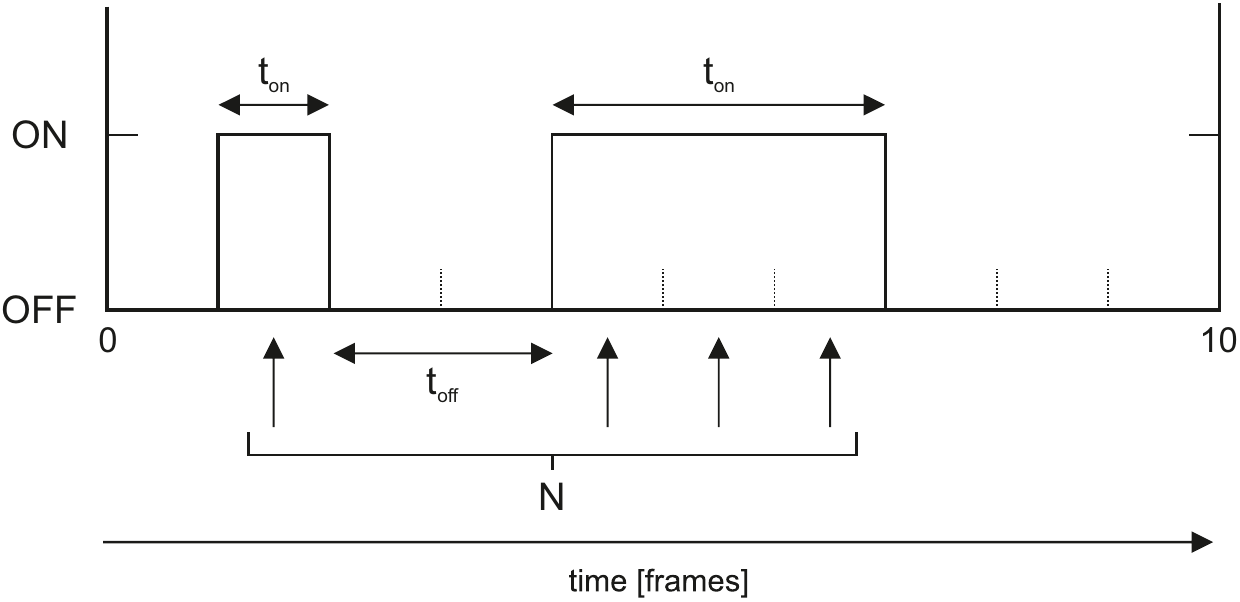}
	\caption{Exemplary time trace for a fluorophore. The fluorophore can switch between a dark off-state and a bright on-state. Indicated are the on- and off-time ($t_{\text{on}}$, $t_{\text{off}}$), representing the number of consecutive frames the molecule is in its bright or dark state, respectively, and the number of detections $N$.}
	\label{fig:timeTraceBlink}
\end{figure}

\begin{figure}[ht] 
	\centering
	\includegraphics[width=0.8\textwidth]{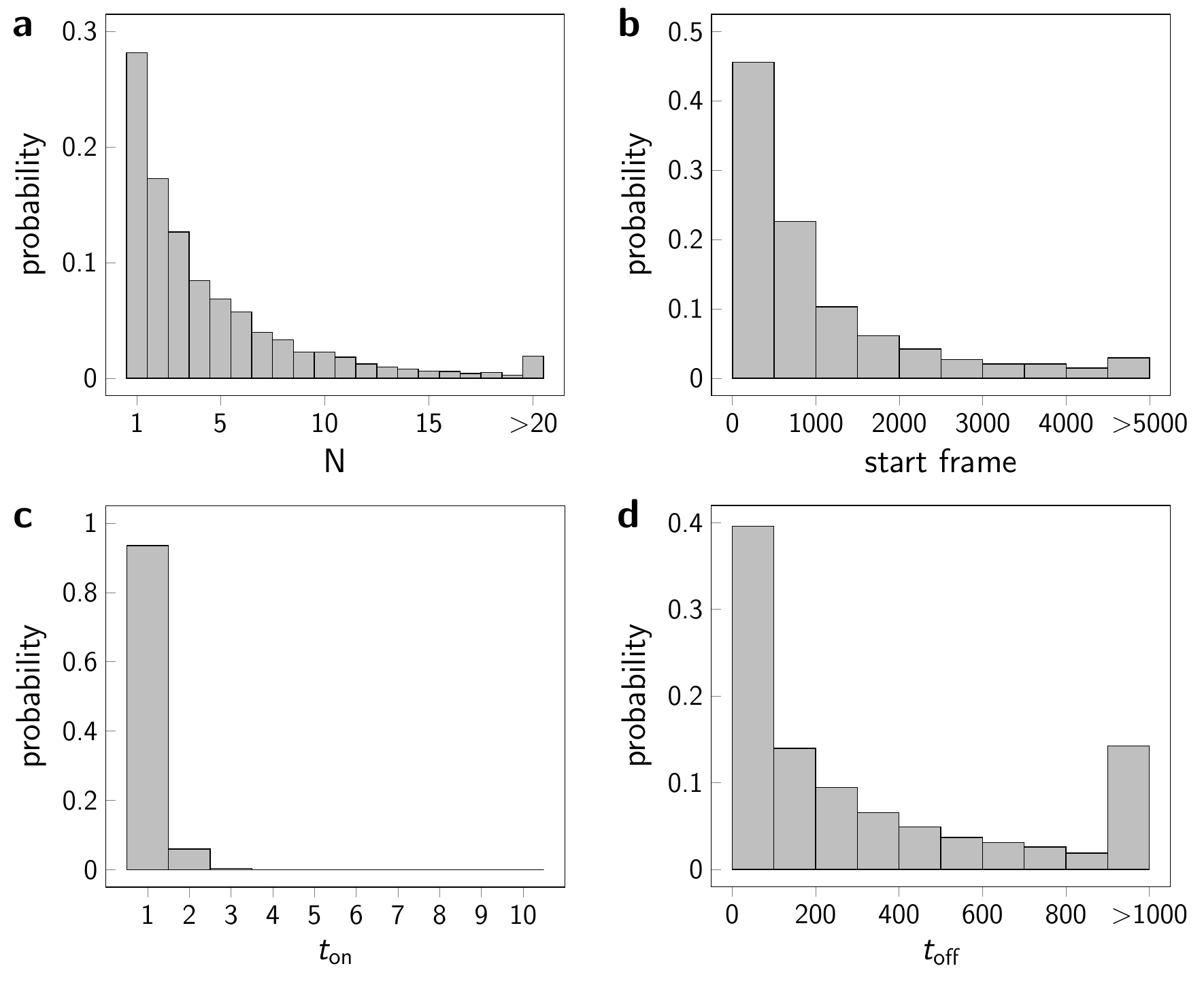}
	\caption{Experimentally derived blinking statistics for Alexa Fluor 647, a commonly used fluorophore for SMLM. Shown are histograms for the number of detections $N$ of a single fluorophore (a), the on-time $t_{\text{on}}$ (b) and the off-time $t_{\text{off}}$ (c). }
	\label{fig:blinkStat}
\end{figure}

A simple approach to account for multiple detections of the same molecule is to merge localizations that occur in close spatial and temporal proximity \cite{AnnVanScaRotRad11}. However, the results of this method highly depend on the chosen thresholds. Moreover, it cannot account for long-lived dark states.
Other post-processing algorithms rely on experimentally derived blinking statistics in order to correct for overcounting. However, care must be taken here, because photophysics of fluorophores, in particular blinking, depends on the local environment of the dye   and may likely vary under different experimental conditions \cite{LevRan11}. An overview over different methods for correcting overcounting artifacts is given in  \cite{BauArnRosBraSchu19}.

\subsubsection*{Forward simulation of SMLM localization maps}
In the following, we describe the main steps in the simulation of localization maps obtained by a 2D SMLM experiment. \autoref{fig:locmap} shows simulation results of the spatial arrangement for the example of NPCs.

\begin{figure}[ht] 
	\centering
	\includegraphics[width=0.9\textwidth]{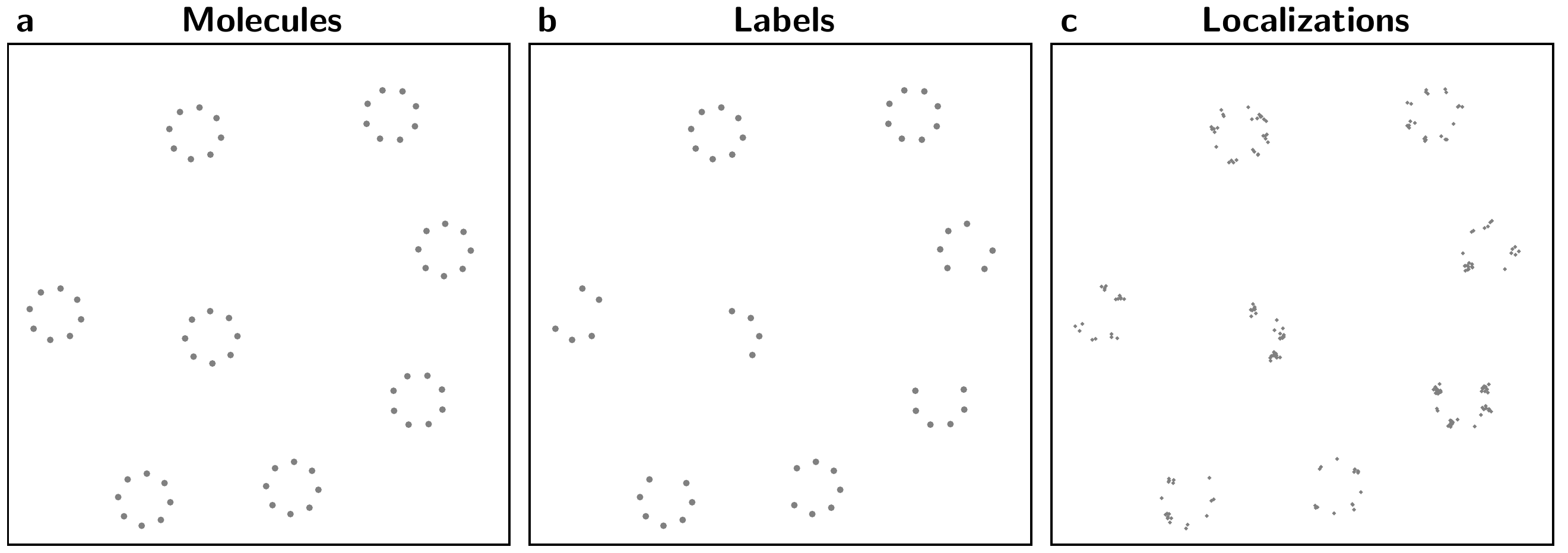}
	\caption{Simulation of a SMLM experiment for the nuclear pore complex (NPC). (a) True spatial arrangement of molecules. The distance between two neighboring molecules of the NPC was set to 40nm. The molecules are labeled with fluorescent probes (b) with a labeling efficiency of $80\%$. The SMLM experiment was simulated with a localization precision of $\sigma_{\text{loc}}\!=\!\SI{5}{\nano\meter}$, and the blinking statistics from \autoref{fig:blinkStat}. (c) Obtained localization map. Due to overcounting and the finite localization precision, individual molecules are observed multiple times.}
	\label{fig:locmap}
\end{figure}

The actual question of interest is the structural arrangement of molecules in a cell membrane. The first step in the simulation is therefore to spread the position of molecules on the region of interest according to the desired distribution. For example, the molecules can be spread randomly, in clusters, or as oligomers of a certain shape. The assigned positions represent ground truth.

As a second step, the simulated molecules are fluorescently labeled. In real experimental conditions, not all molecules of interest are detected: some proteins are not bound to a dye, or the dye is never detected during the imaging time. In the simulations this is accounted for by adjusting the mean labeling efficiency, a parameter in the interval $\left[0,1\right]$ that determines the mean fraction of molecules observed in the experiment. Labeled molecules are selected randomly from all simulated molecules according to the chosen distribution.

Next, overcounting has to be included in the simulation. As described above, a protein molecule can be detected multiple times during the whole imaging procedure. To account for this in the simulations, the number of detections of each molecule of interest, the frame of its first appearance  and the duration of on- and off-times are included. For each labeled molecule, these variables are drawn randomly either from experimentally acquired blinking statistics or from specified theoretical distributions. This allows to assign to each molecule a list of those frames, in which it is detected.

The last step in the simulation is to account for measurement errors. For each detection of a molecule, its true simulated position is displaced by adding a localization error, which is drawn randomly from a normal distribution. The mean and the standard deviation of the error distribution correspond to the localization accuracy and localization precision, respectively. Ideally, the mean value is zero, i.e. the localization is accurate. However, inaccuracy may occur, e.g. due to certain properties of the labeling procedure. Localization precision depends mainly on the collected number of photons and background noise. Typical values that are achieved in SMLM experiments are commonly around 10nm, but precisions of \SI{1}{\nano\meter} have been claimed.

The final result of the simulation is the localization map, i.e. a list of localization coordinates with the according frame numbers of detection. An exemplary simulated localization map for the NPC is shown in \autoref{fig:locmap}. The obtained localizations are the basis for further analysis.

In the next section we model the experiment mathematically. In order to do so we summarize essential notation first in \autoref{tab:notation}.
\begin{center}
\begin{table}
 \begin{tabular}{|c|c||c||c|c|}
 \hline
 Symbol & Description & Reference & Relations & Units \\
 \hline
 \hline
 $\texttt{d}_0$ & maximum thickness of lens & & & $m$\\
 $\texttt{d}$ & thickness of lens (function of height) & \autoref{fig:Scheme} & & $m$\\
 $\flens$ & focal length of the tube lens & \autoref{fig:Scheme} & & $m$ \\
  $\focus_\obj$ & focal length of the objective & \autoref{fig:Scheme} & & $m$ \\
 \hline
 $\lambda$ & wavelength & \autoref{eq:diffraction_limit} & & $m$\\
 $\refractive = 1$ & refractive index in vacuum & \autoref{eq:diffraction_limit} & & $-$ \\ 
 $\refractive_l$ & refractive index of lens & \autoref{eq:phase_shift} & & $-$ \\
 $d$ & resolution limit & \autoref{eq:diffraction_limit} & & $m$\\
 $\NA$ & numerical aperture & \autoref{eq:diffraction_limit} & $d = \frac{\lambda}{2 \NA}$ & $-$\\
 $\tmax \in [0,\pi/2)$ & angle of aperture & \autoref{eq:diffraction_limit} & $\NA = \refractive \sin(\tmax)$ & $-$\\
 $\epsilon_0$ & electric permittivity (vac.) & \autoref{eq:dipl} & & $F/m$ \\ 
 $\mu_0$ & magnetic permeability (vac.) & \autoref{eq:mag} & & $H/m$ Henries per $m$\\
 $\omega$ & wave frequency & & &$Hz = 1/s$\\
 $c$ & light speed (vac.)& \autoref{eq:kappa} & & $m/s$ \\
 $\kappa$, $\kappa_\ve$ & wave number & \autoref{eq:kappa} & $\kappa = \frac{\omega}{c} = \frac{2\pi}{\lambda}$ & $1/m$\\
 $\chi$ & susceptibility & \autoref{chi_1} & & $-$\\ 
 \hline
 $\bpsi$ & dipol & \autoref{eq:emit} & & \\
 $\Psi_p$, $\Psi_s$ & dipol components & \autoref{eq:ps3} & & \\
 \hline
\end{tabular}
\caption{\label{tab:notation} Physical parameters used in the paper and dimensions.}
\end{table}
\end{center}
\clearpage

\section{Mathematical Prerequisites}
In what follows we summarize some basic mathematical framework:

\subsection{Distributions}
In order to define distributions (generalized functions) we need to introduce appropriate function spaces first:
\begin{definition} The \emph{Schwartz-space} of functions from $\R^n$ to $\C$ is defined as  
 \begin{equation} \label{eq:Schwarz}
  \sspace{\R^n;\C} := \set{ \phi \in C^\infty(\R^n;\C): 
                         \text{ for all } \alpha,\beta \in \N_0^n,  
                         \norm{\phi}_{\alpha,\beta} := \sup_{\bx \in \R^n} \abs{\bx^\alpha \partial^\beta \phi(\bx)} < \infty}.
 \end{equation}
 Accordingly the \emph{Schwartz-space} of \emph{vector valued functions} is defined by 
 \begin{equation} \label{eq:Schwarz_m}
  \sspace{\R^n;\C^m} := \set{ {\bf \Phi} \in C^\infty(\R^n;\C^m): \Phi_i \in \sspace{\R^n;\C}, i=1,\ldots,m}.
 \end{equation}
 The space of linear functionals $T:\sspace{\R^n;\C^m} \to \C$ for which there exist $k,l \in \N_0$ and some $C > 0$ such that for all 
 ${\bf \Phi} \in \sspace{\R^n;\C^m}$ the following inequality holds 
 \begin{equation}\label{eq:td}
 \abs{\scal{T}{\bf \Phi}} := 
         \abs{T {\bf \Phi}} \leq C \sum_{i=1}^m \sum_{\abs{\alpha} \leq k, \abs{\beta} \leq l} \norm{\Phi_i}_{\alpha,\beta}
 \end{equation}
 is called \emph{space of tempered distributions} and is denoted by $\tempdist{\R^n;\C^m}$.
\end{definition}

\begin{definition}[Causal Distribution]
\label{def:causdist}
 A tempered distribution $T \in \mathcal S'(\R^{n-1} \times \R;\C)$ is called \emph{causal} if its support in time is included in $[0,+\infty)$. That is $T$ is causal if and only if for all test functions $\phi \in \mathcal S(\R \times \R^{n-1};\C)$ which satisfy 
 \[\phi(\bx,t) = 0 \text{ for all } t \gs 0, \, \bx\in \R^{n-1}, \ \] 
 we have 
 \[ \scal{T}{\phi} = 0.\]
\end{definition}
For causal distributions, the quantity $|\scal{T}{\phi}|$ can be estimated as follows.
\begin{lemma}
 \label{lem:causdist}
 Let $T \in \mathcal S'(\R^{n-1} \times \R;\C)$ be causal. Then, there exists a constant $C>0$ (which depends only on $l$) 
 such that for all test functions $\phi \in \mathcal S(\R^{n-1}\times \R;\C)$ the following estimate holds:
 \begin{equation}
 \label{eq:estcausdist}
 \left \vert \scal{T}{\phi} \right \vert \ls C \sup_{\abs{\alpha} \leq k, \abs{\beta} \leq l} 
 \sup_{\substack{t \gs -1 \\ \bx\in \R^{n-1}}} |(\bx , t)^\alpha \partial^\beta \phi(\bx , t)|.
 \end{equation}
\end{lemma}
\begin{proof}
 Let $\varkappa$ be a $\mathcal C^\infty(\R;\R)$ cut-off function, that satisfies $\varkappa(t) = 1$ for $t \gs 0$ and $\varkappa(t) = 0$ for $t \ls -1$. Then, for all 
 test functions $\phi \in \mathcal S(\R^{n-1} \times \R;\C)$, we define $\psi = \varkappa \phi.$ For $t \gs 0$, we have $\psi - \phi = 0$, which means, since $T$ is causal, 
 that $\scal{T}{\psi - \phi} = 0$, or in other words,
 \begin{equation} \scal{T}{\phi} = \scal{T}{\psi}.
 \label{eq:tphieqtpsi}
 \end{equation}
 Now, let us use the definition of tempered distribution for $T$: there exists $k,l \in \mathbb N_0$ and $\tilde C>0$ such that
 \begin{equation} |\scal{T}{\psi}| \ls \tilde C \sup_{|\alpha| \ls k, |\beta| \ls l} \ \sup_{(\bx , t) \in \R^{n-1} \times \R}  |(\bx , t)^\alpha \partial^\beta \psi(\bx , t)|.\label{eq:Tistd}\end{equation}
 Now, the function $\psi$ has support in $\R^{n-1} \times [-1,+\infty)$, which means that the last inequality can rewrite
\[ |\scal{T}{\psi}| \ls C \sup_{|\alpha| \ls k, |\beta| \ls l} \ \sup_{(\bx , t) \in \R^{n-1} \times [-1,+\infty)}  |(\bx , t)^\alpha \partial^\beta \psi(\bx , t)|.\]
Finally, denoting $\beta = (\beta_t,\beta_{x_1}, \cdots,\beta_{x_{n-1}})$, one can expand
\[ \partial^\beta \psi = \sum_{i = 0}^{\beta_t} \binom{\beta_t}{i}\varkappa^{(i)}(t) \partial^{(\beta_t-i,\beta_{x_1}, \cdots,\beta_{x_{n-1}})} \phi (\bx , t). \]
Since the function $\varkappa$ is fixed (independent of $\phi$), the quantity
\[C_l = \sup_{i \ls l} \sup_{ t \in \R} |\varkappa^{(i)} (t)|\]
is finite and independant of $\phi$ and we have 
\[|\partial^\beta \psi(\bx , t)| \ls 2^l C_l \sup_{\abs{\beta} \leq l} |\partial^\beta \phi(\bx , t)|\]
and we can finally conclude, taking the supremum on $t \gs -1$, that
\[\sup_{|\alpha| \ls k, |\beta| \ls l} \ \sup_{(\bx , t) \in \R^{n-1} \times [-1,+\infty)}  |(\bx , t)^\alpha \partial^\beta \psi(\bx , t)| \ls 2^l C_l \sup_{|\alpha| \ls k, |\beta| \ls l}\sup_{(\bx , t) \in \R^{n-1} \times [-1,+\infty)}  |(\bx , t)^\alpha \partial^\beta \phi(\bx , t)|.  \]
Since $\psi$ has a support included in $[-1,+\infty)$, one can in the left hand side of the inequality take the supremum over $t \in \R$. Plugging this inequality into \autoref{eq:Tistd} and recalling \autoref{eq:tphieqtpsi}, we get \autoref{eq:estcausdist}.
 \end{proof}
 
 We need to notationally differ between $\delta$-distribution in different dimensions:
\begin{definition}[$\delta$-Distributions] \label{de:dist}
 $\delta: \R^3 \to \R$ denotes the $3$-dimensional $\delta$-distribution. $\eindelta : \R \to \R$ denotes the $1$-dimensional $\delta$-distribution.
 For $r_0 \in \R$, $\eindelta_{r_0}: \R \to \R$ is defined by $\eindelta_{r_0}(r)=\eindelta(r-r_0)$ for all $r \in \R$.
 $\eindelta':\R \to \R$ denotes the derivative of the $1$-dimensional $\delta$-distribution.
\end{definition}

\subsection{Fourier- and \texorpdfstring{$k$}{k}-Transform}
The most important mathematical tool in this paper is the Fourier-transform:

\begin{definition}[Temporal Fourier-Transform] 
 Let $T \in \mathcal S'(\R \times \R^{n-1},\C).$ We define its \emph{Fourier-transform} $\hat T$ by its action on a test function $\phi \in \mathcal S(\R \times \R^{n-1},\C)$
 \begin{equation}
  \label{eq:FT}
  \scal{\hat T}{\phi} := \scal{T}{\check{\phi}}
 \end{equation}
where
\[ \check{\phi}(t) := \frac{1}{\sqrt{2 \pi}} \int_{\xi = -\infty}^{\infty} \e^{\i\xi t} \phi(\xi) \dd \xi.\]
As defined, the operator $T \mapsto \hat T$ is well defined and continuous from $S'(\R \times \R^{n-1}, \C)$ into itself \cite{Gru09} with inverse $T \mapsto \check T$ with for all test functions $\phi$,
\[ \scal{\check T}{\phi} := \scal{T}{\hat \phi} \]
where
\[ \hat \phi(\xi) := \frac{1}{\sqrt{2 \pi}} \int_{t = -\infty}^{\infty} \e^{-\i \xi t} \phi(t) \dd t\]
is the Fourier-transform on the Schwartz space (it coincides with the one given in \autoref{eq:FT}). 
\end{definition}

The Fourier-transform in spatial variables is called the $k$-transform:
\begin{definition}[$k$-transform] \label{eq:k} 
Let $i, \hat{i} \in \set{1,2,3}$, $\br, \bk \in \R^3$ and denote by $\brk^{(i)}, \brk^{(i,\hat{i})} = (\tilde{k}_1,\tilde{k}_2,\tilde{k}_3)$
where 
$$\tilde{k}_j = \left\{ \begin{array}{rl} 
                   r_j & \text{ if } j \neq i\\
                   k_i & \text{ if } j=i
                  \end{array} \right. \quad \text{ and } \quad \tilde{k}_j = \left\{ \begin{array}{rl} 
                   r_j & \text{ if } j \neq i \text{ and } j \neq \hat{i}\\
                   k_i & \text{ if } j=i \text{ or } j=\hat{i}
                  \end{array} \right.,
                  $$
respectively.

The $k$-transform of the Fourier-transform of ${\bf V}: \R^3 \to \C^3$ in direction $x_i$, $(x_i,x_j)$ and in all three directions 
are defined by
\begin{equation*}
\boxed{
\begin{aligned}
  \mathcal{F}_i[{\bf V}](\brk^{(i)}) &:=  \frac{1}{\sqrt{2\pi}} \int_{\R} \e^{-\i k_i r_i} {\bf V}(\br) d r_i, \\
  \mathcal{F}_{ij}[{\bf V}](\brk^{(i,j)}) &:=  \frac{1}{2\pi} \int_{\R} \int_{\R} \e^{-\i (k_i r_i + k_j r_j)} {\bf V}(\br) d r_i d r_j, \\
  \mathcal{F}[{\bf V}](\bk) &:=  \mathcal{F}_1[\mathcal{F}_2[\mathcal{F}_3[{\bf V}]]](\bk) = 
  \frac{1}{{(2\pi)}^{\frac{3}{2}}} \int_{\R^3} \e^{-\i \bf k \cdot \bf r} {\bf V}(\br) d \br. 
  \end{aligned}}
\end{equation*}
\end{definition}

\begin{remark}
 From \autoref{de:dist} it follows that for $\br^0 \in \R^3$ fixed 
 \begin{equation} \label{eq:dist3}
  \mathcal{F}[\br \to \delta(\br-\br^0)](\bk) = \frac{1}{{(2\pi)}^{\frac{3}{2}}} \e^{-\i \bf k \cdot \bf r^0}.
 \end{equation}
\end{remark}

\subsection{Coordinate Systems}
\begin{definition}[Spherical Coordinates] \label{de:basis}
 Associated to $\br = \begin{pmatrix}
                         r_1 \\ r_2 \\ r_3
                        \end{pmatrix}
\in \R^3$ is the polar coordinate representation $(r=\abs{\br},\theta,\varphi) \in [0,\infty) \times [0,\pi] \times [0,2\pi)$ 
 such that
 \begin{equation}\label{eq:br}
  \br = r \begin{pmatrix} 
                  \sin(\theta) \cos(\varphi) \\ \sin(\theta) \sin(\varphi) \\ \cos(\theta)
                 \end{pmatrix}.
 \end{equation}
 \end{definition}
\section{Mathematical Modeling of Light Propagation} \label{sec:model}
We consider an optical single molecule localization microscopy experiment. Therefore a mathematical modeling of the light propagation via Maxwell's equation is appropriate:
We consider macroscopic Maxwell’s equations (in SI units), in order to model the interaction of the incoming light with the sample. 
These equations describe the time evolution of the \emph{electric field} $\bE: \R^3 \times \R \to \R^3 $  and the \emph{magnetic field} 
${\bf B}: \R^3 \times \R \to \R^3$ for a given \emph{charge density} $\rho : \R^3 \times \R \to \R $ and an \emph{electric current} 
${\bf J}: \R^3 \times \R \to \R^3$:
\begin{subequations}\label{eqMacroMaxwell}
\begin{align}\label{eq:Gelectriclaw_mac}
\Div {\bf D}(\br;t) &= \rho(\br, t),  &\br \in \R^3, t \in \R, \\ \label{eq:Melectriclaw_mac}
\Div {\bf B}(\br;t) &= 0, &\br \in \R^3, t \in \R, \\ \label{eq:Faradaylaw_mac}
\Curl \bE(\br, t)&= -\partial_t {\bf B}(\br;t), &\br \in \R^3, t \in \R, \\ \label{eq:Amperelaw_mac}
\Curl {\bf H}(\br;t) &= \partial_t {\bf D}(\br;t) + {\bf J}(\br;t), &\br \in \R^3, t \in \R.
 \end{align}
\end{subequations}
Here 
\begin{equation} \label{eq:dipl}
{\bf D} \equiv \epsilon_0\bE + {\bf P}
\end{equation}
denotes the \emph{electric displacement} and 
\begin{equation} \label{eq:mag}
 {\bf H} \equiv \frac{1}{\mu_0}{\bf B} -{\bf M}
\end{equation} 
denotes the \emph{effective magnetic field}, related to the \emph{electric} and \emph{magnetic polarization fields} ${\bf P}$ and ${\bf M}$, 
respectively. All along this paper the differential operators $\Grad$, $\Div$, $\Curl$, $\Delta$ are meant with respect to the variables $\br$.
More background on modeling of electromagnetic wave propagation can be found in \cite{Jac98}. 

In the following we make a series of assumptions for simplifying Maxwell's equations: 
\subsection{Material Properties}
Biological specimens as we are considering in single molecule localization microscopy experiments can be assumed to be non-magnetizable:
\begin{assumption}[Non-Magnetizeable Medium] \label{ass:maxwell}
  A medium is \emph{non-magnetizable} if 
  \begin{equation}\label{eq:ass:maxwell}
   {\bf M}(\br;t) = 0 \text{ for all } \br \in \R^3, t \in \R.
  \end{equation}
\end{assumption}
\begin{remark}
 In single molecule localization microscopy experiments, fluorescent dyes are attached to molecules of interest and upon excitation of the probe with a strong laser impulse they emit light. 
 The mathematical modeling of this process is omitted and we are considering only the influence on a macroscopic level, meaning that 
 charge density and currents are induced. A detailed mathematical modeling of the chemical processes would require a modeling with \emph{microsopic} 
 Maxwell's equations, which is omitted here for the sake of simplicity. In a similar context microscopic Maxwell's equations have been considered 
 in Optical Coherence Imaging in \cite{ElbMinSch15}. 
 
 On a macroscopic level, from \autoref{eq:Gelectriclaw_mac} -- \autoref{eq:Amperelaw_mac} it follows from \autoref{ass:maxwell} that
 \begin{equation} \label{eq:J}
  \partial_t \rho (\br;t) = - \Div {\bf J}(\br;t) \text{ for all } \br \in \R^3, t \in \R.
 \end{equation}
\end{remark}

Taking into account \autoref{ass:maxwell} and combining \autoref{eq:Faradaylaw_mac} and \autoref{eq:Amperelaw_mac} we obtain the 
\emph{vector Helmholtz equation} for the \emph{electric field} $\bE$:
\begin{align}\label{eq:VH}
\boxed{
 \Curl \Curl \bE(\br;t) +\frac{1}{c^2} \partial_{tt} \bE(\br;t) =  
 - \frac{1}{\epsilon_0 c^2}\partial_{tt} {\bf P}(\br;t) - \frac{1}{\epsilon_0 c^2}\partial_{t} {\bf J}(\br;t) \text{ for all } \br \in \R^3, t \in \R}
\end{align} 
where $\mu_0 \epsilon_0 = 1/c^2$, with $c$ being the speed of light in vacuum.

\begin{remark}
If the right hand side of \autoref{eq:VH} vanishes than $\bE$ describes the propagation of the electric field in vacuum. 
The right hand side models the interaction of light and matter and the effect of the external charges.

\autoref{eq:VH} is understood in a distributional sense. That means that for every 
${\bf \Phi} \in \sspace{\R^3;\R^3}$ and $\bpsi \in \sspace{\R;\R^3}$, and with ${\bf \Phi} \otimes \bpsi \in \sspace{\R^3 \times \R;\R^3}$ denoting 
the vector valued function consisting of componentwise multiplication,
\begin{align}\label{eq:VH_weak}
\boxed{
\scal{\bE}{(\Curl \Curl {\bf \Phi}) \otimes \bpsi} +
\scal{\bE}{\frac{1}{c^2} {\bf \Phi} \otimes \partial_{tt} \bpsi} =  
 - \scal{{\bf P}}{\frac{1}{\epsilon_0 c^2}{\bf \Phi} \otimes \partial_{tt}\bpsi}  
 + \scal{{\bf J}}{\frac{1}{\epsilon_0 c^2} {\bf \Phi} \otimes \partial_t \bpsi}.}
\end{align} 
\end{remark}

\subsection{Linear optics}
In linear optics one assumes a linear relation between the magnetic polarization ${\bf P}$ and the electric field $\bE$. 
\begin{assumption}[Polarization Response Function in Linear Optics] \label{ass:lprf} ${\bf P}$ and $\bE$ satisfy the linear relation,
\begin{align}\label{P1_T}
 {\bf P}(\br;t) = \epsilon_0 \int_{\tau = -\infty}^{\infty} \tensortwo(\br;t,\tau) \bE(\br,\tau) d \tau, 
\end{align} 
where $(t;\tau) \to \tensortwo(\br;t,\tau) \in \R^{3 \times 3}$ is a matrix valued function that averages the electric field over time. 
$\tensortwo$ is called the (linear) \emph{polarization response function}.
For fixed $\br$ the matrix valued function $(t;\tau) \in \R^2 \to \tensortwo(\br;t,\tau) \in \R^{3\times 3}$ 
is supposed to satisfy the following assumptions:
\begin{description}
 \item{ \bf Causality:} No polarization is observed before the field is induced, i.e. 
      \begin{equation*}
       \tensortwo(\br;t,\tau)=0, \quad \text{ for all } t\leq \tau.
      \end{equation*}
 \item{ \bf Time invariance} means that $(t;\tau) \to \tensortwo(\br;t,\tau)$ is just a function of $t-\tau$. That is, we can write 
      \begin{equation*}
       \tensortwo(\br;t-\tau) = \tensortwo(\br;t,\tau), \quad \text{ for all } t,\tau \in \R.
      \end{equation*}
      Here we use a slight abuse of notation and identify notationally the two functions $\tensortwo$ on the left and right hand side. 
\end{description}
\end{assumption}

\begin{remark} 
Let \autoref{ass:lprf} hold, then $\tensortwo(\br;t-\tau)=0$ for $t\leq \tau$.
\end{remark}

We now move on to the Fourier-Laplace domain. In order to do so we postulate causality assumptions, which we assume to hold all along 
the remaining paper:
\begin{assumption}[Causality] \label{de:causal}
 The functions ${\bf J}, {\bf P}, \bE$ (and thus in turn $\rho$, ${\bf D}$, ${\bf H}$) are \emph{causal}, meaning that 
 \begin{equation}\label{eq:causal}
  {\bf J}(t;\br) = {\bf P}(t;\br) = \bE(t;\br) = 0 \text{ for all } t < 0, \br \in \R^3. 
 \end{equation}
\end{assumption}

Let \autoref{ass:lprf} hold (in particular we assume that $\tensortwo$ is time invariant and causal), and assume that ${\bf J}, {\bf P}, \bE$
are causal, then from the Fourier convolution theorem it follows that 
\begin{align}\label{P1_chi}
\boxed{ \widehat{\bf P}(\br;\omega) = \epsilon_0 \chi(\br;\omega) \widehat\bE(\br;\omega), \quad \text{ for all }
\br \in \R^3, \omega \in \R,
}
\end{align} 
where 
\begin{equation}\label{chi_1}
\boxed{
 \chi(\br;\omega) = \int_{\tau=-\infty}^{\infty} \tensortwo(\br;\tau) \e^{-\i \omega \tau} d \tau = 
 \sqrt{2\pi} \widehat{\tensortwo}(\br;\omega) \in \C^{3 \times 3}
 \text{ for all } \br \in \R^3, \omega \in \R,}
\end{equation} 
is called the linear \emph{electric dipolar susceptibility}.

We denote the wave number by 
\begin{equation} \label{eq:kappa}
\kappa(\omega) := \frac{\omega}{c} \text{ and more general } \kappa_\ve:= \kappa_\ve(\omega) = \frac{\omega + \i \ve}{c} \text{ for all }\ve > 0.
\end{equation}
The application of the Fourier-transform to the vector Helmholtz equation \autoref{eq:VH} 
gives the following equation for the Fourier-transform $\widehat\bE : \R^3 \times \R \to \C^3$ of the electric field:
\begin{equation*} 
\Curl \Curl \widehat\bE(\br;\omega) - \kappa^2(\omega)  \widehat\bE(\br;\omega) =  
\frac{1}{\epsilon_0 } \kappa^2(\omega)  \widehat{\bf P}(\br;\omega) - \frac{\i \omega}{\epsilon_0 c^2} \fouriercurrent(\br;\omega),
\quad \text{ for all } \br \in \R^3, \omega \in \R
\end{equation*}
and consequently by using \autoref{P1_chi} we get
\begin{equation}\label{eq:VH2}
\boxed{
\Curl \Curl  \widehat\bE(\br;\omega) - \kappa^2(\omega) (\mathbb{I} + \chi(\br;\omega)) \widehat\bE(\br;\omega) = 
- \frac{\i \omega}{\epsilon_0 c^2} \fouriercurrent(\br;\omega) \quad \text{ for all }
\br \in \R^3, \omega \in \R,}
\end{equation} 
where $\mathbb{I} \in \R^{3\times3}$ is the identity matrix.

\subsection{Isotropic media}
Additional simplifications of Maxwell's equations can be made when the medium is assumed to be \emph{isotropic}:
\begin{assumption}[Isotropic Medium] \label{ass:iso}
 Let \autoref{ass:maxwell} and \autoref{ass:lprf} hold. The medium is \emph{isotropic} if the susceptibility is a multiple of the identity, that is 
 it can be written as $\chi(\br;t)\mathbb{I} \in \C^{3\times 3}$ with $\chi(\br;t) \in \C$. With a slight abuse of notation, we identify the diagonal 
 matrix and the diagonal entry. 
\end{assumption}

\subsection{Homogeneous Material}
We consider an isotropic, non magnetizable material with a linear polarization response (that is, \autoref{ass:maxwell}, \autoref{ass:lprf} and 
\autoref{ass:iso} are satisfied), which in addition is \emph{homogeneous}: 
\begin{assumption}[Homogeneous Material] \label{ass:hom}
 An isotropic, non magnetizable material with a linear polarization response is \emph{homogeneous} if $\chi \equiv 0$.
\end{assumption}
For a homogeneous material (that is $\chi \equiv 0$) it follows from \autoref{eq:VH2} that 
\begin{equation}\label{eq:bla1}
- \frac{\i \omega}{\epsilon_0 c^2} \fouriercurrent(\br;\omega) = 
\Curl \Curl  \widehat\bE(\br;\omega) - \kappa^2(\omega) \widehat\bE(\br;\omega).
\end{equation}
Thus, by using the vector identity
$$\Curl \Curl \widehat\bE = \Grad \Div \widehat\bE - \Delta_\br \widehat\bE,$$ 
we get from \autoref{eq:bla1}
\begin{equation} \label{eq:bla2}
- \frac{\i \omega}{\epsilon_0 c^2} \fouriercurrent(\br;\omega) = \Grad \Div \widehat\bE(\br;\omega) - 
\Delta_\br \widehat\bE(\br;\omega) - \kappa^2(\omega) \widehat\bE(\br;\omega).
\end{equation}
Now, by using \autoref{eq:dipl} and the assumption on homogeneity, $\chi \equiv 0$, which together with \autoref{P1_chi} implies that ${\bf P} \equiv 0$,
we get
\begin{equation*} 
 {\bf D} = \epsilon_0\bE + {\bf P} = \epsilon_0 \bE.
\end{equation*}
This, together with \autoref{eq:bla2} shows that
\begin{equation} \label{eq:bla3}
- \frac{\i \omega}{\epsilon_0 c^2} \fouriercurrent(\br;\omega) 
= \frac{1}{\epsilon_0} \Grad \Div \widehat{\bf D}(\br;\omega) - 
\Delta_\br \widehat\bE(\br;\omega) - \kappa^2(\omega) \widehat\bE(\br;\omega).
\end{equation}
Now, by using \autoref{eq:Gelectriclaw_mac} in Fourier domain we get from \autoref{eq:bla3}
\begin{equation} \label{eq:bla4}
- \frac{\i \omega}{\epsilon_0 c^2} \fouriercurrent(\br;\omega) 
= \frac{1}{\epsilon_0} \Grad \widehat{\rho}(\br;\omega) - \Delta_\br \widehat\bE(\br;\omega) - \kappa^2(\omega) \widehat\bE(\br;\omega).
\end{equation}
Finally, by using \autoref{eq:J} in Fourier domain,  
\begin{equation} \label{eq:FJ}
 \i \omega \widehat{\rho} = - \Div \fouriercurrent(\br;\omega)
\end{equation}
in \autoref{eq:bla4} we get
\begin{equation*} 
\begin{aligned}
~& - \frac{\i \omega}{\epsilon_0 c^2} \fouriercurrent(\br;\omega) 
=& -\frac{1}{\i \omega \epsilon_0} \Grad \Div \fouriercurrent(\br;\omega) - 
\Delta_\br \widehat\bE(\br;\omega) - \kappa^2(\omega) \widehat\bE(\br;\omega). 
\end{aligned}
\end{equation*}
 In other words, we have for every $\br \in \R^3$, $\omega \in \R$
\begin{equation}
\begin{aligned}\label{modeleq1}  
 \Delta_\br \widehat\bE(\br;\omega)+ \kappa^2(\omega) \widehat\bE(\br;\omega) &= \frac{\i}{\epsilon_0}
 \left( \frac{\omega}{c^2} + \frac{1}{\omega} \Grad \Div \right)\fouriercurrent(\br;\omega)\\
 & =\frac{\i \omega}{\epsilon_0 c^2} \fouriercurrent(\br;\omega) + \frac{1}{\epsilon_0} \Grad \widehat{\rho} (\br;\omega).
\end{aligned}
\end{equation} 

For any $\tau \in \R$ a solution of the nonhomogenous \autoref{modeleq1} is given by
\begin{equation}\label{modeleq1_soln_help} 
\begin{aligned}
  \widehat\bE(\br;\omega) &= \tau \widehat\bE^+(\br;\omega) + (1-\tau) \widehat\bE^-(\br;\omega) \quad \text{ for all } \quad 
  \br \in \R^3, \; \omega \in \R \quad \text{ where }\\
  \widehat\bE^\pm(\br;\omega) &:= \int_{\R^3} \mathcal{G}_\omega^\pm(\br,\br') \left(\frac{\i \omega}{\epsilon_0 c^2} 
               \fouriercurrent(\br;\omega) + \frac{1}{\epsilon_0} \Grad \widehat{\rho} (\br;\omega) \right) d \br' 
\end{aligned}
\end{equation}
with Green's functions:
\begin{equation}\label{green_func}
 \mathcal{G}_\omega^\pm(\br,\br') = \dfrac{\e^{\pm \i \kappa(\omega) \abs{\br-\br'}}}{4 \pi \abs{\br-\br'}}.
\end{equation}
The physically meaningful solution is, as we motivate below, a convolution with the \emph{retarded Green's} function $\mathcal{G}_\omega^+$:
That is, the retarded solution of the Helmholtz equation \autoref{eq:VH2} is given by \autoref{modeleq1_soln_help} with $\tau=1$ 
(see \cite{Tai94}):
\begin{equation} \label{eq:retarded}
\boxed{
  \widehat\bE(\br;\omega) = \int_{\R^3} \mathcal{G}_\omega^+(\br,\br') \left(\frac{\i \omega}{\epsilon_0 c^2} 
               \fouriercurrent(\br;\omega) + \frac{1}{\epsilon_0} \Grad \widehat{\rho} (\br;\omega) \right) d \br'.
}
\end{equation}

\begin{remark} 
 With a slight abuse of notation we identify $\mathcal{G}_\omega^+$ with $\mathcal{G}_\omega$ and 
 $\widehat\bE_\omega^+$ with $\widehat\bE_\omega$, since we are only interested in the retarded solutions.
\end{remark}

\section{Attenuating Solution and Initial Conditions}
\begin{definition}[Attenuating and Causal Solution of \autoref{eq:VH2}] \label{de:causal_s} Let $\ve > 0$ and $\kappa_\ve(\omega) = 
\frac{\omega + \i \ve}{c}$ as defined in \autoref{eq:kappa}.
\begin{itemize}
\item Then, we call $\widehat\bE_\ve$ the \emph{approximate attenuating solution} of \autoref{eq:VH2} if it satisfies the equation 
\begin{equation}\label{eq:VH2v}
\boxed{
\Curl \Curl  \widehat\bE_\ve (\br;\omega) - \kappa_\ve^2(\omega) (\mathbb{I} + \chi(\br;\omega)) \widehat\bE_\ve(\br;\omega) = 
- \frac{\i \omega-\ve}{\epsilon_0 c^2} \fouriercurrent_\ve(\br;\omega).}
\end{equation} 
\item
 We call $\widehat\bE_\ve$ a \emph{causal attenuating solution} of \autoref{eq:VH2} if $\bE_\ve$ (the inverse Fourier-transform of 
 $\widehat\bE_\ve$) is a \emph{causal} distribution.
 \end{itemize}
\end{definition}
In the following we show that $\widehat\bE_\ve$ approximates the retarded solution of the vector-Helmholtz equation \autoref{eq:retarded} 
in a distributional sense:
\begin{theorem}
\label{thm:convlaplace}
For every $\ve > 0$, let $\widehat\bE_\ve$ be the solution of \autoref{eq:VH2v}, the causal attenuating wave equation, and 
let $\widehat\bE$ be the retarded solution of \autoref{eq:VH2}, which is given by \autoref{eq:retarded}, then 
\begin{equation} \label{eq:Sp}
 \widehat\bE_\ve \xrightarrow[\ve \to 0]{\mathcal S'} \widehat\bE.
\end{equation}
\end{theorem}
\begin{proof}
We define for all $t \in \R$, $\br \in \R^3$
\begin{equation} \label{eq:ve}
   \bE_\ve(\br;t) = \alpha_\ve(t) \bE(\br;t) \text{ where } \alpha_\ve(t) :=\e^{-\ve t}.
 \end{equation}
Because $\bE$ is causal, $\bE_\ve$ is a tempered distribution and since $\widehat{\bE}$ is a solution of \autoref{eq:VH2}, it follows that for all $\ve>0$, 
 $\widehat{\bE}_\ve$ is a solution of \autoref{eq:VH2v} and in particular it is also causal.
 We show that $\bE_\ve \xrightarrow[\ve \to 0]{\mathcal S'}\bE$ and because the Fourier transform (see \autoref{eq:FT}) 
 is a bounded operator on $\tempdist{\R^3 \times \R;\R^3}$
 (see \cite[Theorem 5.17]{Gru09}), the assertion, \autoref{eq:Sp}, then follows. 
 
To prove that $\bE_\ve \xrightarrow[\ve \to 0]{\mathcal S'}\bE$, we need to show that for all $\bphi \in \sspace{\R^3 \times \R;\R^3}$, 
$\scal{\bE_\ve}{\bphi} \to \scal{\bE}{\bphi}.$ Noting that $\scal{\bE_\ve}{\bphi} = \scal{\bE}{\alpha_\ve \bphi}$, we therefore need to show that 
$\scal{\bE}{\bphi-\alpha_\ve \bphi} \to 0$. \autoref{lem:causdist} shows that, because $\bE$ is causal, one can write
 \[\left \vert \scal{\bE}{\bphi-\alpha_\ve \bphi} \right \vert \leqslant C \sup_{\alpha \leqslant k, \beta \leqslant l} \sup_{t \geqslant -1} |t^\alpha \partial_t^\beta (\bphi-\alpha_\ve \bphi)(t)|.\]
Now, note that for all $\beta \in \N_0$ and all $t \in \R$,
 \begin{align*}
  \partial_t^\beta \left[(\e^{-\ve t}-1)\bphi(t)\right] &= -\partial_t^\beta \bphi(t)+ \sum_{i=0}^\beta \binom{\beta}{i} (-\ve)^i \e^{-\ve t} \partial_t^{\beta-i} \bphi(t) \\  
  &= (\e^{-\ve t}-1) \partial_t^\beta \bphi(t) + \ve A_\ve(t) \e^{-\ve t}.
 \end{align*}
 where $A_\ve$ is a polynomial (with coefficients uniformly bounded with $\ve$)
 in the derivatives of $\bphi$ up to the order $\beta-1$. Since the derivatives of $\bphi$ are Schwartz functions, $\sup_{t\gs -1} |t^\alpha A_\ve(t)|$ is then uniformly bounded in $\ve$, which implies
 \[\lim_{\ve \to 0} \sup_{t\gs -1} \left \vert t^\alpha \ve A_\ve(t)e^{-\ve t} \right \vert = 0.\]
 Now, $B(t) := \partial_t^\beta \bphi$ is also a Schwartz function, which means that for every $k \in \N_0$ there exists $C_k$ such 
 that $\sup_t |(t^{k+2}+1) B(t)| \leqslant C_k$. It then follows that for all $t \gs -1$, 
\[ \left \vert t^\alpha (\e^{-\ve t}-1) B(t) \right \vert = \left \vert \frac{t^\alpha}{t^{\alpha + 2} +1}  (\e^{-\ve t}-1) (t^{\alpha +2} +1) B(t) \right \vert \leqslant C_\alpha \sup_{t \gs -1} \left \vert \frac{t^\alpha (\e^{-\ve t} -1)}{t^{\alpha +2}+1} \right \vert,  \]
where the last supremum converges to zero with $\ve \to 0$.
Therefore we conclude that 
\[ \lim_{\ve \to 0} \sup_{t \gs -1} \left \vert t^\alpha \partial_t^\beta ((\e^{-\ve t}-1)\bphi(t)) \right \vert = 0,\]
which means $\scal{\bE}{\bphi-\alpha_\ve \bphi} \to 0$.
\end{proof}

\subsection{Dipoles}
The emission of fluorescent dyes will be modeled as dipoles.
\begin{definition}[Emitting Dipole] \label{de:dens}
 An \emph{emitting dipole} is a vector $\bpsi =\begin{pmatrix} \Psi_1 & \Psi_2 & \Psi_3 \end{pmatrix}^T$, which is associated to a point $\rpsi$ in space;  
 $\abs{\bpsi}$ is called \emph{charge intensity} 
 and  $\frac{\bpsi}{\abs{\bpsi}}$ can be represented in spherical coordinates $(\theta_m, \varphi_m)\in \sphere$. Both notations are used synonymously and called the \emph{orientation} of the emitting dipole. That is 
 \begin{equation} \label{eq:emit}
 \bpsi = \begin{pmatrix} \Psi_1 \\ \Psi_2 \\ \Psi_3 \end{pmatrix} =
               \begin{pmatrix} \abs{\bpsi}  \sin(\theta_m) \cos(\varphi_m) \\
                                \abs{\bpsi} \sin(\theta_m) \sin(\varphi_m) \\
               \abs{\bpsi} \cos(\theta_m)
              \end{pmatrix}.
 \end{equation} 
 The \emph{limiting} density of a dipole at position $\begin{pmatrix}0\\ 0 \\ \rpsi \end{pmatrix}\in \R^3$ 
 is defined as a generalized function in space
 \begin{equation} \label{eq:dens}
  \widehat{\rho}(\br) := \abs{\bpsi} \lim_{s \to 0^+} 
  \frac{\delta_{\rpsi + s \frac{\bpsi}{\abs{\bpsi}}}(\br) - \delta_{\rpsi - s \frac{\bpsi}{\abs{\bpsi}}}(\br)}{2s} \text{ for all }
  \br \in \R^3.
 \end{equation}
 That is, in mathematical terms, the dipole charge is the directional derivative of a three-dimensional $\delta$-distribution in direction 
 $\frac{\bpsi}{\abs{\bpsi}}$.
 Moreover, we denote by 
 \begin{equation} \label{eq:J_formula}
  \fouriercurrent (\br;\omega):= -\i \omega \bpsi \delta(\br-\rpsi)
 \end{equation}
 the \emph{dipole current} (which is frequency dependent).
 
 In what follows we assume that the emitting dipol is a unit-vector (that is $\abs{\Psi}=1$), which simplifies the considerations and 
 the notation.
\end{definition}

\begin{lemma}
 Let $\fouriercurrent$ and $\widehat{\rho}$ be as defined in \autoref{eq:J_formula} and \autoref{eq:dens}, respectively and 
 satisfy \autoref{eq:FJ}.
 Then
 \begin{equation} \label{eq:bR}
   \widehat{\bf R}(\br;\omega)
   := \frac{\i \omega}{c^2} \fouriercurrent(\br;\omega) + \Grad \widehat{\rho} (\br;\omega)
 \end{equation}
 satisfies 
 \begin{equation} \label{eq:rhs}\begin{aligned}
  ~ & \widehat{\bf R}(\br;\omega)\\
   = &
   \frac{\omega^2 \bpsi}{c^2} \delta(\br-\rpsi) + \begin{pmatrix} 
  \psi_1 \eindelta''(x_1) \eindelta(x_2) \eindelta(x_3)
  + 
  \psi_2 \eindelta'(x_1) \eindelta'(x_2) \eindelta(x_3)  + 
  \psi_3 \eindelta'(x_1) \eindelta(x_2) \eindelta'(x_3)
  \\
  \psi_1 \eindelta'(x_1) \eindelta'(x_2) \eindelta(x_3) + 
  \psi_2 \eindelta(x_1) \eindelta''(x_2) \eindelta(x_3) + 
  \psi_3 \eindelta(x_1) \eindelta'(x_2) \eindelta'(x_3)
  \\
  \psi_1 \eindelta'(x_1) \eindelta(x_2) \eindelta'(x_3) + 
  \psi_2 \eindelta(x_1) \eindelta'(x_2) \eindelta'(x_3) + 
  \psi_3 \eindelta(x_1) \eindelta(x_2) \eindelta''(x_3)
  \end{pmatrix},
  \end{aligned}
 \end{equation}
 where $(\bx,x_3)^T := \br-\rpsi$, where $\rpsi$ denotes the dipole position. 
\end{lemma}
\begin{proof}
 Taking into account that the $3$-dimensional $\delta$-distribution can be written as  
 \begin{equation*}
  \delta_{\rpsi \pm s \bpsi}(\br) = \prod_{j = 1}^3 \eindelta_{(\rpsi)_j \pm s \psi_j}(\br_j)= 
  \prod_{j = 1}^3 \eindelta (\br_j - (\rpsi)_j \mp s \psi_j)
 \end{equation*}
 we find 
 \begin{equation*}
  \widehat{\rho}(\br;\omega) = \sum_{i=1}^3 \psi_i (\eindelta_{(\rpsi)_i})'(\br_i) \prod_{j \neq i} \eindelta_{(\rpsi)_j}(\br_j) = 
  \sum_{i=1}^3 \psi_i \eindelta'((\br-\rpsi)_i) \prod_{j \neq i} \eindelta((\br-\rpsi)_j)
 \end{equation*}
 and we get 
 \begin{equation} \label{eq:rho}
  \nabla_{\br} \widehat{\rho}(\br;\omega) = \begin{pmatrix} 
  \psi_1 \eindelta''(x_1) \eindelta(x_2) \eindelta(x_3)
  + 
  \psi_2 \eindelta'(x_1) \eindelta'(x_2) \eindelta(x_3)  + 
  \psi_3 \eindelta'(x_1) \eindelta(x_2) \eindelta'(x_3)
  \\
  \psi_1 \eindelta'(x_1) \eindelta'(x_2) \eindelta(x_3) + 
  \psi_2 \eindelta(x_1) \eindelta''(x_2) \eindelta(x_3) + 
  \psi_3 \eindelta(x_1) \eindelta'(x_2) \eindelta'(x_3)
  \\
  \psi_1 \eindelta'(x_1) \eindelta(x_2) \eindelta'(x_3) + 
  \psi_2 \eindelta(x_1) \eindelta'(x_2) \eindelta'(x_3) + 
  \psi_3 \eindelta(x_1) \eindelta(x_2) \eindelta''(x_3)
  \end{pmatrix}.
 \end{equation}
 On the other hand   
 \begin{equation*}
  \begin{aligned}
   - \nabla \cdot \fouriercurrent (\br;\omega) &= \i \omega \nabla \cdot (\bpsi \delta(\br-\rpsi)) \\
   &= \i \omega \sum_{i=1}^3 \psi_i \eindelta'((\br-\rpsi)_i) \prod_{j \neq i} \eindelta((\br-\rpsi)_j) = \i \omega \widehat{\rho}(\br),
  \end{aligned}
 \end{equation*}
 and thus \autoref{eq:FJ} is satisfied.

Moreover, using \autoref{eq:rho} in \autoref{eq:FJ} gives \autoref{eq:rhs}.
\end{proof}
In the following we calculate the solution $\widehat\bE$ of \autoref{modeleq1_soln_help}, similar as in \cite{ForWeb84}.

The following lemma and its proof are based on \cite{ForWeb84}.

\begin{lemma}\label{lemmaE}
Let $\widehat\bE$ as in \autoref{eq:retarded} be the retarded solution of \autoref{eq:VH2v} 
at fixed frequency $\omega$. 
In what follows we omit therefore the dependency of $\omega$ and write 
$\widehat\bE(\br) := \widehat\bE(\br;\omega)$.

Moreover, let the medium be isotropic, non magnetizable, homogeneous and have a linear polarization response 
(that is, $\chi \equiv 0$).

As above we assume that a dipole $\Psi \in \R^3$ is located at position $\rpsi = (0,0,\rpsi)^T$.

Moreover, for all $\ve > 0$ let $\kappa_\ve$ as in \autoref{eq:kappa} and we define for fixed $k_1, k_2 \in \R$
\begin{equation} \label{eq:qk1}
 q : = \lim_{\ve \to 0^+} q_\ve  \text{ where } q_\ve:= a_\ve+\i b_\ve := \sqrt{\kappa_\ve^2 - k_1^2 - k_2^2} \quad \text{ with } \quad b_\ve > 0  
\end{equation}
(that is $q_\ve$ is the complex root with positive imaginary part).  
Let now $\br \in \R^3$ be such that $r_3-\rpsi \geq 0$, then 
\begin{equation} \label{eq:EFW}
 \widehat\bE(\br) = -\frac{1}{4\pi}\frac{1}{\epsilon_0} 
 \mathcal F^{-1}_{12} \left[  (k_1,k_2) \mapsto \psi_3 \be_3\eindelta(r_3 - \rpsi) + \frac{\i \e^{\i q (r_3-\rpsi)}}{2q} 
 (\bpsi \times \bk_q) \times \bk_q\right](r_1,r_2).
\end{equation}
\end{lemma}
\begin{proof} 
First let $\ve > 0$, and we prove an identity of the form \autoref{eq:EFW} for $\widehat{\bE}_\ve$.
We note that 
\begin{equation*}
\mathcal{F} [\Grad \times \Grad \times \widehat\bE_\ve](\bk) = - (\mathcal{F} [ \widehat\bE_\ve ](\bk) \times \bk ) \times \bk \text{ for all } \bk \in \R^3.
\end{equation*}
Thus from \autoref{eq:VH2} with $\chi \equiv 0$ it follows by applying the $k$-transform, and by using \autoref{eq:J_formula}, 
\autoref{eq:kappa} and \autoref{eq:dist3} that 
\begin{equation} \label{eq:h1}
- (\mathcal{F}[\widehat\bE_\ve](\bk) \times \bk ) \times \bk - \kappa_\ve^2 \mathcal{F}[\widehat\bE_\ve](\bk) = -\frac{ \i \omega -\ve}{\epsilon_0 c^2} \mathcal{F}[\widehat{\bf J}_\ve](\bk) 
= -\frac{\kappa_\ve^2}{(2\pi)^{\frac{3}{2}}\epsilon_0} \bpsi \e^{- \i k_3 \rpsi} .
\end{equation}
Elementary calculation rules for $\times$ provide that 
\begin{equation}\label{eq:doubleprod}
 (\bv \times \bk) \times \bk = (\bk \cdot  \bv) \bk - \abs{\bk}^2 \bv \text{ for all } \bv, \bk \in \R^3, 
\end{equation} 
which, by application to $\bv = \mathcal{F}[\widehat\bE_\ve](\bk)$ and $\bv = \bpsi$, respectively, shows that 
\begin{equation} \label{eq:ktransformEandpsi}
\begin{aligned}
\abs{\bk^2} \mathcal{F}[\widehat\bE_\ve](\bk)  &= 
- (\mathcal{F}[\widehat\bE_\ve](\bk) \times \bk) \times \bk + (\bk \cdot  \mathcal{F}[\widehat\bE_\ve](\bk)) \bk 
\text{ and }\\
\abs{\bk^2}\bpsi  &= - (\bpsi \times \bk) \times \bk + (\bk \cdot  \bpsi) \bk.
\end{aligned}
\end{equation}
Therefore, by multiplying \autoref{eq:h1} with $\abs{\bk^2}$ and using \autoref{eq:ktransformEandpsi}, it follows that
\begin{equation} \label{eq:h2}
\begin{aligned}
 ~ & (\kappa_\ve^2 - \abs{\bk^2})(\mathcal{F}[\widehat\bE_\ve](\bk) \times \bk) \times \bk 
     - \kappa_\ve^2 ( \bk \cdot \mathcal{F}[\widehat\bE_\ve](\bk)) \bk \\
 = & -\frac{\kappa_\ve^2}{(2\pi)^{\frac{3}{2}}\epsilon_0}  \e^{- \i k_3 \rpsi} \left[
     - (\bpsi \times \bk) \times \bk + (\bk \cdot  \bpsi) \bk \right].
\end{aligned}
\end{equation}

Since $\bk$ and  $ (\bv \times \bk) \times \bk$ are orthogonal, it follows from \autoref{eq:h2} that:
\begin{equation*}
 \begin{aligned}
 (\mathcal{F}[\widehat\bE_\ve](\bk) \cdot \bk)\bk &= 
 \frac{1}{(2\pi)^{\frac{3}{2}}}\frac{1}{\epsilon_0}  \e^{-\i k_3 \rpsi} (\bpsi \cdot \bk)\bk = 
 \frac{1}{(2\pi)^{\frac{3}{2}}}\frac{1}{\epsilon_0}  \e^{-\i k_3 \rpsi} \left( \abs{\bk}^2 \bpsi + (\bpsi \times \bk ) \times \bk \right),\\
 (\abs{\bk}^2 - \kappa_\ve^2) (\mathcal{F}[\widehat\bE_\ve](\bk) \times \bk ) \times \bk &= 
   -\frac{\kappa_\ve^2}{(2\pi)^{\frac{3}{2}}\epsilon_0} \e^{-\i k_3 \rpsi}
  (\bpsi \times \bk) \times \bk .
 \end{aligned}
\end{equation*}
Inserting these two identities into \autoref{eq:ktransformEandpsi} and noting that since $\kappa_\ve$ is not real, one can divide by $|\bk|^2 - \kappa_\ve$, yields
\begin{equation*}
\begin{aligned}
 \abs{\bk}^2 \mathcal{F}[\widehat\bE_\ve](\bk) &= 
\frac{1}{(2\pi)^{\frac{3}{2}}}\frac{1}{\epsilon_0} \e^{-\i k_3 \rpsi} 
\left(  \abs{\bk}^2  \bpsi + (\bpsi \times \bk ) \times \bk + \frac{\kappa_\ve^2}{\abs{\bk}^2- \kappa_\ve^2}(\bpsi \times \bk ) \times \bk\right)\\&= 
\frac{1}{(2\pi)^{\frac{3}{2}}}\frac{1}{\epsilon_0}  \e^{-\i k_3 \rpsi} 
\left( \abs{\bk}^2 \bpsi + \frac{\abs{\bk}^2}{\abs{\bk}^2 - \kappa_\ve^2}(\bpsi \times \bk ) \times \bk\right),
\end{aligned}
\end{equation*}
such that 
\begin{equation*}
\begin{aligned}
\mathcal{F}[\widehat\bE_\ve](\bk) &= 
\frac{1}{(2\pi)^{\frac{3}{2}}}\frac{1}{\epsilon_0}  \e^{-\i k_3 \rpsi} 
\left( \bpsi + \frac{(\bpsi \times \bk ) \times \bk}{\abs{\bk}^2 - \kappa_\ve^2} \right).
\end{aligned}
\end{equation*}
Therefore 
\begin{equation*}
\widehat\bE_\ve({\bf r}) = \frac{1}{(2\pi)^{\frac{3}{2}}}\frac{1}{\epsilon_0} 
\mathcal{F}_{12}^{-1} \left[ \mathcal{F}_3^{-1} 
\left[ \left(\bpsi  +   \frac{(\bpsi \times \bk) \times \bk}{\abs{\bk}^2 - \kappa_\ve^2} \right)  \e^{-\i k_3 \rpsi}\right](r_3)\right] (r_1,r_2).
\end{equation*}

In order to prove \autoref{eq:EFW} for $\widehat{\bE}_\ve$, it remains to show that 
\begin{equation*}
\frac{1}{\sqrt{2\pi}}
\mathcal{F}_3^{-1} 
\left[k_3 \to  \left(\bpsi  +   \frac{(\bpsi \times \bk) \times \bk}{\abs{\bk}^2 - \kappa_\ve^2} \right)\e^{-\i k_3 \rpsi}\right](r_3) =
 \psi_3 \be_3\eindelta(r_3 - \rpsi) + \frac{\i \e^{\i q_\ve  (r_3-\rpsi)}}{ 2q_\ve  } (\bpsi \times \bk_{q_\ve}) \times \bk_{q_\ve},
\end{equation*}
which is done by standard, but quite lengthy computations, which are presented in \autoref{app:FT}.

Now, we consider $\ve \to 0$. \autoref{thm:convlaplace} combined with the continuity of the inverse Fourier transform 
$\mathcal F_{12}^{-1}$ in $\mathcal S'(\R^2,\R^2)$
which implies that 
\begin{equation*}
\widehat{\bE}(\br) = -\frac{1}{4\pi}\frac{1}{\epsilon_0} \mathcal F_{12}^{-1} \left[(k_1,k_2) \to \psi_3 \be_3 \eindelta(r_3 - \rpsi) +\lim_{\ve \to 0}  \frac{\i \e^{\i q_\ve  (r_3-\rpsi)}}{ 2q_\ve  } (\bpsi \times \bk_{q_\ve}) \times \bk_{q_\ve} \right](r_1,r_2).
\end{equation*}
To prove the assertion, we simply need to check that, in $\mathcal S'$
\[ \lim_{\ve \to 0} \frac{\i \e^{\i q_\ve  (r_3-\rpsi)}}{ 2q_\ve  } (\bpsi \times \bk_{q_\ve}) \times \bk_{q_\ve}=    \frac{\i \e^{\i q (r_3-\rpsi)}}{2q} 
 (\bpsi \times \bk_q) \times \bk_q. \]
These two quantities being $L^1_{\mathrm{loc}}$ functions, it is enough to show that the limit holds in $L^1_{\mathrm{loc}} (\R \times (\R^2\times\R))$. The $L^1_{\mathrm{loc}}$ convergence is then obtained noticing that 
\[ \e^{\i q_\ve  (r_3-\rpsi)}(\bpsi \times \bk_{q_\ve}) \times \bk_{q_\ve} -  \e^{\i q (r_3-\rpsi)}(\bpsi \times \bk_q) \times \bk_q \xrightarrow{L^\infty} 0\]
and that 
\[\frac{1}{q_\ve}- \frac{1}{q} = 
\frac{\kappa_\ve^2 -  \kappa^2}{(\kappa_\ve^2 - k_1^2 -k_2^2)\sqrt{\kappa^2 - k_1^2 -k_2^2} +(\kappa^2 - k_1^2 -k_2^2)\sqrt{\kappa_\ve^2 - k_1^2 -k_2^2}}  
\]
converges to zero in $L^1_{\mathrm{loc}}$. Note that this would imply only a convergence in $\mathcal D'$, but the two functions are actually uniformly $L^\infty$ outside the compact set $\{k_1^2 + k_2^2 \gs |\kappa|^2+1\}$, so the convergence holds in $\mathcal S'$ as well.
\end{proof}

 Moreover, we make the assumption that the dipole can be \emph{rotating}.
 \begin{definition}[Rotating Dipole]
 The emitting dipole is considered wobbling uniformly distributed around the dipole orientation 
 $\frac{\bpsi_m}{|\bpsi_m|} = (\theta_m,\varphi_m) \in \sphere$ in a cone of semi-angle $\alpha_m$.  
 Assuming a dipole-emission from a oscillating source we get after averaging an \emph{source} represented as the indicator function 
 \begin{equation} \label{eq:dye}
  \mathds{1}_m = \frac{1}{\abs{C(\bpsi_m,\alpha_m)}} \mathds{1}_{C(\bpsi_m,\alpha_m)}, 
 \end{equation}
 where 
 \begin{equation} \label{eq:wob}
 C(\bpsi_m,\alpha_m) = \set{\tau \bphi \in \sphere : \abs{\measuredangle \bphi \bpsi_m} \leq \alpha_m, 0 \leq \tau \leq |\bpsi_m|}.
\end{equation}
Note that $\abs{C(\bpsi_m,\alpha_m)} = \frac{1}{3} \pi |\bpsi_m|^3 \tan^2(\alpha_m)$.
Taking into account \autoref{eq:J_formula} and \autoref{eq:dens} the according charge density and current of dye $m$ are given by 
\begin{equation} \label{eq:J_wobbel}
\begin{aligned}
 \fouriercurrent_{m}^o(\br;\omega) &= -\i \omega \mathds{1}_m, \quad \widehat{\rho}_{m}^o(\br;\omega) = \frac{\i}{\omega} \Div \fouriercurrent_{m}^o(\br;\omega) \text{ and } \\
\widehat{\bf R}_{m}^o(\br;\omega) &:= \frac{\i \omega}{c^2} \fouriercurrent_{m}^o(\br;\omega) + \Grad \widehat{\rho}_{m}^o (\br;\omega).
\end{aligned}
\end{equation}

\end{definition}

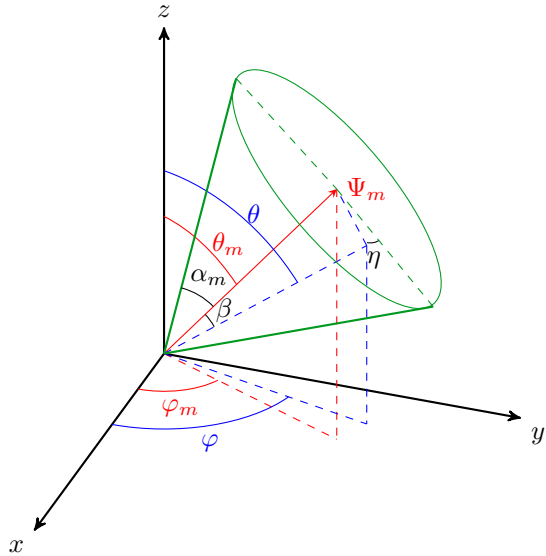
\begin{figure}[h] 
\centering

\tdplotsetmaincoords{60}{110}
\pgfmathsetmacro{\rvec}{1.0}
\pgfmathsetmacro{\thetavec}{40}
\pgfmathsetmacro{\phivec}{65}
\pgfmathsetmacro{\srvec}{0.85}
\pgfmathsetmacro{\sthetavec}{50}
\pgfmathsetmacro{\sphivec}{75}
\pgfmathsetmacro{\pthreervec}{0.85}
\pgfmathsetmacro{\pthreethetavec}{66.83}
\pgfmathsetmacro{\pthreephivec}{85}
\pgfmathsetmacro{\prvec}{1.1}
\pgfmathsetmacro{\pthetavec}{20.2}
\pgfmathsetmacro{\pphivec}{50}

\pgfmathsetmacro{\tilttheta}{-77}
\pgfmathsetmacro{\tiltphi}{30}
\pgfmathsetmacro{\tiltpsi}{55}

\begin{tikzpicture}[scale=5,tdplot_main_coords]
    \coordinate (O) at (0,0,0);
    \draw[thick,->] (0,0,0) -- (1,0,0) node[anchor=north east]{$x$};
    \draw[thick,->] (0,0,0) -- (0,1,0) node[anchor=north west]{$y$};
    \draw[thick,->] (0,0,0) -- (0,0,1) node[anchor=south]{$z$};
    
    \coordinate (ez) at (0,0,1);

    \tdplotsetcoord{P}{\rvec}{\thetavec}{\phivec}
    \draw[-stealth,color=red] (O) -- (P) node[ right ] {$\Psi_m$};
    \draw[dashed, color=red] (O) -- (Pxy);
    \draw[dashed, color=red] (P) -- (Pxy);
    \tdplotdrawarc[color=red]{(O)}{0.2}{0}{\phivec}{anchor=north,color=red}{$\varphi_m$}
    \tdplotsetthetaplanecoords{\phivec}
    \tdplotdrawarc[tdplot_rotated_coords,color=red]{(0,0,0)}{0.42}{0}
        {\thetavec}{anchor=west, color=red}{$\theta_m$}
    
    \tdplotsetcoord{P1}{\srvec}{\sthetavec}{\sphivec}
    \draw[dashed, color=blue] (O) -- (P1) node[above right] {};
    \draw[dashed,  color=blue] (O) -- (P1xy);
    \draw[dashed, color=blue] (P1) -- (P1xy);
    \tdplotdrawarc[color=blue]{(O)}{0.4}{0}{\sphivec}{anchor=north,color=blue}{$\varphi$}
    \tdplotsetthetaplanecoords{\sphivec}
    \tdplotdrawarc[tdplot_rotated_coords,color=blue]{(0,0,0)}{0.56}{0.1}%
        {\sthetavec}{anchor=west,color=blue}{$\theta$}    
    
    
    \tdplotsetcoord{P2}{\prvec}{\pthetavec}{\pphivec}
    \draw[thick,color={blue!30!black!40!green}] (O) -- (P2) node[above right] {};
    
    \tdplotsetcoord{Pmid}{1}{49}{66}
    \tdplotdrawarc{(Pmid)}{0.025}{0.02}{85}{anchor=north}{$\eta$}

    \tdplotsetcoord{P3}{\pthreervec}{\pthreethetavec}{\pthreephivec}
    \draw[thick,color={blue!30!black!40!green}] (O) -- (P3) node[above right] {};
    
    \draw[dashed, color={blue!30!black!40!green}] (P2) -- (P);
    \draw[dashed, color={blue!30!black!40!green}] (P3) -- (P);
    \draw[dashed, color={blue}] (P1) -- (P);
    
    \pic[draw,-,angle radius=.75cm,angle eccentricity=1.3,"$\beta$"] {angle=P1--O--P};
    \pic[draw,-,angle radius=.9cm,angle eccentricity=1.3,"$\alpha_m$"] {angle=P--O--P2};
    
    \tdplotsetthetaplanecoords{20}
    \tikzset{xyplane/.estyle={cm={
      cos(\tilttheta),sin(\tilttheta)*sin(\tiltphi),sin(\tiltpsi)*sin(\tilttheta),
      cos(\tiltpsi)*cos(\tiltphi)-sin(\tiltpsi)*cos(\tilttheta)*sin(\tiltphi),(P)}
      }}
    \draw[blue!30!black!40!green,xyplane] (P) circle (0.6); 
\end{tikzpicture}
 \caption{The axis of cone has angular coordinates $\theta_m$ and $\varphi_m$ in the coordinate system.  A general orientation within the cone has coordinates $\theta$ and $\varphi$ in the coordinate system, and axial coordinate $\beta$, and azimuthal coordinate $\eta$ with respect to the cone axis. The outer limit of motion in the cone is given by $\beta = \alpha_m$.}
	\label{fig:Cone}
\end{figure}

\section{The Forward Problem} 
In the following we present mathematical models describing the emission and propagation of light caused by dyes, which are 
exposed to strong laser light illumination. See \autoref{fig:cell} for a schematic representation of the experiment.
In single molecule localization microscopy two-dimensional images are recorded after exposing the probe subsequently 
to strong laser illuminations, such that the dyes appear in dark (``off'') and light (``on'') state. This allows to separate the fluorescent emission of individual dyes in time, allowing for high resolution images. In order to minimize the notational effort we consider 
recording of a single image frame first. 
The mathematical model of consecutive recordings of multiple frames is analogous and requires one additional parameter representing 
numbering of frames (a virtual time).

\begin{figure}[h] 
	\centering
  \pgfmathsetmacro{\cellRadius}{3}
  \pgfmathsetmacro{\cellSpacing}{5}
  \pgfmathsetmacro{\delradius}{\cellRadius/\cellSpacing}
\begin{tikzpicture}[baseline=(current bounding box.north)]
  \path
    (-\cellRadius, 0) coordinate (A) 
    -- coordinate (M)
    (\cellRadius, 0) coordinate (B)
    (M) +(60:\cellRadius) coordinate (C)
    +(120:\cellRadius) coordinate (D)
  ;

  \coordinate (A1) at (-\cellRadius-0.5, 0);
  \coordinate (A2) at (-\cellRadius-0.5, -0.5);
  \coordinate (B1) at (\cellRadius+0.5, 0);
  \coordinate (B2) at (\cellRadius+0.5, -0.5);
  
  \coordinate (D) at (\cellRadius+0.5, -1.5);
  
%
  \draw (A1) -- (A2)--(B2) --(B1) --(A1);
 
  \foreach \i in {1,...,9}
  {
      \pgfmathtruncatemacro{\rr}{ -\cellRadius + \i * \delradius}; 
      \filldraw[black] (\rr,0.05) circle(1.4pt);
      \pgfmathsetmacro{\ar}{ \rr + 0.5  * \delradius};
      \draw[thick, ->,color={blue}] (\ar,-2) -- (\ar,-1); 
  }   
  
  \node at (D) [below=.1em] {$\textcolor{blue}{\text{Illumination}}$};
  
\end{tikzpicture}	
\caption{\label{fig:cell} Illustration of the experiment: Biomolecular structures are placed on the glass surface at position $r_3^\bpsi \ls 0$ and illuminated from the bottom.
 The glass plate has a thickness $r_3^\bpsi$. 
 }
\end{figure}
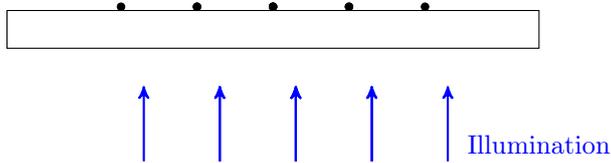

In the following we state a series of assumptions, which are used throughout the remainder of the paper: 
\begin{assumption}[Medium, Monochromatic Source and Response] \label{ass:simple2} 
 In the following we assume that
 \begin{itemize}
  \item The incident light is a \emph{monochromatic} \emph{plane wave} of frequency $\omega_{inc}$ and orientation $\bv$.
  \item The medium is assumed to be isotropic, non magnetizable, homogeneous and has a linear polarization response.
  \item Moreover, we assume that a dye can be modeled as an \emph{absorbing dipole} $\Psi_a$, which emits monochromatic waves of frequency 
   $\omega \neq \omega_{inc}$ resulting in an \emph{emitting dipole}
    \begin{equation} \label{eq:plane_wave}
     \Psi = (\bv \cdot \Psi_a) \bv.
    \end{equation}
    Indeed what we will measure is the electric field at frequency $\omega$, which is not affected by the incident field 
    at frequency $\omega_{inc}$. As a consequence we only have to consider the electric field at the frequency $\omega \in \R$.   
 \item In what follows we assume that the considered dipole $\bpsi =\vp$ is located at position 
  $\br^\bpsi = \begin{pmatrix} 0 & 0 & \rpsi \end{pmatrix}^T$ with $\rpsi \ls 0$. Unless stated otherwise 
  $\br \in \R^3$ with $r_3 > 0$. 
  The sign assumptions on $r_3$ and $\rpsi$ are in accordance with the experiment: the object is assumed left of the lens 
  system (see \autoref{fig:Scheme}) and $\br$ is a point of the measurement system.
 \item The dyes absorb light, which can result in fluoresence emission. We describe the states of an absorbing dye 
  with index $m$ via a time indicator function:
  The \emph{on-off} indicator 
  \begin{equation} \label{ass:simple}
    \mathbf{I}_m \in \set{0,1},
  \end{equation}
  tells us whether the $m$-th dye is an emitting state or not. 
 \end{itemize}
 \end{assumption}
 
The complete experimental setup of the optical experiment of single molecule localization microscopy is represented in \autoref{fig:Scheme}. For the mathematical 
modeling we are considering the propagation light at different locations of the optical system. The dyes are considered at 
positions $\br^{\bpsi_m}$ with $r_3^{\bpsi_m} \ls 0$ and the \emph{focal plane} (which contains the focal point of the objective) corresponds to the bottom of the glass plate, which is not mathematically modeled, that is the focal plane is at position $r_3=0$. 
Note that in particular that the dipole \emph{is not located at the focal plane}, unless if $r_3^{\bpsi_m} = 0$.
For the sake of simplicity of presentation we consider only a single dye, 
and leave the subscript $m$ whenever appropriate.

The mathematical modeling of the experimental setup follows \cite{Axe12}, however it is adapted to our notation: 
\begin{itemize}
 \item In \autoref{ss:ff} we describe the propagation of the electric field in the medium, that is from the bottom of the cell 
       (the assumption is that only molecule labeled with a dye at the bottom of the cell emit light) 
       up to the objective (see \autoref{fig:Scheme}). 
       This domain will be denoted by $\Omega$. 
       Since the objective is far away from the molecule (relative to the size of the molecule) the 
       electric field can be approximated well by its far field, which is calculated below. 
       The 3-dimensional $k$-transformed coordinate system is denoted by $\bk \in \R^3$ (see \autoref{ss:ff}).
       
 \item In \autoref{ss:cyl} we present in mathematical terms the propagation of the emitted light when it passes 
       through the objective; that is after passing through the medium. In fact the light rays are aligned parallel by the objective in $r_3$ direction. 
       The objective has a focal length $\fobj$ and it is positioned orthogonal to the $r_3$ axis with left distance to the 
       focal plane (glass plate) $r_3^\obj = \fobj$. 
       Indeed the lens system is complicated and a detailed mathematical modeling is not possible. A simplified model 
       assumes that the objective is big compared to the wavelength, such that the intensity law of \emph{Geometric Optics} applies 
       (see \autoref{fig:Area_dA12} and \cite{BorWol99}), and phase shifts due to the curvature of the lenses can be neglected. 
              
 \item In \autoref{ss:between} we calculate the propagation of the light after passing through the back focal plane of the objective, that is in between 
       $r_3^{\texttt{bfp}} + \ttd_\obj$ and $r_3^{\texttt{tl}} - \texttt{d}_0 $, 
       from knowledge of the field at the plane with third coordinate $r_3^{\texttt{bfp}}$. 
       Here $\texttt{d}_0$ denotes the maximal width of the lens (see \autoref{fig:Scheme}).
       This is achieved by solving the Helmholtz equation in air between the back focal plane of the objective and the 
       incident plane of the tube lens.
       
 \item We assume that the lens is a circular \emph{tube lens} with maximal thickness $\texttt{d}_0$. The thickness is described as a function $\texttt{d}$.
       Moreover, we assume that the lens has a \emph{focal length} $\flens$ and that its pupil function is given by $P_\L:\R^2 \to \R$,  
       \begin{equation} \label{eq:pupil}
        P_\L(\bx) = \left\{
        \begin{array}{lcl}
        1 & \text{ for } & \abs{\bx} \leq R \\
        0 & \text{ for} &\abs{\bx} > R\\
        \end{array} 
        \right. .
       \end{equation} 
       The lens is assumed to be \emph{converging}, such that the \emph{paraxial approximation} holds, that is we can assume that the wave vector of the wave is almost aligned with the optical axis \cite[Sec. 4.2.3]{Goo05}. The adequate formulas are derived in 
       \autoref{ss:tube}.
 \item Finally the light is bundled to the image plane, which provides an image described by coordinates 
       $\bx_\focus \in \R^2$ (see \autoref{ss:final}).
\end{itemize}

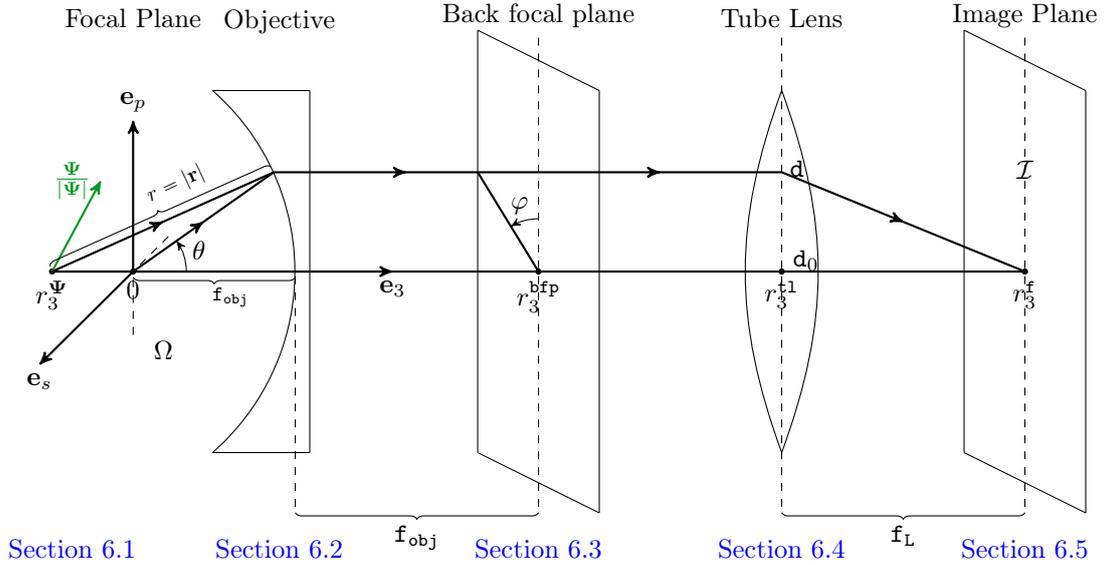
\begin{figure}[ht] 
\centering

  \pgfmathsetmacro{\lensRadius}{4}
  \pgfmathsetmacro{\lensHeight}{3}
  \pgfmathsetmacro{\startAngle}{asin(\lensHeight/\lensRadius)}
  \pgfmathsetmacro{\lensCenter}{cos(\startAngle)*\lensRadius}
  \pgfmathsetmacro{\ptAngle}{\startAngle/2}
  \pgfmathsetmacro{\xptlens}{cos(\ptAngle)*\lensRadius}
  \pgfmathsetmacro{\yptlens}{sin(\ptAngle)*\lensRadius} 
  
  \pgfmathsetmacro{\tubelensRadius}{4.25}
  \pgfmathsetmacro{\tubelensHeight}{3.25}
  \pgfmathsetmacro{\tubestartAngle}{asin(\tubelensHeight/\tubelensRadius)}
  \pgfmathsetmacro{\lensCenter}{cos(\tubestartAngle)*\tubelensRadius}

\begin{tikzpicture}[scale=0.8]

   \coordinate (D) at (0,0,0);
   \filldraw[black] (D) circle(1.2pt);
   \coordinate[label=below:$r_3^{\bpsi}$] (dipole) at (D);
  
   \coordinate (Dip) at (\lensRadius/5,\lensHeight/2,0);  
   \draw[thick, ->,color={blue!30!black!40!green}] (D) -- (Dip) node[left] {$\frac{\bpsi}{\abs{\bpsi}}$};
  
   \coordinate[label=below:$r_3^{\texttt{f}}$] (rImage) at (4*\lensRadius,0,0);
   \filldraw[black] (rImage) circle(1.2pt);
   \draw[thick] (D) -- (rImage);

   \coordinate[label=below:$\be_3$] (E) at (1.4*\lensRadius,0,0);
   \draw[thick,->] (D) -- (E);
  
   \coordinate (O) at (\lensRadius/3,0,0);
   \filldraw[black] (O) circle(1.2pt);
   \coordinate[label=below:$0$] (origin) at (O);
   \coordinate[label=below:$\Omega$] (origin) at (\lensRadius/3+0.5,-1,0);
  
   \coordinate[label=above:$\be_p$] (Ep) at (\lensRadius/3,\lensHeight-0.5,0);
   \coordinate[label=below:$\be_s$] (Es) at (\lensRadius/3,0,\lensHeight+1); 
  
   \draw[dashed] (\lensRadius/3,-0.35*\lensHeight,0) -- (O);
   \draw[thick,->] (O) -- (\lensRadius/3,\lensHeight-0.5,0);
   \coordinate[label=above:$\text{Focal Plane}$] (fplane) at (\lensRadius/3,\lensHeight+0.9,0);
   \coordinate[label=above:\autoref{ss:ff}] (fplane) at (\lensRadius/3-1,\lensHeight-7.9,0);
   
   \draw[thick,->] (O) -- (\lensRadius/3,0,\lensHeight+1);
   \draw[dashed] (\lensRadius/3,0,-0.5*\lensHeight) -- (O);
  
   \coordinate (P) at (\xptlens,\yptlens);
  
   \tikzset{middlearrow/.style={
        decoration={markings,
            mark= at position 0.5 with {\arrow{#1}} ,
        },
        postaction={decorate}  }
    }
    
   \draw[thick,middlearrow={>}] (D) --  (P);
  
   \draw[thick,middlearrow={>}] (O) --  (P);

   \draw (\lensCenter-0.099,\lensHeight) arc[start angle=\startAngle,delta angle=-2*\startAngle,radius=\lensRadius];
  
   \coordinate[label=above:$\text{Objective}$] (obj) at (\lensCenter+1,\lensHeight+0.8);
   \coordinate[label=above:\autoref{ss:cyl}] (objSec) at (\lensCenter+1,\lensHeight-7.9,0);
   \draw (\lensCenter-0.099,\lensHeight) -- (\lensCenter+1.5,\lensHeight);
   \draw (\lensCenter-0.099,-1*\lensHeight) -- (\lensCenter+1.5,-1*\lensHeight); 
   \draw(\lensCenter+1.5,\lensHeight) -- (\lensCenter+1.5,-1*\lensHeight); 
  
    \tikzset{
	position label/.style={
	  above = 3pt,
	  text height = 1.5ex,
	  text depth = 1ex
	},
      brace/.style={
	decoration={brace,mirror},
	decorate
      }
    }  
  \coordinate (dobj) at  (\lensRadius,0);
  \draw [brace,decoration={raise=0.5ex}] [brace]  (P.north) -- (dipole.north) node [position label, pos=0.4, rotate = 25,scale=0.8] {$r=\abs{\br}$};
  \draw [brace,decoration={raise=0.5ex}] [brace]  (O.south) -- (dobj.south) node [position label, below=0.1, pos=0.6, rotate = 0,scale=0.8] {$\fobj$};
  \pic[draw,->,angle radius=.7cm,angle eccentricity=1.3,"$\theta$"] {angle=E--O--P}; 


  \coordinate (ptBFP) at  (2*\lensRadius,\yptlens);
  \coordinate (pPlane) at  (2*\lensRadius-1,\yptlens);
  \coordinate[label=below:$r_3^{\texttt{bfp}}$] (rBFP) at  (2*\lensRadius,0);
  \filldraw[black] (rBFP) circle(1.2pt);
  \draw[thick,middlearrow={>}] (P) --  (ptBFP);

  \draw (2*\lensRadius -1,-1*\lensHeight) -- (2*\lensRadius -1,1*\lensHeight+1);
  \draw (2*\lensRadius -1,1*\lensHeight+1) -- (2*\lensRadius +1,1*\lensHeight);
  \draw (2*\lensRadius +1,-1*\lensHeight-1) -- (2*\lensRadius +1,1*\lensHeight);
  \draw (2*\lensRadius +1,-1*\lensHeight-1) -- (2*\lensRadius -1,-1*\lensHeight);
  
  \draw[dashed] (2*\lensRadius,-1.3*\lensHeight,0) -- (2*\lensRadius,1.3*\lensHeight,0);
  \coordinate[label=above:$\text{Back focal plane}$] (BFPtext) at (2*\lensRadius,1.3*\lensHeight,0);
  \coordinate[label=above:\autoref{ss:between}] (BFPtext) at (2*\lensRadius,\lensHeight-7.9,0);
  
  \draw[thick] (rBFP) -- (pPlane);
  \pic[draw,->,angle radius=.7cm,angle eccentricity=1.3,"$\varphi$"] {angle=ptBFP--rBFP--pPlane}; 
 
  \draw[dashed] (\lensRadius,-1.3*\lensHeight,0) -- (\lensRadius,0,0);
  \coordinate (obj1) at (1*\lensRadius,-1.3*\lensHeight,0);
  \coordinate (bfp1) at (2*\lensRadius,-1.3*\lensHeight,0);
  \tikzset{
    position label/.style={
       below = 3pt,
       text height = 1.5ex,
       text depth = 1ex
    },
   brace/.style={
     decoration={brace,mirror},
     decorate
    }
  }
  \draw [brace,decoration={raise=0.5ex}] [brace] (obj1.south) -- (bfp1.south) node [position label, pos=0.5, rotate = 0,scale=1.0] {$\fobj$};
  
  \coordinate (ptTube) at  (3*\lensRadius,\yptlens);
  \coordinate[label=below:$r_3^{\texttt{tl}}$] (rtl) at  (3*\lensRadius,0);
  \filldraw[black] (rtl) circle(1.2pt);
  \draw[thick,middlearrow={>}] (ptBFP) --  (ptTube);
  

  \coordinate (tarc1) at  (3*\lensRadius,\lensHeight);
  \coordinate (tarc2) at  (3*\lensRadius,-1*\lensHeight);
  \draw  (tarc1) to[in=110,out=250] (tarc2)  to[in=-70,out=70] (tarc1)  -- cycle;

  \coordinate[label=below:$\ttd_0$] (d0) at (3*\lensRadius+0.4,0.5);
  \coordinate[label=below:$\ttd$] (d) at  (3*\lensRadius+0.25,2.0);

  \draw[dashed] (3*\lensRadius,-1.3*\lensHeight,0) -- (3*\lensRadius,1.3*\lensHeight,0);
  \coordinate[label=above:$\text{Tube Lens}$] (Tubetext) at (3*\lensRadius,1.3*\lensHeight,0);
  \coordinate[label=above:\autoref{ss:tube}] (Tubetext) at (3*\lensRadius,\lensHeight-7.9,0);
  
  \draw[thick,middlearrow={>}] (ptTube) --  (rImage);

  \draw (4*\lensRadius -1,-1*\lensHeight) -- (4*\lensRadius -1,1*\lensHeight+1);
  \draw (4*\lensRadius -1,1*\lensHeight+1) -- (4*\lensRadius +1,1*\lensHeight);
  \draw (4*\lensRadius +1,-1*\lensHeight-1) -- (4*\lensRadius +1,1*\lensHeight);
  \draw (4*\lensRadius +1,-1*\lensHeight-1) -- (4*\lensRadius -1,-1*\lensHeight);
  
  \draw[dashed] (4*\lensRadius,-1.3*\lensHeight,0) -- (4*\lensRadius,1.3*\lensHeight,0);
  \coordinate[label=above:$\text{Image Plane}$] (Imagetext) at (4*\lensRadius,1.3*\lensHeight,0);
  \coordinate[label=above:\autoref{ss:final}] (Imagetext) at (4*\lensRadius,\lensHeight-7.9,0);
  
  \coordinate (t2) at (3*\lensRadius,-1.3*\lensHeight,0);
  \coordinate (im2) at (4*\lensRadius,-1.3*\lensHeight,0);
  \tikzset{
    position label/.style={
       below = 3pt,
       text height = 1.5ex,
       text depth = 1ex
    },
   brace/.style={
     decoration={brace,mirror},
     decorate
    }
  }  
  \draw [brace,decoration={raise=0.5ex}] [brace] (t2.south) -- (im2.south) node [position label, pos=0.5, rotate = 0,scale=1.0] {$\flens$};
  
  \coordinate[label=below:$\mathcal{I}$] (rImage) at (4*\lensRadius,2,0);
\end{tikzpicture}
  \caption{The \emph{plane of observation} is defined as the plane containing the dipole $\bpsi$, 
        the $\be_3$-axis, and the path of a particular ray through the objective, the \emph{back focal plane} and the tube lens 
        (with focal length $\fobj$).}
  \label{fig:Scheme}
\end{figure}

We summarize the different coordinate systems used below in a table:
\begin{table}[h]
\begin{center}
 \begin{tabular}{|r||c|c|}
 \hline
  Position & Coordinates & Fourier \\
  \hline
  Medium $\Omega$ & ${\br} \in \R^3$, $(r,\theta,\varphi) \in \R_+ \times \sphere$ &$\bk \in \R^3$\\
  Back focal plane (BFP) & $\bx \in \R^2$, $(\rho,\varphi) \in \R_+ \times [0,2\pi)$ & $\bu \in \R^2$, $(\xi,\nu) \in \R_+ \times [0,2\pi)$\\
  Tube Lens  & $\by \in \R^2$, $(\varrho,\sigma) \in \R_+ \times [0,2\pi)$ & $\bv \in \R^2$, $(\varkappa ,\vartheta) \in \R_+ \times [0,2\pi)$\\
  Image plane (IP) $\mathcal{I}$ & $\bx_\focus \in \R^2$ & $\bu_\focus \in \R^2$ \\
  Between BFP and IP & $(\bx,r_3) \in \R^3$ & \\
  \hline
 \end{tabular}
 \begin{tabular}{|r|l|}
  \hline
  General notation &  $\bk_{12} = (k_1,k_2)^T \in \R^2$ $\bk = (k_1,k_2,k_3)^T \in \R^3$\\ 
  & $\bk_z = (k_1, k_2, z)^T$  $z \in \C$, $k_1,k_2 \in \R$\\
  \hline
  \end{tabular}
\caption{\label{ta:abb} Some abbreviation to look after in \autoref{lemmaE} and its proof, as well as in \autoref{app:FT} and \autoref{app:farField}.}
\end{center}
\end{table}

\subsection{Far Field Approximation in the Medium} \label{ss:ff}
In this subsection we derive the far field approximation of the Fourier-transform of the electric field, $\widehat\bE$, \emph{in the medium}.
The derivation expands \cite{ForWeb84}.

First, we give the definition of the far field: 
\begin{definition}\label{de:far_field}
The \emph{far field} $F_\infty: \sphere \to \C^3$ of a function $F: \R^3 \to \C^3$ satisfies: 
There exists $\hat{C} > 0$ and a function $C :[0,\infty) \to [0,\infty)$ such that
 \begin{equation} \label{eq:farfield}
 \lim_{r \to \infty} \abs{F(r,\theta,\varphi)  -  C(r) F_\infty(\theta,\varphi) } = 0 \quad \text{ with } \quad
 \abs{ r C(r) } \leq \hat{C} \text{ for all } r \in [0,\infty).
\end{equation}
\end{definition}


\begin{lemma}\label{lemmaEfar}
Let the medium be isotropic, non magnetizable, homogeneous and have a linear polarization response.
We assume that the considered dipole $\bpsi$ is located at position $\br^\bpsi = \begin{pmatrix} 0 \\ 0 \\ \rpsi \end{pmatrix}$
with $\rpsi < 0$. Moreover, let $\br = \vr \in \R^3$ with $r_3 > \rpsi$; 
The later assumption means that we are considering only light rays, which are propagating into the lens system (see \autoref{fig:Scheme}).

Then the \emph{far field} of $\widehat{\bE}$ \emph{in the medium} is given by
\begin{equation}\label{eq:Efarfield1}\boxed{
\begin{aligned}
 \e^{\i \kappa  \rpsi\cos(\theta)} \widehat{\bE}_\infty(\theta,\varphi) = 
 \cos(\theta) 
\biggl(-\Psi_p \cos(\theta)+\Psi_3 \sin(\theta) \biggr) \be_p - \Psi_s \be_s + 
 \sin(\theta) \biggl(\Psi_p \cos(\theta) - \Psi_3 \sin(\theta)\biggr) \be_3
\end{aligned}}
 \end{equation} 
 and 
\begin{equation} \label{eq:C}
\boxed{
C(r) = \frac{\kappa^2}{8\pi \ve_0} \frac{\e^{\i \kappa r }}{r}.}
\end{equation}
where $\Psi_j = \inner{\bpsi}{\be_j}$, $j=p,s,3$ are the coefficients of $\bpsi$ with respect to the 
orthonormal basis
\begin{equation} \label{eq:be}
   \be_p := \begin{pmatrix} 
                  \cos(\varphi) & \sin(\varphi) & 0
                 \end{pmatrix}^T, \quad 
  \be_s := \begin{pmatrix}
                 -\sin(\varphi) & \cos(\varphi) & 0
               \end{pmatrix}^T, \quad
  \be_3,
 \end{equation}
 that is 
\begin{equation} \label{eq:ps3}
   \bpsi = \Psi_p \be_p + \Psi_s \be_s + \Psi_3 \be_3.
\end{equation}
\end{lemma}

\begin{proof}
Taking into account the assumption that $r_3-\rpsi > 0$, and by representing the vector 
$\br \in \R^3$ as 
\begin{equation} \label{eq:v}
 \br = r_3 \be_3 + r_3 \begin{pmatrix}
                                      \bv & 0
                                     \end{pmatrix}^T = r_3 \begin{pmatrix}
                                      v_1 & v_2 & 1
                                     \end{pmatrix}^T,
\end{equation}
with a (non-unit) vector $\begin{pmatrix}
                                      v_1 & v_2
                                     \end{pmatrix}^T = \bv \in \R^2$ in the plane spanned by  
$\be_1$ and $\be_2$, it follows from \autoref{eq:EFW} that
\begin{equation} \label{eq:fern}
\widehat{\bf E}(\br) = -\frac{1}{16\pi^2}\frac{1}{\epsilon_0} 
 \int_{\bk_{12} \in \R^2} \e^{ \i r_3 \bv \cdot \bk_{12}} \left( \frac{\i \e^{\i q (r_3-\rpsi)}}{q} (\bpsi \times \bk_q) \times \bk_q\right) \dd \bk_{12}, 
\end{equation} 

where $q$ and $\bk_q$ are as defined in \autoref{eq:qk1}. 
Note that in \autoref{eq:fern} $q=q(\bk_{12})$ is as defined in \autoref{eq:bv3}, and therefore the integral 
on the right hand side is of the form (neglecting the factor $-\frac{\i}{16\pi^2}\frac{1}{\epsilon_0}$)
\begin{equation*} 
  \int_{\bk_{12} \in \R^2} \e^{\i r_3 \zeta(\bk_{12})} \beta(\bk_{12}) \dd \bk_{12}
\end{equation*}
with 
\begin{equation} \label{eq:alpha}
\zeta(\bk_{12}) = \bk_{12} \cdot \bv +  q    \quad \text{ and } \quad
\beta(\bk_{12}) = \frac{\e^{-\i \rpsi q}}{q} (\bpsi \times \bk_q) \times \bk_q .
\end{equation}
The \emph{stationary phase method}, \cite[Th.~7.7.5]{Hoe03}, states that if $\hat{\bk}$ is a critical point of $\zeta$, which has been calculated 
in \autoref{eq:bv}, then 
\begin{equation*} 
 \begin{aligned}
  \int_{\bk_{12} \in \R^2} \e^{\i r_3 \zeta(\bk_{12})} \beta(\bk_{12}) \dd \bk_{12} &=
  \e^{\i r_3 \zeta(\hat{\bk})} \left( \det \begin{pmatrix}
                                             r_3 H(\zeta) (\hat{\bk})/(2\pi \i)
                                            \end{pmatrix} \right)^{-1/2} \beta(\hat{\bk}) + o\left( \frac{1}{r_3} \right).
 \end{aligned}
\end{equation*}
Taking into account \autoref{eq:crit} in \autoref{le:alpha}, and $\hat{\bk}_{12}$ of $\zeta$ as defined in \autoref{eq:bv}, and 
being aware that $q=q(\bk_{12})$ (that is $q$ is a function of $\bk_{12}$), we apply \autoref{eq:bv2}, \autoref{eq:bv3} and get 
\begin{align*}
  & \int_{\bk_{12} \in \R^2} \frac{\e^{\i r_3 (\bk_{12} \cdot \bv + q)}}{q} \e^{-\i q \rpsi} (\bpsi \times \bk_q) \times \bk_q \dd \bk_{12} \\
 =& 2 \i \pi \kappa^2 
    \frac{\e^{\i r_3 \kappa \sqrt{1 + \abs{\bv}^2} }}{r_3}  
     \e^{-\i \frac{\kappa\rpsi}{\sqrt{1 + \abs{\bv}^2}} }
     \left(\bpsi \times \frac{\br}{\abs{\br}}\right) \times \frac{\br}{\abs{\br}} + o\left(\frac{1}{r_3}\right) \text{ for } r_3 \to \infty.
\end{align*}
Now, we recall \autoref{eq:br} and \autoref{eq:v}, which imply that 
\begin{equation*}
\begin{aligned}
  ~ & \sqrt{1+\abs{\bv}^2} = \frac{1}{\cos(\theta)}, \quad \frac{\br}{\abs{\br}} = \begin{pmatrix}
   \sin(\theta) \cos(\phi) \\ \sin(\theta) \sin(\phi) \\ \cos(\theta)
                                                                       \end{pmatrix} = \sin(\theta) \be_p + \cos(\theta) \be_3\\
 & \text{ and } \abs{\br} \cos(\theta) = r_3,
\end{aligned}
\end{equation*}
such that we get
\begin{align*}
 \int_{\bk_{12} \in \R^2} & 
 \frac{\e^{\i r_3 (\bk_{12} \cdot \bv + q)}}{q} \e^{-\i q \rpsi} (\bpsi \times \bk_q) \times \bk_q \dd \bk_{12} \\
 &= 
 2 \i \pi \kappa^2 \e^{-\i \kappa  \rpsi\cos(\theta)} \frac{\e^{\i \kappa \abs{\br}}}{\abs{\br}}  \left(\bpsi \times \frac{\br}{\abs{\br}}\right) \times \frac{\br}{\abs{\br}} + o\left(\frac{1}{r_3}\right).
\end{align*} 
This shows that
\begin{equation}\label{eq:app_farfield}
 \widehat{\bf E}(\br) = 
   \e^{-\i \kappa \rpsi\cos(\theta)} \frac{\kappa^2}{8\pi \epsilon_0} \frac{\e^{\i \kappa \abs{\br}}}{\abs{\br}} 
  \left(\bpsi \times \frac{\br}{\abs{\br}}\right) \times \frac{\br}{\abs{\br}} + o\left(\frac{1}{r_3}\right)
  = 
 C(r) \widehat{\bf E}_\infty(\br) + o\left(\frac{1}{r_3}\right).
\end{equation}
It remains to compute the second identity of \autoref{eq:Efarfield1}.
Expressing $\displaystyle \frac{\br}{\abs{\br}}$ and $\bpsi$ in terms of the associated basis $\be_p$, $\be_s$, $\be_3$ from \autoref{de:basis}, 
and using \autoref{eq:doubleprod}, we get from \autoref{eq:ps3}
\begin{equation*} \begin{aligned}
 \left(\bpsi \times \frac{\br}{\abs{\br}}\right) \times \frac{\br}{\abs{\br}} =&
 \left(\bpsi \cdot \frac{\br}{\abs{\br}}\right) \frac{\br}{\abs{\br}} - \bpsi \\
 = & \biggl( (\sin(\theta) \be_p + \cos(\theta) \be_3) \cdot (\Psi_p \be_p + 
      \Psi_s \be_s + \Psi_3 \be_3)\biggr) (\sin(\theta) \be_p + \cos(\theta) \be_3) \\
   & - \Psi_p \be_p - \Psi_s \be_s - \Psi_3 \be_3 \\
 = & \biggl( \sin(\theta) \Psi_p + \cos(\theta) \Psi_3 \biggr) (\sin(\theta) \be_p + \cos(\theta) \be_3) - \Psi_p \be_p - \Psi_s \be_s - \Psi_3 \be_3,
\end{aligned} \end{equation*}
which after rearrangement proves the second identity.
\end{proof}

In the imaging system the calculation of the electric field is not done at once but in different sections (\autoref{ss:ff}, \autoref{ss:cyl} and \autoref{ss:between}, \autoref{ss:tube} and \autoref{ss:final}). In each of these sections the electric field is calculated by transmission from the electric field computed at the previous section. In addition, we assume that the light which hits the objective from $\Omega$ can be approximated by its far field expansion $C(\focus_\obj)\widehat{\bE}_\infty$, which we will use instead of $\widehat{\bE}$.

\subsection{Propagation of the electric field through the objective} \label{ss:cyl}
In the following we calculate the electric field in the objective.
Assuming that the electric field (light) emitted from the dipoles travels along straight lines in the medium to the 
objective, the objective aligns the emitted rays from the dipole parallel to the $r_3$-axis in such a way that the electric 
field between the incidence surface of the objective and the back focal plane undergoes a phase shift that does not depend on the distance to the optical axis.
In the ideal situation, where the wavelength is assumed to be infinitely small compared to the length parameters of the optical system, the electric field can be computed via the 
\emph{intensity law of geometrical optics} (see \autoref{fig:Area_dA12} and \cite[Sec. 3.1.2]{BorWol99} for a derivation). 

\begin{assumption} \label{ass:farfield}
The objective consists of a set of optical elements (lenses and mirrors) which are not modelled here (see some examples in \cite[Sec. 6.6]{BorWol99}). Its aim is to transform spherical waves originated at its focal point into waves which propagate along the optical axis. In what follows, the computations are made ignoring a constant (independent on the point in the back focal plane) phase shift which is underwent by the wave through the objective.
\end{assumption}

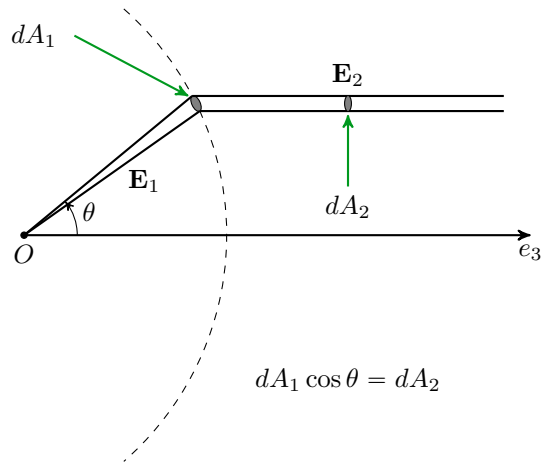
\begin{figure}[ht] 
\centering

  \pgfmathsetmacro{\lensRadius}{4}
  \pgfmathsetmacro{\lensHeight}{3}
  \pgfmathsetmacro{\startAngle}{asin(\lensHeight/\lensRadius)}
  \pgfmathsetmacro{\lensCenter}{cos(\startAngle)*\lensRadius}
  \pgfmathsetmacro{\ptAngle}{\startAngle/2}
  \pgfmathsetmacro{\xptlens}{cos(\ptAngle)*\lensRadius}
  \pgfmathsetmacro{\yptlens}{sin(\ptAngle)*\lensRadius}

\begin{tikzpicture}[scale=1]
  \coordinate (O) at (\lensRadius/3,0,0);
  \filldraw[black] (O) circle(1.2pt);
  \coordinate[label=below:$O$] (origin) at (O);
    
  \coordinate[label=below:$e_3$] (E) at (2*\lensRadius,0,0);
  \draw[thick,->] (O) -- (E);
  
  \coordinate (D) at (0,0,0);
  
  \draw[dashed] (\lensCenter,\lensHeight) arc[start angle=\startAngle,delta angle=-2*\startAngle,radius=\lensRadius]; 
  
  \coordinate (P) at (\xptlens,\yptlens);
  \coordinate (P1) at (\xptlens-0.1,\yptlens+0.2);
  
  \coordinate (M) at (\xptlens-0.05,\yptlens+0.1); 
  \draw [black,fill=gray, rotate = -65] (M) ellipse (1.12mm and 0.55mm);
  
  \draw[thick] (O) -- (P);
  \draw[thick] (O) -- (P1);

  \coordinate (B) at (\xptlens + \lensRadius,\yptlens);
  \coordinate (B1) at (\xptlens +\lensRadius-0.1,\yptlens+0.2);
  
  \coordinate (M1) at (\xptlens + 0.5*\lensRadius-0.05,\yptlens+0.1);
  
  \draw[thick] (P) -- (B);
  \draw[thick] (P1) -- (\xptlens +\lensRadius,\yptlens+0.2);
  \draw [black,fill=gray, rotate = -90] (M1) ellipse (1.05mm and 0.450mm);
   
  \draw[thick, ->,color={blue!30!black!40!green}] (\xptlens + 0.5*\lensRadius-0.05,\yptlens-1) -- (\xptlens + 0.5*\lensRadius-0.05,\yptlens-0.05);
  \node at (\xptlens + 0.5*\lensRadius-0.05,\yptlens-1) [below] {$dA_2$};
  \node at (\xptlens + 0.5*\lensRadius-0.05,\yptlens+0.75) [below] {$\bE_2$};
  
  \draw[thick, ->,color={blue!30!black!40!green}] (\xptlens - 0.4*\lensRadius-0.05,\yptlens+1) -- (\xptlens -0.15,\yptlens+0.2);
  \node at (\xptlens - 0.4*\lensRadius-0.05,\yptlens+1) [left=.3em] {$dA_1$}; 
  \node at (0.8*\xptlens ,0.6*\yptlens) [below] {$\bE_1$}; 
  
  \node at (\xptlens + 0.5*\lensRadius-0.05,-\yptlens) [below] {$dA_1 \cos\theta = dA_2$};
  
  \tikzset{
    position label/.style={
       above = 3pt,
       text height = 1.5ex,
       text depth = 1ex
    },
   brace/.style={
     decoration={brace,mirror},
     decorate
   }
  } 
  \pic[draw,->,angle radius=.7cm,angle eccentricity=1.3,"$\theta$"] {angle=E--O--P1}; 
  
\end{tikzpicture}
  \caption{Intensity law of Geometrical Optics: The energy carried along a ray must remain constant.
  The power transported by a ray is proportional to $\abs{\bE}^2 dA$,
  where $dA$ is an infinitesimal cross-section perpendicular to the ray propagation.
  Thus, the fields must satisfy $\abs{\bE_2} = \abs{\bE_1} \frac{1}{\sqrt{\cos(\theta)}} $. 
  }
  \label{fig:Area_dA12}
\end{figure}

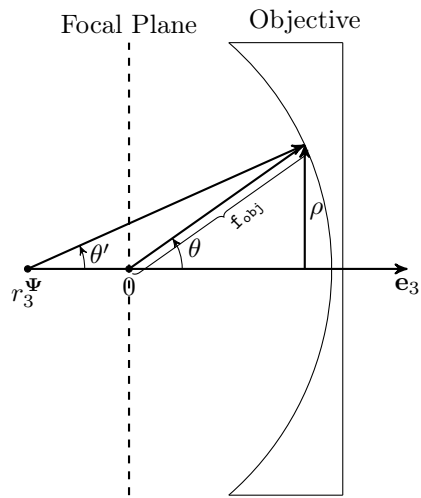
\begin{figure}[ht] 
\centering

  \pgfmathsetmacro{\lensRadius}{4}
  \pgfmathsetmacro{\lensHeight}{3}
  \pgfmathsetmacro{\startAngle}{asin(\lensHeight/\lensRadius)}
  \pgfmathsetmacro{\lensCenter}{cos(\startAngle)*\lensRadius}
  \pgfmathsetmacro{\ptAngle}{\startAngle/2}
  \pgfmathsetmacro{\xptlens}{cos(\ptAngle)*\lensRadius}
  \pgfmathsetmacro{\yptlens}{sin(\ptAngle)*\lensRadius} 

\begin{tikzpicture}[scale=1]

  \coordinate (D) at (0,0,0);
  \filldraw[black] (D) circle(1.2pt);
  \coordinate[label=below:$r_3^{\bpsi}$] (dipole) at (D);
  
  \draw[thick,->] (D) -- (\lensRadius+1,0,0);
  
  \coordinate[label=below:$\be_3$] (E) at (\lensRadius+1,0,0);

  \coordinate[label=below:$\rho$] (R) at (\lensRadius-0.2,1.0,0);
  \draw[thick,->] (\xptlens,0.0,0) --  (\xptlens,\yptlens);
  
  \coordinate (O) at (\lensRadius/3,0,0);
  \filldraw[black] (O) circle(1.2pt);
  \coordinate[label=below:$0$] (origin) at (O);
  
  \draw[thick,dashed] (\lensRadius/3,-1*\lensHeight,0) -- (\lensRadius/3,\lensHeight,0);
  \coordinate[label=above:$\text{Focal Plane}$] (fplane) at (\lensRadius/3,\lensHeight,0); 
  
  \coordinate (P) at (\xptlens,\yptlens);
  
  \draw[thick,->] (D) --  (\xptlens,\yptlens);
  
  \draw[thick] (O) -- (\xptlens,\yptlens);

  \draw (\lensCenter,\lensHeight) arc[start angle=\startAngle,delta angle=-2*\startAngle,radius=\lensRadius];
  
  \draw (\lensCenter,\lensHeight) -- (\lensCenter+1.5,\lensHeight);
  \coordinate[label=above:$\text{Objective}$] (obj) at (\lensCenter+1,\lensHeight);
  \draw (\lensCenter,-1*\lensHeight) -- (\lensCenter+1.5,-1*\lensHeight); 
  \draw(\lensCenter+1.5,\lensHeight) -- (\lensCenter+1.5,-1*\lensHeight); 
    
\tikzset{
    position label/.style={
       above = 3pt,
       text height = 1.5ex,
       text depth = 1ex
    },
   brace/.style={
     decoration={brace,mirror},
     decorate
   }
}  
\draw [brace,decoration={raise=0.5ex}] [brace]  (O.south) -- (P.south) node [position label, below=0.1, pos=0.6, rotate = 45,scale=0.8] 
{$\fobj$};
\pic[draw,->,angle radius=.75cm,angle eccentricity=1.3,"$\theta'$"] {angle=E--D--P};
\pic[draw,->,angle radius=.7cm,angle eccentricity=1.3,"$\theta$"] {angle=E--O--P}; 

\end{tikzpicture}
\caption{Approximation used: We assume that the cell is fixed to the glass, and the distance of the dipole $\abs{r_3^\bpsi}$ from the focal plane is sufficiently smaller than $r = \abs{\bf r}$, such that $\fobj \approx r$, and $\theta' \approx \theta$. }
\label{fig:approxf_obj}
\end{figure}

\begin{lemma} \label{le:obj}
Let $\br$ be a point at the back focal plane of the objective, that is with $r_3$ coordinate $r_3^\bfp$ and with spherical coordinates $(r,\theta,\varphi)$. We define, after the objective, the \emph{radial length} on the propagation plane (planes with constant $r_3$ coordinate), by
\begin{equation} \label{eq:rad_length}
 \rho := \rho(\theta):= \fobj \sin(\theta).
\end{equation}
Then
\begin{equation}\label{eq:EBFP_pre}
\begin{aligned}
 ~ & \widehat{\bE}^\bfp(\rho,\varphi) 
 := \begin{cases} \frac{C(\fobj)}{\sqrt{\cos(\theta)}} \e^{-\i \kappa \rpsi\cos(\theta)}
\begin{pmatrix} 
 \cos(\varphi)  \left( -\Psi_p \cos(\theta)  +  \Psi_3 \sin(\theta)  \right) - \sin(\varphi)  \Psi_s   \\
 \sin(\varphi)  \left( -\Psi_p \cos(\theta)  +  \Psi_3 \sin(\theta)  \right) + \cos(\varphi)   \Psi_s   
 \\  0
\end{pmatrix} & \theta \ls \tmax \\
0 & \theta > \tmax
\end{cases}
 \end{aligned}
\end{equation} 
where 
\begin{equation*}
\tmax := \arcsin(\NA)
\end{equation*}
is the maximal angle $\theta$ for rays to enter the objective (the other rays simply do not enter the optical system). Note that 
the refractive index in air is assumed one.
\end{lemma}
\begin{proof} The electric field is transmitted according to the law of geometrical optics \cite[Eq. 16]{Axe12} into the objective 
at the points 
\begin{equation} \label{eq:circle_obj}
\br = \begin{pmatrix} 0 \\ 0 \\ \rpsi 
          \end{pmatrix} + \fobj \mathbb{S}^2,
\end{equation}
that is the electric field simply undergoes a rotation of axis $\be_s$ and angle $\theta$ as well as a magnification of $\frac{1}{\sqrt{\cos \theta}}$ (see \autoref{ass:farfield}).

The rotation with angle $\theta$ around the axis $\be_s$ changes the unit vectors as follows:
\begin{equation} \label{eq:rottheta} \begin{matrix}
    \be_p & \to & \sin(\theta) \be_3 + \cos(\theta) \be_p\\
    \be_3 & \to & \cos(\theta) \be_3 - \sin(\theta) \be_p\\
    \be_s & \to & \be_s.
   \end{matrix} \end{equation}

Now, \autoref{eq:rottheta} shows that the expression of the electric field in the back focal plane will be simpler using coordinates $(\be_p,\be_s,\be_3)$. \autoref{eq:Efarfield1} leads to
 \begin{equation*}
 \widehat{\bE}^\bfp(\rho,\varphi) = \frac{C(\fobj)}{\sqrt{\cos(\theta)}} \e^{-\i \kappa  \rpsi\cos(\theta)}
 \left \{  
  \biggl(-\Psi_p \cos(\theta)+\Psi_3 \sin(\theta) \biggr) \be_p - \Psi_s \be_s
\right \}, 
  \end{equation*} where $C(\fobj)$ is as defined in \autoref{eq:C}. Writing the unit vectors $\be_p$ and $\be_s$ 
  in the fixed system of coordinates $(x_1,x_2,x_3)$ gives \autoref{eq:EBFP_pre}.
  \end{proof}
  \subsection{Between the Objective and the Lens} \label{ss:between}
 After the objective, the light propagates through air until it reaches the tube lens. 
 Denoting by $\kappa^2$ the wave number in air (see \autoref{eq:kappa}) the electric field 
 satisfies the homogeneous Helmholtz equation in the tube lens: 
 \begin{equation} \label{eq:bvp}
 \Delta \widehat{\bE}(\br) + \kappa^2 \widehat{\bE}(\br) = 
 0 \quad \text{ in } \quad \mathcal{H} := \set{\br \in \R^3 : r_3^\bfp  < r_3 < r_3^{\texttt{tl}}-\texttt{d}_0}
 \end{equation}
together with the boundary condition 
\begin{equation} \label{eq:bc}
\widehat{\bE}(r_1,r_2,r_3^\bfp)  = \widehat{\bE}^\bfp(r_1,r_2) \text{ for all } (r_1,r_2) \in \R^2.
\end{equation}
The solution of \autoref{eq:bvp} can actually be calculated by applying a phase shift to $\widehat{\bE}^\bfp$ as the following 
lemma shows. 

\begin{lemma} \label{le:aperture}
Representing $\bx = (\fobj \sin(\theta)\cos(\varphi),\fobj \sin (\theta)\sin(\varphi)) \in \R^2$, then the Fourier transforms of $\widehat{\bE}$ 
 in the transverse plane of $(\bx,r_3)$ can be calculated from $\widehat{\bE}(\bx,r_3^{\texttt{bfp}})$ in the following way
 \begin{equation} \label{eq:EFbetween}
  \mathcal F_{12}(\widehat{\bE})(k_1,k_2,r_3) = \mathcal F_{12}[\widehat{\bE}^{\texttt{bfp}}](k_1,k_2) 
  \e^{(r_3 - r_3^{\texttt{bfp}}) \sqrt{-\kappa^2+k_1^2+k_2^2}}.
 \end{equation}
 where the square root can denote both of the complex square roots.
\end{lemma}
\begin{proof}
First, we notice that since $\widehat{\bE}^{\texttt{bfp}}$ is bounded with compact support, it is a $L^2$ function in the plane $\{r_3 = r_3^{\texttt{bfp}}\}$. Taking the Fourier transform in these two variables, \autoref{eq:bvp} and \autoref{eq:bc} are equivalent to
\begin{equation} \label{eq:bvpFour}
 \partial^2_{r_3} \mathcal F_{12}(\widehat{\bE})(k_1,k_2,r_3) + (\kappa^2-k_1^2-k_2^2) \mathcal F_{12}(\widehat{\bE})(k_1,k_2,r_3) = 
 0 \text{ for all } \br \in \R^3, \; r_3^{\texttt{bfp}} < r_3 < r_3^{\texttt{tl}}
 \end{equation}
with the boundary condition
\begin{equation} \label{eq:bcFour}
\mathcal F_{12}(\widehat{\bE})(k_1,k_2,r_3^{\texttt{bfp}}) = \mathcal F_{12}[\widehat{\bE}^{\texttt{bfp}}](k_1,k_2) \text{ for all } (k_1,k_2) \in \R^2.
\end{equation}
Now, \autoref{eq:bvpFour} is a simple ODE whose solution writes (for $\kappa^2-k_1^2-k_2^2 \neq 0$)
\[ \mathcal F_{12}(\widehat{\bE})(k_1,k_2,r_3) = \mathcal F_{12}[\widehat{\bE}^{\texttt{bfp}}](k_1,k_2) \e^{(r_3 - r_3^{\texttt{bfp}}) \sqrt{-\kappa^2+k_1^2+k_2^2}}.\]
\end{proof}
Among the fields computed in \autoref{eq:EFbetween}, several are not physical or will not be observed:
\begin{itemize}
 \item Having $\kappa^2 < k_1^2+k_2^2$ leads to either a real positive square root which corresponds to a wave exploding as $r_3$ increases and is therefore not physical or a real negative root, which yields an exponentially decreasing wave (\emph{evanescent}) which exist but, since $(r_3 - r_3^{\texttt{bfp}})$ is several oders of magnitude bigger than the wave length, will be damped by the time it hits the tube lens. Therefore we also do not consider it.
 \item When, $\kappa^2 < k_1^2+k_2^2$, we get two imaginary roots, namly $\pm \i \sqrt{\kappa^2 - k_1^2 - k_2^2}$, which corresponds to the two Green functions \autoref{green_func}. For the same reason as above, we will only consider the positive sign.
\end{itemize}
This can be summerized in the following assumption, that will hold in what follows.

\begin{assumption}
We only consider $\kappa^2 \gs k_1^2 + k_2^2$ and we obtain
\begin{equation}
 \mathcal F_{12}(\widehat{\bE})(k_1,k_2,r_3) = \mathcal F_{12}[\widehat{\bE}^{\texttt{bfp}}](k_1,k_2) \e^{\i(r_3 - r_3^{\texttt{bfp}}) \sqrt{\kappa^2-k_1^2-k_2^2}}.
\end{equation}
\end{assumption}

In the following we calculate $\mathcal F_{12}[\widehat{\bE}^{\texttt{tl}_i}]$, where 
$\widehat{\bE}^{\texttt{tl}_i} = \widehat{\bE}(\bx,r_3^{\texttt{tl}_i})$.
\begin{lemma}\label{lem:Fourier_EBFP}

 Let $\widehat{\bE}^{\texttt{tl}_i}(\bx)= \widehat{\bE}(\bx,r_3^{\texttt{tl}_i})$ be the electric field at the indicent plane of the tube thin lens (at $r_3^{\texttt{tl}_i}$) as defined in \autoref{eq:EFbetween}, then the Fourier transform of $\widehat{\bE}^{\texttt{tl}_i}$ in this plane in polar coordinates $(\xi,\nu)$ of $(k_1,k_2)$ is given, for $\xi^2 \ls \kappa^2$, by 
\begin{equation} \label{eq:Fourier_EBFP_polar}
\begin{aligned}
   &\mathcal F_{12}[\widehat{\bE}^{\texttt{tl}_i}] (\xi,\nu)  =  \mathcal F_{12}[\widehat{\bE}^{\texttt{bfp}}](\xi, \nu) \e^{\i(r_3^{\texttt{tl}_i} -r_3^\bfp) \sqrt{\kappa^2-\xi^2}} \\ 
  &\; = \frac{1}{2} (\fobj)^2 C(\fobj) \e^{\i (r_3^{\texttt{tl}_i} -r_3^\bfp)\sqrt{\kappa^2 - \xi^2}}  \cdot \\
  & \quad   \cdot
\begin{pmatrix} 
-\Psi_1 [ I_{1,0}(\xi) + I_{2,0}(\xi) ] + \Psi_1 \cos (2 \nu) [I_{1,2}(\xi) + I_{2,2}(\xi)] +\Psi_2 \sin (2 \nu)[I_{1,2}(\xi) +I_{2,2}(\xi)] -2 \i \Psi_3 \cos(\nu)I_{2,1}(\xi) 
\\
-\Psi_2 [ I_{1,0}(\xi) + I_{2,0}(\xi) ] + \Psi_2 \cos (2 \nu) [I_{1,2}(\xi) + I_{2,2}(\xi)] +\Psi_1 \sin (2 \nu)[I_{1,2}(\xi) +I_{2,2}(\xi)] -2 \i \Psi_3 \sin(\nu)I_{2,1}(\xi)
\\
0
\end{pmatrix},
 \end{aligned}
\end{equation} 
where
\begin{equation}\label{eq:intBessels}
\begin{aligned}
I_{1,0}(\xi) &= \int_0^\tmax  \sqrt{\cos(\theta)}  \, \sin(\theta)  \, \e^{-\i \kappa  \rpsi\cos(\theta)} J_0\left( \fobj \xi \sin(\theta)\right)   d \theta \\
I_{1,2}(\xi) &= \int_0^\tmax  \sqrt{\cos(\theta)}  \, \sin(\theta)  \, \e^{-\i \kappa  \rpsi\cos(\theta)} J_2\left( \fobj \xi \sin(\theta)\right)   d \theta \\
I_{2,1}(\xi) &= \int_0^\tmax  (\cos(\theta))^{3/2} \, \frac{1-\cos(2\theta)}{2} \e^{-\i \kappa  \rpsi\cos(\theta)} J_1\left( \fobj \xi \sin(\theta)\right)   d \theta \\
I_{2,0}(\xi) &= \int_0^\tmax  \sqrt{\cos(\theta)}  \, \frac{\sin(2\theta)}{2} \e^{-\i \kappa  \rpsi\cos(\theta)} J_0\left( \fobj \xi \sin(\theta)\right)   d \theta \\
I_{2,2}(\xi)&= \int_0^\tmax  \sqrt{\cos(\theta)}  \, \frac{\sin(2\theta)}{2} \e^{-\i \kappa  \rpsi\cos(\theta)} J_2\left( \fobj \xi \sin(\theta)\right)   d \theta,
\end{aligned}
\end{equation} 
$J_m$ denotes the Bessel function of the first kind of order $m$, and $\tmax$ is the angle of aperture as defined in \autoref{eq:diffraction_limit}.
\end{lemma}

\begin{proof}
We use the following notation
\begin{equation*} 
\bu = \begin{pmatrix}
       u_1\\u_2
      \end{pmatrix} = \xi \begin{pmatrix}
       \cos(\nu)\\ \sin(\nu)
      \end{pmatrix}
      \text{ and } \bx = \begin{pmatrix}
       x_1\\
       x_2
      \end{pmatrix} = \rho\begin{pmatrix}
        \cos(\varphi)\\
        \sin(\varphi)
      \end{pmatrix}, 
\end{equation*} 
where $\rho=\rho(\theta)$ (see \autoref{eq:rad_length}) is the radial length on the back focal plane.

The two-dimensional Fourier transform of the $\widehat{\bE}^{\texttt{obj}}$ 
(defined in \autoref{eq:EBFP_pre}) reads as follows: 
\begin{equation} \label{eq:Fourier_EBFP}
\begin{aligned}
  &\mathcal F_{12}[\widehat{\bE}^{\texttt{obj}}] (\bu) 
  =  \frac{1}{2\pi}  \int_{\bx \in \R^2} \widehat{\bE}^{\texttt{obj}}(\bx) \e^{-\i \bu \cdot \bx}  d \bx  \\
  & \quad =  \frac{1}{2\pi} (\fobj)^2 \int_0^\tmax \int_{0}^{2 \pi} \widehat{\bE}^{\texttt{obj}}(\rho(\theta),\varphi) \e^{-\i \xi \rho(\theta)\cos (\varphi - \nu)  }  
      \cos(\theta)\sin(\theta) d \varphi d \theta \\
  & \quad =  \frac{(\fobj)^2 }{2\pi}  C(\fobj)  
  \int_0^\tmax \int_{0}^{2 \pi} \left \{  
  \biggl(-\Psi_p \cos(\theta)+\Psi_3 \sin(\theta) \biggr) \be_p - \Psi_s \be_s
\right \} \\
  & \qquad \qquad \qquad \cdot 
     \e^{-\i \kappa  \rpsi\cos(\theta)}
     \e^{-\i \xi \rho(\theta)\cos (\varphi - \nu)  }  
      \sqrt{\cos(\theta)}\sin(\theta) d \varphi d \theta.
 \end{aligned}
\end{equation} 
Next we calculate the integral on the right hand side of \autoref{eq:Fourier_EBFP}:
\begin{equation} \label{eq:EBFP_integrand}
\begin{aligned}
& \int_0^\tmax \int_{0}^{2 \pi} \left \{  
  \biggl(-\Psi_p \cos(\theta)+\Psi_3 \sin(\theta) \biggr) \be_p - \Psi_s \be_s
\right \} 
\\   & \qquad \qquad \qquad \cdot 
     \e^{-\i \kappa  \rpsi\cos(\theta)}
     \e^{-\i \xi \rho(\theta)\cos (\varphi - \nu)  }  
      \sqrt{\cos(\theta)}\sin(\theta) d \varphi d \theta \\
  & =   -\int_0^\tmax \e^{-\i \kappa  \rpsi\cos(\theta)}  \sqrt{\cos(\theta)} \frac{\sin(2\theta)}{2} \left( \int_{0}^{2 \pi} \Psi_p  \be_p \e^{-\i \xi \rho(\theta)\cos (\varphi - \nu)} d \varphi \right)  d \theta \\   
  &\quad - \int_0^\tmax \e^{-\i \kappa  \rpsi\cos(\theta)}  \sqrt{\cos(\theta)} \sin(\theta) \left( \int_{0}^{2 \pi} \Psi_s  \be_s \e^{-\i \xi \rho(\theta)\cos (\varphi - \nu)} d \varphi \right)  d \theta \\
  &\quad + \int_0^\tmax \e^{-\i \kappa  \rpsi\cos(\theta)}  (\cos(\theta))^{3/2} \, \frac{1-\cos(2\theta)}{2}  \left( \int_{0}^{2 \pi} \Psi_3  \be_p \e^{-\i \xi \rho(\theta)\cos (\varphi - \nu)} d \varphi \right)  d \theta.
 \end{aligned}
\end{equation} 

We proceed by first evaluating the inner integrals (involving the $\varphi$ variable) on the right hand side of \autoref{eq:EBFP_integrand}, 
by transforming the $(\be_p,\be_s,\be_3)$ system to $(\be_1,\be_2,\be_3)$ system, and then using the Bessel identities \autoref{eq:Besselidentities}, to evaluate the integrals. 

Using \autoref{eq:be} it follows from \autoref{eq:emit} that  
 \begin{equation} \label{eq:Psicomps}
  \Psi_p = \abs{\bpsi} \sin (\theta_m) \cos (\varphi_m - \varphi), \quad \Psi_s = \abs{\bpsi} \sin (\theta_m) \sin (\varphi_m - \varphi). 
 \end{equation} 
Again by application of \autoref{eq:emit} and $\sin$ and $\cos$ summation formulas we get
 \begin{equation} \label{eq:PsiEPScomps}
\begin{aligned}
 \Psi_p \cos(\varphi)  &= \Psi_1  \frac{1+\cos (2\varphi)}{2} + \Psi_2 \frac{\sin (2\varphi)}{2}, \quad 
 \Psi_p \sin(\varphi)  = \Psi_1 \frac{\sin (2\varphi)}{2}  + \Psi_2  \frac{1-\cos (2\varphi)}{2}, \\
 \Psi_s \cos(\varphi)  &= -\Psi_1 \frac{\sin (2\varphi)}{2}  + \Psi_2  \frac{1+\cos (2\varphi)}{2}, \quad 
 \Psi_s \sin(\varphi)  = -\Psi_1  \frac{1-\cos (2\varphi)}{2} + \Psi_2 \frac{\sin (2\varphi)}{2}. 
\end{aligned}
\end{equation} 
Using \autoref{eq:be}, we express the first inner integral on the right hand side of \autoref{eq:EBFP_integrand}: 
 \begin{equation} \label{eq:firstexp_phi}
\begin{aligned}
 \int_{0}^{2 \pi} \Psi_p   \e^{-\i \xi \rho(\theta)\cos (\varphi - \nu)} d \varphi \be_p &= 
 \int_{0}^{2 \pi} \Psi_p \cos(\varphi)\e^{-\i \xi \rho(\theta)\cos (\varphi - \nu)} d \varphi  \be_1 \\ 
 &\quad + \int_{0}^{2 \pi} \Psi_p \sin(\varphi) \e^{-\i \xi \rho(\theta)\cos (\varphi - \nu)} d \varphi \be_2,
\end{aligned}
\end{equation} and to evalute the integral we use the Bessel identities \autoref{eq:Besselidentities}, and \autoref{eq:PsiEPScomps}.

We use \autoref{eq:Besselidentities} for $m=0$ and $m=2$, to evaluate the first integral in  \autoref{eq:firstexp_phi}:
\begin{equation} \label{eq:Psip_int} 
\begin{aligned}
 &\int_{0}^{2 \pi}  \Psi_p  \cos(\varphi)  \e^{-\i \eta \cos (\varphi - \nu)  }  d \varphi 
 = \Psi_1 \int_{0}^{2 \pi}  \frac{1+\cos (2\phi)}{2}  \e^{-\i \eta \cos (\varphi - \nu)  }  d \varphi  + \Psi_2  \int_{0}^{2 \pi} \frac{\sin (2\phi)}{2} \e^{-\i \eta \cos (\varphi - \nu)  }  d \varphi \\
 &\quad = \frac{\Psi_1}{2} \int_{0}^{2 \pi}  \e^{-\i \eta \cos (\varphi - \nu)  }  d \varphi + \frac{\Psi_1}{2} \int_{0}^{2 \pi}  \cos (2\phi)  \e^{-\i \eta \cos (\varphi - \nu)  }  d \varphi 
 + \frac{\Psi_2}{2}  \int_{0}^{2 \pi} \sin (2\phi) \e^{-\i \eta \cos (\varphi - \nu)  }  d \varphi \\
&\quad =  \pi  \Psi_1 J_0(\eta)  - \pi \Psi_1  \cos (2\nu) J_2(\eta)  - \pi \Psi_2  \sin (2\nu) J_2(\eta),
 \end{aligned}
\end{equation} where $\eta = \fobj \xi \sin(\theta)$, and calculation similar to \autoref{eq:Psip_int} yields 
\begin{equation} \label{eq:Psip_2int} 
\begin{aligned}
&\frac{1}{\pi} \int_{0}^{2 \pi}  \Psi_p  \sin(\varphi)  \e^{-\i \eta \cos (\varphi - \nu)  }  d \varphi 
 = \Psi_2   J_0(\eta) -  \Psi_2  \cos (2\nu) J_2(\eta) - \Psi_1 \sin (2\nu) J_2(\eta).
 \end{aligned}
\end{equation}

Thus, using \autoref{eq:Psip_int}, and \autoref{eq:Psip_2int}, in \autoref{eq:firstexp_phi}, the first integral expression on the right hand side of \autoref{eq:EBFP_integrand}, becomes
\begin{equation}\label{eq:firstexp}
 \begin{aligned}
  & -\frac{1}{\pi} \int_0^\tmax \e^{-\i \kappa  \rpsi\cos(\theta)}  \sqrt{\cos(\theta)} \frac{\sin(2\theta)}{2} \left( \int_{0}^{2 \pi} \Psi_p  \be_p \e^{-\i \xi \rho(\theta)\cos (\varphi - \nu)} d \varphi \right)  d \theta \\
&\quad = - (\Psi_1 \be_1 + \Psi_2 \be_2) \int_0^\tmax  \sqrt{\cos(\theta)} \frac{\sin(2\theta)}{2} \e^{-\i \kappa  \rpsi\cos(\theta)} J_0\left( \fobj \xi \sin(\theta)\right)   d \theta \\
&\qquad + (\Psi_1 \be_1 + \Psi_2 \be_2) \cos (2\nu)  \int_0^\tmax  \sqrt{\cos(\theta)} \frac{\sin(2\theta)}{2} \e^{-\i \kappa  \rpsi\cos(\theta)}  J_2\left( \fobj \xi \sin(\theta)\right)   d \theta \\
&\qquad + (\Psi_2 \be_1 + \Psi_1 \be_2) \sin (2\nu)  \int_0^\tmax  \sqrt{\cos(\theta)} \frac{\sin(2\theta)}{2} \e^{-\i \kappa  \rpsi\cos(\theta)}   J_2\left( \fobj \xi \sin(\theta) \right) d \theta \\
&\quad = -(\Psi_1 \be_1 + \Psi_2 \be_2) I_{2,0}(\xi,\rpsi) 
+   (\Psi_1 \be_1 + \Psi_2 \be_2) \cos (2\nu) I_{2,2}(\xi,\rpsi)
+  (\Psi_2 \be_1 + \Psi_1 \be_2) \sin (2\nu) I_{2,2}(\xi,\rpsi),
 \end{aligned}
\end{equation} where integrals $I_{p,q}(\xi,\rpsi)$ are as in \autoref{eq:intBessels}.

Similar calculation to \autoref{eq:Psip_int} yields
\begin{equation} \label{eq:Psis_integrals}
\begin{aligned}
 &\frac{1}{\pi} \int_{0}^{2 \pi}  \Psi_s  \sin(\varphi)  \e^{-\i \eta \cos (\varphi - \nu)  }  d \varphi 
  =  -  \Psi_1 J_0(\eta)  -  \Psi_1  \cos (2\nu) J_2(\eta)  -  \Psi_2  \sin (2\nu) J_2(\eta), \\
&\frac{1}{\pi} \int_{0}^{2 \pi}  \Psi_s  \cos(\varphi)  \e^{-\i \eta \cos (\varphi - \nu)  }  d \varphi 
 =  \Psi_2   J_0(\eta) + \Psi_2  \cos (2\nu) J_2(\eta) + \Psi_1 \sin (2\nu) J_2(\eta).
 \end{aligned}
\end{equation} 

Next, using \autoref{eq:be} and \autoref{eq:Psis_integrals}, we compute the second integral term on the right hand side of \autoref{eq:EBFP_integrand}:
\begin{equation}\label{eq:secexp}
 \begin{aligned}
  & -\frac{1}{\pi} \int_0^\tmax \e^{-\i \kappa  \rpsi\cos(\theta)}  \sqrt{\cos(\theta)} \sin(\theta) \left( \int_{0}^{2 \pi} \Psi_s  \be_s \e^{-\i \xi \rho(\theta)\cos (\varphi - \nu)} d \varphi \right)  d \theta \\
&\quad = - (\Psi_1 \be_1 + \Psi_2 \be_2) \int_0^\tmax  \sqrt{\cos(\theta)} \sin(\theta) \e^{-\i \kappa  \rpsi\cos(\theta)} J_0\left( \fobj \xi \sin(\theta)\right)   d \theta \\
&\qquad + (\Psi_1 \be_1 + \Psi_2 \be_2) \cos (2\nu)  \int_0^\tmax  \sqrt{\cos(\theta)} \sin(\theta) \e^{-\i \kappa  \rpsi\cos(\theta)}  J_2\left( \fobj \xi \sin(\theta)\right)   d \theta \\
&\qquad + (\Psi_2 \be_1 + \Psi_1 \be_2) \sin (2\nu)  \int_0^\tmax  \sqrt{\cos(\theta)} \sin(\theta) \e^{-\i \kappa  \rpsi\cos(\theta)}   J_2\left( \fobj \xi \sin(\theta) \right) d \theta \\
&\quad = -(\Psi_1 \be_1 + \Psi_2 \be_2) I_{1,0}(\xi,\rpsi) 
+   (\Psi_1 \be_1 + \Psi_2 \be_2) \cos (2\nu) I_{1,2}(\xi,\rpsi) \\
& \qquad +  (\Psi_2 \be_1 + \Psi_1 \be_2) \sin (2\nu) I_{1,2}(\xi,\rpsi),
 \end{aligned}
\end{equation} where integrals $I_{p,q}(\xi,\rpsi)$ are as in \autoref{eq:intBessels}.

Next, we compute the last integral term on the right hand side of \autoref{eq:EBFP_integrand}:
\begin{equation}\label{eq:lastexp}
 \begin{aligned}
  & \frac{1}{\pi}  \int_0^\tmax \e^{-\i \kappa  \rpsi\cos(\theta)}  (\cos(\theta))^{3/2} \, \frac{1-\cos(2\theta)}{2}  \left( \int_{0}^{2 \pi} \Psi_3  \be_p \e^{-\i \xi \rho(\theta)\cos (\varphi - \nu)} d \varphi \right)  d \theta \\
  &\quad = 
  \frac{1}{\pi} \Psi_3 \int_0^\tmax \e^{-\i \kappa  \rpsi\cos(\theta)}  (\cos(\theta))^{3/2} \, \frac{1-\cos(2\theta)}{2} 
  \left( \be_1 \int_{0}^{2 \pi} \cos(\varphi) \e^{-\i \xi \rho(\theta)\cos (\varphi - \nu)} d \varphi \right)  d \theta \\
  &\qquad + \frac{1}{\pi} \Psi_3 \int_0^\tmax \e^{-\i \kappa  \rpsi\cos(\theta)}  (\cos(\theta))^{3/2} \, \frac{1-\cos(2\theta)}{2} 
  \left( \be_2 \int_{0}^{2 \pi} \sin(\varphi) \e^{-\i \xi \rho(\theta)\cos (\varphi - \nu)} d \varphi \right)  d \theta \\ 
&\quad = - 2i\Psi_3 \left( \cos (\nu) \be_1 + \sin (\nu)\be_2 \right) \int_0^\tmax  (\cos(\theta))^{3/2} \, \frac{1-\cos(2\theta)}{2} \e^{-\i \kappa  \rpsi\cos(\theta)} J_1\left( \fobj \xi \sin(\theta)\right)   d \theta \\
&\quad =  - 2i\Psi_3 \left( \cos (\nu) \be_1 + \sin (\nu)\be_2 \right) I_{2,1}(\xi,\rpsi) ,
 \end{aligned}
\end{equation} where in the second equality we use \autoref{eq:Besselidentities} for $m=1$, and in the last equality we use the integral $I_{p,q}(\xi,\rpsi)$ as in \autoref{eq:intBessels}.

Using \autoref{eq:firstexp}, \autoref{eq:secexp}, and \autoref{eq:lastexp}, the expression of \autoref{eq:EBFP_integrand} becomes
\begin{equation} \label{eq:int_EBFP2}
\begin{aligned}
~ &
\frac{1}{\pi}\int_0^\tmax \int_{0}^{2 \pi}    \biggl( \Psi_3 \cos(\varphi) \sin(\theta) \cos(2\theta)  - \Psi_p   \cos(\varphi) \cos(\theta)   -    \Psi_s  \sin(\varphi) \biggr) \cdot \\
&\qquad \qquad \qquad \cdot \e^{-\i \xi \rbfp \tan(\theta)\cos (\varphi - \nu)} 
   \e^{\i \kappa (\rbfp\sec^2(\theta) -  r_3^\bpsi)\cos(\theta)} \sqrt{\cos(\theta)} \sec^2(\theta) \tan(\theta)d \varphi d \theta \\    
= & - (\Psi_1 \be_1 + \Psi_2 \be_2) [ I_{1,0}(\xi,\rpsi) + I_{2,0}(\xi,\rpsi) ] \\
&\quad + (\Psi_1 \be_1 + \Psi_2 \be_2) \cos (2\nu)[ I_{1,2}(\xi,\rpsi) + I_{2,2}(\xi,\rpsi) ] \\
&\quad + (\Psi_2 \be_1 + \Psi_1 \be_2) \sin (2\nu)[ I_{1,2}(\xi,\rpsi) + I_{2,2}(\xi,\rpsi) ] \\
&\quad -2  \i \Psi_3 (\cos(\nu) \be_1 +\sin \be_2) I_{2,1}(\xi), 
 \end{aligned}  
\end{equation} 
where the integrals $I_{p,q}(\xi)$ are defined in \autoref{eq:intBessels}.

\end{proof}


\subsection{Electric field approximation in the lens} \label{ss:tube}
After the light ray has passed through the objective and the back focal plane, a tube lens is placed to focus the light rays onto the image plane.

\begin{definition}[Tube lens parameters]
For the tube lens, we assume that it is a \emph{converging lens} with focal length $\flens>0$, which is placed starting at 
$r_3^{\texttt{tl}_i}$ and $r_3^{\texttt{tl}_o}$. Moreover the lens has a thickness which is measured orthogonal to $\be_3$ by the function 
$\ttd$.
\end{definition}

The incoming field at the tube lens $\widehat{\bE}^{\texttt{tl}_i}$ (as defined in \autoref{eq:EFbetween}) and the outgoing wave field 
$\widehat{\bE}_{\texttt{o}}^{\texttt{tl}_o}$ immediately after the lens aperture are related by (we use the same polar coordinates $(\rho,\varphi)$ in both planes $\{r_3 = r_3^{\texttt{tl}_i}\}$ and $\{r_3 = r_3^{\texttt{tl}_o}\}$ )
\begin{align}\label{eq:Etub}
 \widehat{\bE}^{\texttt{tl}_o}(\rho, \varphi) = \e^{\i \mu (\rho)} P_\L(\rho) \widehat{\bE}^{\texttt{tl}_i}(\rho,\varphi), 
\end{align} 
where $P_\L$ is the pupil function associated with the tube lens as in \autoref{eq:pupil}, $\mu$ is the phase shift experienced by the 
field through the tube lens (note that it does not depend on $\varphi$):
\begin{equation} \label{eq:phase_shift}
 \mu (\rho) = \underbrace{\kappa \refractive_l \ttd(\rho)}_{\text{phase delay by lens}} + \underbrace{\kappa (\ttd_0 -\ttd(\rho))}_{\text{phase delay by vacuum}},
\end{equation} 
where $\ttd_0$ is the maximum thickness of the lens, $\ttd(\rho)$ is the thickness of the lens at distance $\rho$ from the optical axis, $\refractive_l$ is the 
refractive index of the lens material, and $\kappa$ as defined in \autoref{eq:kappa}.
The phase delay induced by the lens, under the assumption of a \emph{paraxial approximation} reads as follows (see \autoref{tab:notation} for the 
summary of all physical parameters below):
\begin{equation}\label{eq:phaseparaxial}
 \begin{aligned}
 \mu (\rho) &  \approx \kappa \refractive_l \ttd_0 - \frac{\kappa}{2 \flens} \rho^2 \text{ for all } \rho \in \R.
\end{aligned}
\end{equation}

\subsection{Electric field Approximation in the Image Plane} \label{ss:final}
\begin{definition}
The \emph{electric field at the focal plane},
\begin{equation} \label{eq:I}
 \mathcal{I} := \set{ (\bx_\focus,r_3^\focus) : \bx_\focus \in \R^2},
\end{equation} 
is denoted by $\widehat{\bE}^\focus (\bx_\focus): = \widehat{\bE}(\bx_\focus,r_3^\focus)$, where $\widehat{\bE}$ solves 
the boundary value problem 
\begin{equation} \label{eq:bvpz}
\Delta \widehat{\bE}(\bx,r_3) + \kappa^2(\bx,r_3) \widehat{\bE}(\bx,r_3) = 0
\text{ for all } \bx \in \R^2, \; r_3^{\ttd_o} < r_3 < r_3^\focus 
\end{equation}
with boundary data 
\begin{equation}\label{eq:bcz}
 \widehat{\bE}(\bx,r_3^{\texttt{tl}_o}) = \widehat{\bE}^{\texttt{tl}_o}(\bx) \text{ for all } \bx \in \R^2.
\end{equation}
Note that $\widehat{\bE}^{\texttt{tl}_o}$ as defined in \autoref{eq:EBFP_pre} is already an approximation of the electric field outside of the lens system.
\end{definition}

Following \cite[Eqs 5-14]{Goo05}, we can calculate the field $\widehat{\bE}$ in the \emph{image plane}.
\begin{lemma}\label{lem:Huygens}
At a point $\bx_\focus$ in the image plane,
\begin{equation} \label{eq:Huygens}
\boxed{
\begin{aligned}
 \widehat{\bE}^\focus(\bx_\focus) &=  
 \frac{1}{\i \lambda \flens} \e^{\i \frac{2\pi}{\lambda}(\flens + \refractive_l \ttd_0)} \e^{\i \frac{\pi}{\lambda \flens}\abs{\bx_\focus}^2}
 \int_{\bx \in \R^2}  P_\L(\bx) \widehat{\bE}^{\texttt{tl}_i}(\bx) \e^{-\i \frac{2 \pi}{\lambda \flens}  \inner{\bx_\focus}{\bx}} d \bx
   \text{ for all } \bx_\focus \in \R^2.
 \end{aligned}}
\end{equation} 
\end{lemma}
\begin{proof}
We apply the Huygens-Fresnel principle (see \cite[Eqs 4-17]{Goo05}) to compute the field in the \emph{image plane}:
\begin{align*}
 \widehat{\bE}^\focus(\bx_\focus) &= 
 \frac{1}{\i \lambda \ttd} \e^{\i \kappa  \ttd} \e^{\frac{\i \kappa }{2\ttd}\abs{\bx_\focus}^2} \int_{\bx  \in \R^2} \widehat{\bE}^{\texttt{tl}_o}(\bx )  \e^{\i \frac{\kappa }{2\ttd}\abs{\bx }^2} \e^{-\i \frac{ \kappa }{\ttd}  \inner{\bx_\focus}{\bx }} d \bx  \\
 &= \frac{1}{\i \lambda \ttd} \e^{\i \kappa  \ttd} \e^{\frac{i\kappa }{2\ttd}\abs{\bx_\focus}^2} 
    \int_{\bx  \in \R^2}  P_\L(\bx ) \widehat{\bE}^{\texttt{tl}_i}(\bx ) \e^{\i \mu (\bx )} \e^{\i \frac{\kappa }{2\ttd}\abs{\bx }^2} \e^{-\frac{\i \kappa }{\ttd}  \inner{\bx_\focus}{\bx }} d \bx  \\
 &= \frac{1}{\i \lambda \ttd} \e^{\i \kappa  \ttd} \e^{\frac{\i \kappa }{2\ttd}\abs{\bx_\focus}^2} \int_{\bx  \in \R^2}  P_\L(\bx ) \widehat{\bE}^{\texttt{tl}_i}(\bx ) \e^{\i \kappa  \refractive_l \ttd_0} \e^{-\i \frac{\kappa }{2f} \abs{\bx }^2} \e^{\i \frac{\kappa }{2\ttd} \abs{\bx }^2} \e^{-\i \frac{\kappa }{\ttd}  \inner{\bx_\focus}{\bx }} d \bx  \\
 &= \frac{1}{\i \lambda \ttd} \e^{\i \frac{2\pi}{\lambda}(\ttd + \refractive_l \ttd_0)} \e^{\i \frac{\pi}{\lambda \ttd }\abs{\bx_\focus}^2} \int_{\bx  \in \R^2}  P_\L(\bx ) \widehat{\bE}^{\texttt{tl}_i}(\bx ) \e^{-\i \frac{\pi}{\lambda f} (1-\frac{f}{\ttd}) \abs{\bx }^2} \e^{-\i \frac{2\pi}{\lambda \ttd}  \inner{\bx_\focus}{\bx }} d \bx ,
\end{align*}  
where in the second equality we use \autoref{eq:Etub}, in the third equality we use the paraxial approximation \autoref{eq:phaseparaxial}, and in the last equality we use $\kappa  = \frac{2\pi}{\lambda}$.
Now, if the image plane is at the distance $\ttd=\flens$, the quadratic phase factor term within the integrand exactly cancel, leaving 
\begin{align*}
 \boxed{\widehat{\bE}^\focus(\bx_\focus) =  \frac{1}{\i \lambda \flens} \e^{\i \frac{2\pi}{\lambda}(\flens + \refractive_l \ttd_0)} \e^{\i \frac{\pi}{\lambda \flens}\abs{\bx_\focus}^2} \int_{\bx  \in \R^2}  P_\L(\bx ) \widehat{\bE}^{\texttt{tl}_i}(\bx ) \e^{-\i \frac{2 \pi}{\lambda \flens}  \inner{\bx_\focus}{\bx }} d \bx ,}
\end{align*} where the term $e^{\i \frac{2\pi}{\lambda}(\flens + \refractive_l \ttd_0)}$ is a constant amplitude, and the term $\e^{\i \frac{\pi}{\lambda \flens}\abs{\bx_\focus}^2}$ describes a spherical phase curvature in the focal plane.
\end{proof}
\begin{remark}
From \autoref{eq:Huygens} it follows by the convolution theorem for the Fourier transform that
\begin{equation} \label{eq:Huygens_one}
\boxed{
\begin{aligned}
 \widehat{\bE}^\focus(\bx_\focus) &=
 \frac{2\pi}{\lambda \flens} \e^{\i \pi \left( -\frac{1}{2} + \frac{2}{\lambda}(\flens + \refractive_l \ttd_0) + \frac{1}{\lambda \flens}\abs{\bx_\focus}^2 \right)}
 \mathcal{F}_{12}[P_\L \widehat{\bE}^{\texttt{tl}_i}]\left( 2\pi \frac{\bx_\focus}{\lambda \flens} \right)\\
 &= C_\focus \Phi_\focus  
 \left( \mathcal{F}_{12}[P_\L] \ast \mathcal{F}_{12}[\widehat{\bE}^{\texttt{tl}_i}]\right) \left( \frac{2\pi}{\lambda \flens} \bx_\focus\right)
   \text{ for all } \bx_\focus \in \R^2,
 \end{aligned}}
\end{equation} 
where 
\begin{equation}\label{eq:phase}
 C_\focus = \frac{1}{\lambda \flens} \text{ and } 
 \Phi_\focus = \e^{\i \pi \left( -\frac{1}{2} + \frac{2}{\lambda}(\flens + \refractive_l \ttd_0) + \frac{1}{\lambda \flens}\abs{\bx_\focus}^2 \right)}.
\end{equation}
\end{remark}

In the next step we calculate the Fourier-transform of a circular pupil function:
\begin{lemma}\label{lem:Fourier_pupil}
 Let $P:\R^2 \to \R$ be the \emph{circular pupil function} with radius $R$, as defined in \autoref{eq:pupil}, then 
\begin{equation} \label{eq:Fourier_pupil_polar}
  \mathcal F_{12}[P] (\bv) = R^2\frac{J_1(R\abs{\bv})}{R\abs{\bv}} \text{ for all } \bv \in \R^2.
\end{equation} 
\end{lemma}
\begin{proof} 
We use the following polar coordinates: 
\begin{equation*} 
\bv  = \varkappa \begin{pmatrix}
       \cos(\vartheta)\\ \sin(\vartheta)
      \end{pmatrix}
      \text{ and } \by = \varrho\begin{pmatrix}
        \cos(\sigma)\\
        \sin(\sigma)
      \end{pmatrix}, 
\end{equation*} 
then the Fourier transform of the pupil function 
\begin{equation} \label{eq:Fourier_pupil}
  \begin{aligned}
  \mathcal F_{12}[P] (\bv) &= \frac{1}{2\pi} \int_{\by \in \R^2} P(\by) \e^{-\i\by \cdot \bv} d \by = \frac{1}{2\pi}  \int_{0}^{R} \varrho \int_{0}^{2\pi} \e^{-\i \varkappa \varrho \cos(\sigma-\vartheta)}   d\sigma d \varrho \\
  &= \int_{0}^{R} \varrho J_0(\varkappa \varrho) d \varrho = \frac{R}{\varkappa} J_1(R \varkappa) = R^2\frac{J_1(R\abs{\bv})}{R\abs{\bv}},
 \end{aligned}
  \end{equation} where we use \autoref{eq:Besselidentities} for $m=0$ in the third equality.
\end{proof}

Now, we calculate the $k$-transform of $\widehat{\bE}^\bfp$ in this approximation.

\subsection{Small Aperature}
\label{Sec:sa}
In what follows, we are interested in a small aperture. Of course small aperture have the problem that only little light of the emitted 
dipols passes through the lens and thus these considerations are more of theoretical nature.

In case the numerical aperture $\NA$ is small, that is if $\tmax$ is small, it follows from \autoref{eq:EBFP_pre} that
\begin{equation} \label{eq:EBFP_sa} \widehat{\bE}_{\mathrm{small}}^\bfp(\rho,\varphi) = C(\fobj) \begin{pmatrix} \Psi_1 \\ \Psi_2 \\ 0 \end{pmatrix} \chi_{\rho \ls \NA}.\end{equation}
In what follows, we denote by $P_\NA$ the pupil function $\chi_{\rho \ls \NA}$. We emphasize that the left hand side of \autoref{eq:EBFP_sa} is actually an approximation of the right hand side of \autoref{eq:EBFP_pre}.
Next we calculate the $k$-transform of $\widehat{\bE}^\bfp$  in this approximation. Noting that the objective acts again as a pupil function with disc 
radius $\NA$ we get analogously to \autoref{lem:Fourier_pupil} and by using \autoref{eq:Fourier_EBFP_polar}
\begin{equation} \label{eq:Fourier_EBFP_polar_III}
\begin{aligned}
   \mathcal F_{12}[\widehat{\bE}_{\mathrm{small}}^{\texttt{tl}_i}] (\xi,\nu)  
   = C(\texttt{f}_\obj) 
   \e^{\i(r_3^{\texttt{tl}_i} -r_3^\bfp) \sqrt{\kappa^2-\xi^2}} 
  \begin{pmatrix} \psi_1 \\ \psi_2 \\ 0 \end{pmatrix}  \mathcal F_{12}[P_\NA](\xi,\nu), \quad  \xi^2 \leq \kappa^2, \nu \in [0,2\pi).
 \end{aligned}
\end{equation}

 
 \section{single molecule localization microscopy Experiments and Inverse Problems} \label{s:ip}
 We consider two experiments in two different settings: 
 \begin{experiment}\label{exp:two}  
  \begin{description} 
   \item{For \emph{setting $1$} we assume $n$ static emitting dipoles}: 
     Several monochromatic plane waves of the same frequency $\omega_{inc}$ but with different orientations 
     $\bv^{(j)}$, $j=1,2,\ldots,M$ with $M > 1$ are used to illuminate the cell.
     Every emitting dipole emits light in an orientation $\bpsi^{(j)}(k)$, $k=1,2,\ldots,n$, $j=1,2,\ldots,M$ 
     (depending on the incident field according to \autoref{eq:plane_wave}). 

    Therefore, for each experiment $j=1,2,\ldots,M$ the current 
    $$
     \fouriercurrent^{(j)}(\br) = -\i \omega \sum_{k=1}^n \Psi^{(j)}(k) \delta(\br-\br^{\bpsi(k)})
    $$ 
    and the density 
    $$\widehat{\rho}^{(j)} (\br) = \sum_{k=1}^n \abs{\Psi^{(j)}(k)} \delta'(\br-\br^{\bpsi(k)}) = \sum_{k=1}^n \delta'(\br-\br^{\bpsi(k)})$$ 
    are given via \autoref{eq:J_formula}, \autoref{eq:dens} as the superposition of all dipoles. Note, that we assume that all dipols are unit vectors.
    
    Consequently, the electric field $\widehat{\bE}:=\widehat{\bE}^{(j)}$ solves \autoref{modeleq1}, 
    \begin{equation*} 
     \begin{aligned}
      \Delta_\br \widehat\bE(\br;\omega)+ \kappa^2(\omega) \widehat\bE(\br;\omega) 
       & = \frac{\i \omega}{\epsilon_0 c^2} \fouriercurrent^{(j)}(\br) + \frac{1}{\epsilon_0} \nabla_\br \widehat{\rho}^{(j)} (\br) 
       \text{ for all } \br \in \Omega.
     \end{aligned}
    \end{equation*} 
 \begin{enumerate}
  \item The measurements recorded in a \emph{static experiment} are the energies of the electric field in the image plane, after 
    the light has passed through the imaging system. That is, for each experiment $j=1,2,\ldots,M$ the data 
  \begin{equation} \label{eq:measurement_c}
   \boxed{
    m^{(j)}_i(\bx_\focus;t) = \abs{(\bE^\focus_i)^{(j)}}^2(\bx_\focus;t) \text{ for all } \bx_\focus \in \R^2, t >0 \text{ and for } i=1,2}
  \end{equation}
  are recorded. 
  \item In the \emph{dynamic experiment} setting the static experiment is repeated. We denote the experiment repetitions with 
  the parameter $s$: This experimental setup is used in practice because it makes use of blinking dyes, which allows for 
  better localization of the dyes, and thus molecules.
  That is, the measurements are  
   \begin{equation} \label{eq:measurement_d}
    \boxed{
    m^{(j)}(\bx_\focus;t;s) = \abs{(\bE^\focus_i)^{(j)}}^2(\bx_\focus;t;s) \text{ for all } \bx_\focus \in \R^2,\;t,s > 0 \text{ for } i=1,2.}
  \end{equation}
 \end{enumerate}
\item{For \emph{setting $2$} we assume rotating dipoles}, which are modelled via \autoref{eq:J_wobbel} - here in particular we assume that the cone 
     becomes the ball. That is, $\theta= \pi$. 
    Note that in this case several monochromatic excitations do not provide an asset, so we can constrain ourselves to the case $M=1$.
    As a consequence of \autoref{eq:J_wobbel}, the emitted currents of all rotating dipoles of an single molecule localization microscopy experiment are given by 
    \begin{equation} \label{eq:special_source}
     \fouriercurrent(\br;\omega)  = -\i \omega  \sum_{k=1}^n \mathbf{I}_k a_k \mathds{1}_k(\br),
    \end{equation}
   and according to \autoref{modeleq1} the electric field satisfies the equation
   \begin{equation} \label{modeleq1a}  
   \begin{aligned}
    \Delta_\br \widehat\bE(\br;\omega)+ \kappa^2(\omega) \widehat\bE(\br;\omega) 
     & = \sum_{m} \mathbf{I}_m a_m \widehat{\bf R}_m(\br;\omega) \text{ for all } \br \in \Omega.
   \end{aligned}
  \end{equation} 
 \begin{enumerate}
  \item The measurements recorded in a single molecule localization microscopy are the energies of the electric field in the focal plane (see \autoref{eq:Huygens}), after 
    the light has passed through the imaging system. That is, the data 
  \begin{equation} \label{eq:measurement_c2}
   \boxed{
    m(\bx_\focus;t) = \abs{\bE^\focus}^2(\bx_\focus;t) \text{ for all } \bx_\focus \in \R^2, \quad t > 0}
  \end{equation}
  are recorded. 
  \item In the \emph{dynamic experiment} setting the static experiment is repeated, and we denote every repetition experiment with 
  the parameter $s$: 
  That is the measurements are  
   \begin{equation} \label{eq:measurement_d2}
    \boxed{
    m(\bx_\focus;t;s) = \abs{\bE^\focus}^2(\bx_\focus;t;s) \text{ for all } \bx_\focus \in \R^2, \quad t,s > 0.}
  \end{equation}
 \end{enumerate} 
 Note the difference between setting 1 and 2. In the former it is much easier to identify dipoles because the orientation can be resolved and is not changing over time.
 \end{description}
 \end{experiment}

\subsection{The limit} \label{ss:limit}
Using the small aperture limit of \autoref{Sec:sa} in combination with the formula for the electric field on the image plane \autoref{eq:Huygens_one}, we can compute the electric field in the image plane from \autoref{eq:Fourier_EBFP_polar_III}.  We first make use of a linear approximation of the function $\rho \in [0,\kappa]  \to \sqrt{\kappa^2-\rho^2} \simeq \kappa$ -- that is we assume that between the back focal plane of the objective and the lens, the electric field only undergoes a phase shift which does not depend on the distance to the optical axis.
It follows then from \autoref{eq:Fourier_EBFP_polar_III} that
\begin{equation} \label{eq:Fourier_EBFP_polar_IV}
\mathcal F_{12}[\widehat{\bE}_{\mathrm{small}}^{\texttt{tl}_i}] (\xi,\nu) \simeq C(\focus_\obj)  
\e^{\i(r_3^{\texttt{tl}_i} -r_3^\bfp) \kappa} \begin{pmatrix} \psi_1 \\ \psi_2 \\ 0 \end{pmatrix} \mathcal F_{12}[ P_\NA] (\xi,\nu).
\end{equation}
Applying \autoref{eq:Huygens_one} where we replace $\mathcal F_{12}[\widehat{\bE}^{\texttt{tl}_i}]$ by 
$\mathcal F_{12}[\widehat{\bE}_{\mathrm{small}}^{\texttt{tl}_i}]$ 
and inserting \autoref{eq:Fourier_EBFP_polar_IV}, we obtain
\begin{equation*}
 \begin{aligned}
 \widehat{\bE}_{\mathrm{small}}^\focus(\bx_\focus) & = C_\focus \Phi_\focus  
 \left( \mathcal{F}_{12}[P_\L] \ast \mathcal{F}_{12}[\widehat{\bE}_{\mathrm{small}}^{\texttt{tl}_i}]\right) \left( \frac{2\pi}{\lambda \focus} \bx_\focus\right) \\
&= 
  \begin{pmatrix} \psi_1 \\ \psi_2 \\ 0 \end{pmatrix} \left(\mathcal F_{12}[ P_\L] \ast \mathcal F_{12}[P_\NA]\right) (\xi,\nu)\\
&= 
  \begin{pmatrix} \psi_1 \\ \psi_2 \\ 0 \end{pmatrix} \mathcal F_{12}[ P_\L P_\NA] (\xi,\nu).
 \end{aligned}
\end{equation*}

Now, assuming that $R \gs \NA$ (the lens is bigger than the objective), we have $P_\L P_\NA = P_\NA$ and \autoref{eq:Fourier_pupil} provides,  
\[\widehat{\bE}_{\mathrm{small}}^\focus(\bx_\focus)  = C(\focus_\obj) C_\focus \Phi_\focus  \e^{\i(r_3^{\texttt{tl}_i} -r_3^\bfp) \kappa }
    \begin{pmatrix} \psi_1 \\ \psi_2\\ 0 \end{pmatrix}\NA^2\frac{J_1(\NA \frac{2\pi}{\lambda \flens} \abs{\bx_\focus})}{\NA \frac{2\pi}{\lambda \flens} \abs{\bx_\focus}}   =: \Gamma(\omega)  \begin{pmatrix} \psi_1 \\ \psi_2\\ 0 \end{pmatrix}\NA^2\frac{J_1(\NA \frac{2\pi}{\lambda \flens} \abs{\bx_\focus})}{\NA \frac{2\pi}{\lambda \flens} \abs{\bx_\focus}} .\]
   
Finally, what is actually measured in experiments is the intensity $m(\bx_\focus,t) = |\bE^\focus(\bx_\focus,t)|^2$ 
averaged in time. If we assume that the signal $\bE$ is compactly supported in time and that what is measured by the detector contains this whole support, we can use that the time Fourier transform is a unitary operator and write
\begin{equation} \label{eq:Huygens_two}
 \begin{aligned}
  \bar{m}(\bx_\focus) &= \int_{t \in \R} |\bE^\focus(\bx_\focus)|^2 \dd t \\
  &= \int_{\omega \in \R} |\widehat{\bE}^\focus(\bx_\focus)|^2 \dd \omega \\
  &= \int_{\omega \in \R} \abs{\Gamma(\omega)}^2 \dd \omega 
     \left( \NA^2\frac{J_1(\NA \frac{2\pi}{\lambda \flens} 
     \abs{\bx_\focus})}{\NA \frac{2\pi}{\lambda \flens} \abs{\bx_\focus}} \right)^2  \left( \psi_1^2 + \psi_2^2\right),
     \end{aligned}
\end{equation}
which means that under these approximations (small aperture), what is observed is an \emph{Airy pattern} whose intensity depends on the optical system but is also proportional to squared norm of the component of the dipole which lies orthogonally to the optical axis. Making use of the notation introduced before for the dipole orientation, we obtain that \emph{the measured intensity $m$ is proportional to $\cos^2(\theta_m)$.}
Since the Airy function can be well approximated by a Gaussian function, the measured signals of the emitted dyes very much looks like superposition of Gaussian functions.

Now, having summarized the mathematical modeling of single molecule localization microscopy we can state the associated inverse problems:

\begin{definition}[Inverse Problem]
The inverse problem of single molecule localization microscopy in the two different settings
environment consists in calculating 
\begin{equation*}
 (a_m, \Psi_m, {\bf r}_m, \chi_m(s))_{m=1,\ldots,M},
 \text{ from the measurements } m. 
\end{equation*}
Indeed the inverse problem could also be generalized to reconstruct \autoref{modeleq1a} in addition, which would result in an 
\emph{inverse scattering problem} \cite{CakCol14}. However, this complex problem is not considered here further.
\end{definition}

In current practice of single molecule localization microscopy the simplified formulas \autoref{eq:Huygens_two} 
are used for reconstruction of the center of gravities $(r_m^1,r_m^2)$ in the measurement data $m$ induced by the point spread function $\psf_\omega$. 

\section*{Conclusion}
The main objective of this work was to model mathematically the propagation of light emitted from dyes 
in a superresolution imaging experiment. We formulated basic inverse problems related to single molecule 
superresolution microscopy, with the goal to have a basis for computational and quantitative single molecule 
superresolution imaging. The derivation of the according equations follows the physical and chemical 
literature of supperresolution microscopy, in particular \cite{HelAxe87, Axe12}, which are combined with the 
mathematical theory of distributions to translate physical and chemical terminologies into a mathematical framework. 

\appendix

\section{Derivation of Particular Fourier transforms}
\label{app:FT}

Before reading through the appendix it might be useful to recall the notation of \autoref{ta:abb}.
The derivation in \autoref{lemmaE} uses the \emph{residual theorem}.
\begin{theorem}[Residual Theorem] \label{th:residual}
  Let the function $z \in \C \to \tilde{f}(z):=f(z) \e^{\i a z}$ with $a > 0$ satisfy the following properties:
 \begin{enumerate}
  \item $f$ is analytic with at most finitely many poles $p_i$, $i=1,\ldots,m$, which do not lie on the real axis.
  \item There exists $M,R>0$ such that for every $z \in \C$ satisfying $\Im(z) \geq 0$ and $\abs{z} \geq R$
        \begin{equation} \label{eq:growth}
         \abs{f(z)} \leq \frac{M}{\abs{z}}.
        \end{equation}
 \end{enumerate}
 Then 
 \begin{equation*}
  \int_\R \tilde{f}(x) dx = 2 \pi \i \sum_{i=1}^m  \text{Res}(\tilde{f};p_i).
 \end{equation*}
\end{theorem}
Using this lemma we are able to prove the following result used in \autoref{lemmaE}:
\begin{lemma} 
Let the assumptions and notation of \autoref{lemmaE} hold (in particular this means that $r_3, r_3^\bpsi \in \R$ 
satisfy $r_3-r_3^\bpsi > 0$), then 
 \begin{equation} \label{eq:FTinz}
\frac{1}{\sqrt{2\pi}}
\mathcal{F}_3^{-1} 
\left[k_3 \to  \left(\bpsi  +   \frac{(\bpsi \times \bk) \times \bk}{\abs{\bk}^2 - \kappa_\ve^2} \right)\e^{-\i k_3 r_3^\bpsi}\right](r_3) =
 \psi_3 \be_3\eindelta(r_3 - r_3^\bpsi) + \frac{\i \e^{\i q_{\ve}  (r_3-r_3^\bpsi)}}{ 2q_{\ve}  } (\bpsi \times \bk_{q_{\ve}}) \times \bk_{q_{\ve}}.
\end{equation}
\end{lemma}

First, we note that 
\begin{equation*}
 \Psi +  \frac{(\bpsi \times \bk) \times \bk}{k_3^2 - q_{\ve}^2} = 
 \left(\Psi - \psi_3 \be_3 + \frac{(\bpsi \times \bk) \times \bk}{ k_3^2 - q_{\ve}^2} \right) +
 \psi_3 \be_3.
\end{equation*}
Now, we calculate the Fourier-transform of each of the two terms on the right hand side. 
The second term can be calculated from \autoref{eq:dist3} and is given by 
\begin{align}
 \mathcal F^{-1}_3[k_3 \mapsto \e^{-\i k_3 r_3^\bpsi}](r_3) &=  \sqrt{2\pi} \eindelta(r_3 - r_3^\bpsi) \label{eq:FT1}.
\end{align} 
For the calculation of the first term we use the residual theorem:

\begin{itemize}
 \item Clearly $f$ is analytic with potential poles at $k_3 = \pm q_{\ve}$.
 \item To verify \autoref{eq:growth} we use the elementary calculation rules for $\Curl$, summarized in \autoref{eq:doubleprod} and get
\begin{align*}
~ & \frac{(\bpsi \times \bk) \times \bk}{ k_3^2 - q_{\ve}^2} \\
= &-\frac{1}{ k_3^2 - q_{\ve}^2} \left(k_3^2 (\bpsi - \psi_3 \be_3) 
   - k_3 \begin{pmatrix} \psi_3 k_1 \\ \psi_3 k_2 \\ \psi_1 k_1 + \psi_2 k_2\end{pmatrix}   + (k_1^2+k_2^2) \bpsi - 
         \begin{pmatrix} (\psi_1 k_1 + \psi_2 k_2)k_1 \\ (\psi_1 k_1 + \psi_2 k_2)k_2 \\0 \end{pmatrix} \right) \\
= &-\frac{1}{k_3^2 - q_{\ve}^2} \left(k_3^2 \left( \bpsi - \psi_3 \be_3\right) - k_3 \begin{pmatrix} \psi_3 k_1 \\ \psi_3 k_2 \\ \psi_1 k_1 + \psi_2   
    k_2\end{pmatrix}   + (k_1^2+k_2^2) \left( \bpsi - \psi_3 \be_3\right) + 
    \begin{pmatrix} -\psi_1 k_1^2 - \psi_2 k_2k_1 \\ -\psi_2 k_2^2 - \psi_1 k_1k_2 \\(k_1^2+k_2^2) \psi_3 \end{pmatrix}     \right) \\
= &-\frac{1}{ k_3^2 - q_{\ve}^2} \left((k_3^2-q_{\ve}^2) \left( \bpsi - \psi_3 \be_3\right) - k_3 \begin{pmatrix} \psi_3 k_1 \\ \psi_3 k_2 \\ \psi_1 k_1 + 
  \psi_2 k_2\end{pmatrix}   + \kappa^2 \left( \bpsi - \psi_3 \be_3\right) + 
  \begin{pmatrix} -\psi_1 k_1^2 - \psi_2 k_2k_1 \\ -\psi_2 k_2^2 - \psi_1 k_1k_2 \\(k_1^2+k_2^2) \psi_3 \end{pmatrix}     \right) \\
= & - \bpsi + \psi_3 \be_3 +  \frac{k_3}{k_3^2 - q_{\ve}^2} {\bf a}
  - \frac{1}{k_3^2 - q_{\ve}^2}  {\bf b}
\end{align*}
where
\begin{equation*}
 {\bf a} = \begin{pmatrix} \psi_3 k_1 \\ \psi_3 k_2 \\ \psi_1 k_1 + \psi_2 k_2\end{pmatrix} \text{ and }
 {\bf b} = \begin{pmatrix} (\kappa^2 - k_1^2) \psi_1 -\psi_2 k_1 k_2 \\ 
  (\kappa^2 - k_2^2)  \psi_2 - \psi_1 k_1 k_2 \\ (k_1^2 + k_2^2)\psi_3 \end{pmatrix}.
\end{equation*}
Using that 
\begin{equation*}
 \frac{k_3}{k_3^2-q_{\ve}^2} = \frac{1}{2} \left( \frac{1}{k_3-q_{\ve}} + \frac{1}{k_3+q_{\ve}} \right) \text{ and }
 \frac{1}{k_3^2-q_{\ve}^2} = \frac{1}{2q_{\ve}} \left( \frac{1}{k_3-q_{\ve}} - \frac{1}{k_3+q_{\ve}} \right)
\end{equation*}
we get
\begin{equation}\label{eq:residual}
\frac{(\bpsi \times \bk) \times \bk}{ k_3^2 - q_{\ve}^2} + \bpsi - \psi_3 \be_3  
=  \frac{1}{k_3-q_{\ve}} \left( \frac{1}{2} {\bf a} - \frac{1}{2q_{\ve}} {\bf b} \right) 
+ \frac{1}{k_3+q_{\ve}} \left( \frac{1}{2}{\bf a} + \frac{1}{2q_{\ve}} {\bf b} \right).
\end{equation}
This shows \autoref{eq:growth}.
\end{itemize}
Therefore we can apply the residual \autoref{th:residual}: For this purpose let 
 \begin{equation} \label{eq:extension}
 z \in \C \to f(z):= \frac{1}{z^2 - q_{\ve}^2} \left( {\bpsi} \times \begin{pmatrix}k_1\\k_2\\z \end{pmatrix} \right) 
 \times \begin{pmatrix}k_1\\k_2\\z \end{pmatrix} + {\bpsi} - \psi_3 \be_3.
\end{equation}
Then, since $q_{\ve}$ is complex with positive imaginary part, $f$ has only one pole in the upper half plane, and we get
\begin{equation*}
\begin{aligned}
  \int_\R f(k_3) \e^{\i k_3 (r_3 - r_3^\bpsi)} d k_3 &= 2 \pi \i \text{Res} \left(f(z) \e^{\i z (r_3 - r_3^\bpsi)}; z=q_{\ve}\right) \\
 &= 2 \pi \i \left( \frac{1}{2q_{\ve}} \left( {\bpsi} \times \begin{pmatrix}k_1\\k_2\\q_{\ve} \end{pmatrix} \right) 
 \times \begin{pmatrix}k_1\\k_2\\q_{\ve} \end{pmatrix} \e^{\i q_{\ve} (r_3 - r_3^\bpsi)} \right).
\end{aligned}
\end{equation*}
This implies
 \begin{equation*}
\frac{1}{\sqrt{2\pi}}
\mathcal{F}_3^{-1} 
\left[k_3 \to  \left(\bpsi  +   \frac{(\bpsi \times \bk) \times \bk}{\abs{\bk}^2 - \kappa_\ve^2} \right)\e^{-\i k_3 r_3^\bpsi}\right](r_3) =
 \psi_3 \be_3\eindelta(r_3 - r_3^\bpsi) + \frac{\i \e^{\i q_{\ve}  (r_3-r_3^\bpsi)}}{ 2q_{\ve}  } (\bpsi \times \bk_{q_{\ve}}) \times \bk_{q_{\ve}}.
\end{equation*}


\section{Derivation of the far field approximation for the electric field}
\label{app:farField}
\begin{lemma} \label{le:alpha}
Let $\zeta$ be the function defined in \autoref{eq:alpha}. That is, for all $\bk_{12} \in \R^2$ 
$\zeta(\bk_{12}) = \bk_{12} \cdot \bv + q(\bk_{12})$, with $q = \sqrt{\kappa^2 - k_1^2 - k_2^2}$; 
note the formal definition of $q$ in \autoref{eq:qk1} is via the limit $\ve \to 0$.
Then the gradient of $\zeta$ is given by 
\begin{equation} \label{eq:crit}
 \nabla \zeta (\bk_{12}) = \bv- \frac{\bk_{12}}{\sqrt{\kappa^2 - \abs{\bk_{12}}^2}},
\end{equation}
which vanishes for 
\begin{equation} \label{eq:bv}
 \hat{\bk}_{12} := \frac{\kappa}{\sqrt{1+\abs{\bv}^2}}\bv.
\end{equation}
Consequently $\br = r_3 \be_3 + r_3 \begin{pmatrix} \bv & 1 \end{pmatrix}^T$ as in \autoref{eq:v} satisfies 
\begin{equation} \label{eq:bv3}
\frac{\br}{\abs{\br}} = \frac{1}{\sqrt{1+\abs{\bv}^2}} \begin{pmatrix} \bv\\ 1 \end{pmatrix}.
\end{equation}
Moreover, with $q$ as defined in \autoref{eq:qk1}, we get the following identities
\begin{equation} \label{eq:bv2}
\begin{aligned}
  \zeta(\hat{\bk}) = \hat{\bk}_{12} \cdot \bv + q =\kappa \sqrt{1+\abs{\bv}^2} \quad \text{ and } \quad 
  q = q(\hat{\bk}) = \sqrt{ \kappa^2 - \hat{k}_1^2-\hat{k}_2^2} = \frac{\kappa}{\sqrt{1+\abs{\bv}^2}}.
\end{aligned}
\end{equation}
The Hessian of $\zeta$ is given by 
\begin{equation*} 
 H(\zeta)(\bk_{12}) = \begin{pmatrix}
                           - \frac{1}{\sqrt{\kappa^2 - \abs{\bk_{12}}^2}} - \frac{k_1^2}{(\kappa^2 - \abs{\bk_{12}}^2)^{3/2}} & - \frac{k_1 k_2}{(\kappa^2 - \abs{\bk_{12}}^2)^{3/2}} \\- \frac{k_1 k_2}{(\kappa^2 - \abs{\bk_{12}}^2)^{3/2}} & - \frac{1}{\sqrt{\kappa^2 - \abs{\bk_{12}}^2}} - \frac{k_2^2}{(\kappa^2 - \abs{\bk_{12}}^2)^{3/2}}                          
                          \end{pmatrix}, \end{equation*}
which evaluated at $\hat{\bk}$ gives
\begin{align*}H(\zeta)(\hat{\bk}) &= \begin{pmatrix}
- \frac{c}{\omega}\sqrt{1+\abs{\bv}^2} - \frac{ c^3}{\omega^3 }(1+\abs{\bv}^2)^{3/2} \cdot \frac{v_1^2 \omega^2}{c^2 (1+ \abs{\bv}^2)} 
&- \frac{ c^3}{\omega^3 }(1+\abs{\bv}^2)^{3/2} \cdot \frac{v_1v_2 \omega^2}{c^2 (1+ \abs{\bv}^2)} \\
- \frac{ c^3}{\omega^3 }(1+\abs{\bv}^2)^{3/2} \cdot \frac{v_1v_2 \omega^2}{c^2 (1+ \abs{\bv}^2)} 
& - \frac{c}{\omega}\sqrt{1+\abs{\bv}^2} - \frac{ c^3}{\omega^3 }(1+\abs{\bv}^2)^{3/2} \cdot \frac{v_2^2 \omega^2}{c^2 (1+ \abs{\bv}^2)}    \end{pmatrix}, \\
& = - \frac{c}{\omega} \sqrt{1+\abs{\bv}^2} \begin{pmatrix}
1 + v_1^2 
& v_1 v_2  
\\v_1 v_2  & 1+v_2^2                          
\end{pmatrix},
\end{align*}
and the determinant satisfies
\begin{equation*} 
= \frac{c^2(1+\abs{\bv}^2)^2}{\omega^2} > 0.
\end{equation*}
\end{lemma}

\section{Bessel identities}
For $x \in \R$ and $\varphi_0 \in [0,2\pi)$ the Bessel-identities hold:
\begin{equation} \label{eq:Besselidentities} 
\begin{aligned}
 2 \pi (-1)^m \i^m J_m(x) \cos (m \varphi_0) &= \int_{0}^{2 \pi} \e^{-\i x \cos (\varphi - \varphi_0)} \cos (m \varphi) d \varphi, \\
 2 \pi (-1)^m \i^m J_m(x) \sin (m \varphi_0) &= \int_{0}^{2 \pi} \e^{-\i x \cos (\varphi - \varphi_0)} \sin (m \varphi) d \varphi,
 \end{aligned}
\end{equation} where $J_m$ is the Bessel function of the first kind of order $m$.

\section*{Acknowledgement} 
OS is supported by the Austrian Science Fund (FWF), within SFB F68 (Tomography across the scales), 
project F6807-N36 (Tomography with Uncertainties) and I3661-N27 (Novel Error Measures and Source Conditions of Regularization 
Methods for Inverse Problems). 
KS is supported by the Austrian Science Fund (FWF), project P31053-N32 (Weighted X-ray transform and applications).
MLM, MS and GS are also supported by the SFB F68, project F6809-N36 (Ultra-high Resolution Microscopy). JS was supported in part by the NSF grant DMS-1912821 and the 
AFOSR grant FA9550-19-1-0320.

\bibliography{lm5st18_biber}
\bibliographystyle{plain}

\end{document}